\documentclass[reqno,12pt,letterpaper]{article}
\pdfoutput=1
\usepackage{color, xcolor, colortbl}
\usepackage{graphicx,epstopdf}
\usepackage{geometry}
\usepackage{enumerate}
\usepackage{amsmath,amssymb,amsthm}
\usepackage{algorithm}
\usepackage{algorithmic}
\usepackage{caption}
\usepackage{comment}
\usepackage{subcaption}
\usepackage{bm}
\usepackage{appendix}
\usepackage{multirow}
\usepackage{braket}
\usepackage[english]{babel}
\usepackage{hyperref}
\usepackage[capitalize]{cleveref}
\usepackage[sort&compress,square,numbers]{natbib}
\usepackage{adjustbox}
\usepackage{xspace}
\usepackage[roman]{complexity}
\usepackage{soul}
\usepackage{tikz}

\usepackage{rotating}
\usepackage{setspace}
\usepackage{fancyhdr}

\newcommand{\mc}[1]{\mathcal{#1}}

\makeatletter
\newcommand{\opnorm}{\@ifstar\@opnorms\@opnorm}
\newcommand{\@opnorms}[1]{%
  \left|\mkern-1.5mu\left|\mkern-1.5mu\left|
   #1
  \right|\mkern-1.5mu\right|\mkern-1.5mu\right|
}
\newcommand{\@opnorm}[2][]{%
  \mathopen{#1|\mkern-1.5mu#1|\mkern-1.5mu#1|}
  #2
  \mathclose{#1|\mkern-1.5mu#1|\mkern-1.5mu#1|}
}
\makeatother

\usepackage{authblk}

\newcommand{\I}{\mathrm{i}}

\newtheorem{thm}{\protect\theoremname}
\theoremstyle{plain}
\newtheorem{lemma}[thm]{\protect\lemmaname}
\theoremstyle{plain}
\newtheorem{rem}[thm]{\protect\remarkname}
\theoremstyle{plain}
\theoremstyle{plain}
\newtheorem{prop}[thm]{\protect\propositionname}
\theoremstyle{plain}
\newtheorem{cor}[thm]{\protect\corollaryname}
\newtheorem{assumption}[thm]{\protect\assumptionname}
\newtheorem{defn}[thm]{\protect\definitionname}

\providecommand{\definitionname}{Definition}
\providecommand{\assumptionname}{Assumption}
\providecommand{\corollaryname}{Corollary}
\providecommand{\lemmaname}{Lemma}
\providecommand{\propositionname}{Proposition}
\providecommand{\remarkname}{Remark}
\providecommand{\theoremname}{Theorem}

\newcommand\blfootnote[1]{%
  \begingroup
  \renewcommand\thefootnote{}\footnote{#1}%
  \addtocounter{footnote}{-1}%
  \endgroup
}

\newcommand{\norm}[1]{\left\lVert#1\right\rVert}

\newcommand{\REV}[1]{\textcolor{black}{#1}}

    \title{Trotterization with Many-body Coulomb Interactions: Convergence for General Initial Conditions and State-Dependent Improvements}

\author[1,2]{Di Fang}
\author[3]{Xiaoxu Wu}

\affil[1]{Department of Mathematics, Duke University, Durham, NC 27710, USA}
\affil[2]{Duke Quantum Center, Duke University, Durham, NC 27701, USA}
\affil[3]{Mathematical Sciences Institute, Australia National University, Canberra 2601, Australia}

\date{} 

\begin{document}
\maketitle

\begin{abstract}
Efficiently simulating many-body quantum systems with Coulomb interactions is a fundamental question in quantum physics, quantum chemistry, and quantum computing, yet it presents unique challenges: 
the Hamiltonian is an unbounded operator (both kinetic and potential parts are unbounded); its Hilbert space dimension grows exponentially with particle number; and the Coulomb potential is \textit{singular, long-ranged, non-smooth, and unbounded}, violating the regularity assumptions of many prior state-of-the-art many-body simulation analyses.
In this work, we establish rigorous error bounds for Trotter formulas applied to many-body quantum systems with Coulomb interactions. Our first main result shows that for general initial conditions in the domain of the Hamiltonian, second-order Trotter achieves a $1/4$ convergence rate with explicit polynomial dependence of the error prefactor on the particle number. The polynomial dependence on system size suggests that the algorithm remains quantumly efficient, even without introducing any regularization of the Coulomb singularity. Although the result under general conditions constitutes a worst-case bound, this rate has been observed in prior work for the hydrogen ground state, demonstrating its relevance to physically and practically important initial conditions. We further establish exact one-step local error asymptotics of order $5/4$ for both the first-order and second-order Trotter formulas, showing that the local exponent is sharp. Our second main result identifies a set of physically meaningful conditions on the initial state under which the convergence rates improve, with the full first-order and second-order rates recovered at sufficiently high angular momentum. For hydrogenic systems, these conditions are connected to excited states with sufficiently high angular momentum. Our theoretical findings are consistent with prior numerical observations. Our third main result establishes improved convergence rates for many-body fermionic Coulomb systems with initial data of different Sobolev regularity. In particular, even without imposing any additional regularity beyond the natural domain of the Hamiltonian, both the first- and second-order Trotter formulas converge with rate $1/2$, rather than the general $1/4$ rate. 
\blfootnote{Emails:  di.fang@duke.edu; xiaoxu.wu@anu.edu.au.}
\end{abstract}

\tableofcontents

\section{Introduction}

Many-body quantum systems with Coulomb interactions are central to physics, chemistry, and materials science, as they underpin problems ranging from atomic and molecular dynamics to electronic systems. Simulating these systems efficiently on quantum computers has been an important topic in the quantum computing and simulation community. Depending on the spatial discretization scheme, the underlying Hamiltonian admits different circuit encodings, including both first-quantized and second-quantized formulations, e.g., ~\cite{AspuruGuzikDutoiLoveHead-Gordon2005,KassalJordanLoveMohseniAspuru-Guzik2008,BabbushBerryDominicKivlichanWeiLoveAspuru-Guzik2016,BabbushBerrySandersKivlichanSchererWeiLoveAspuru-Guzik2017,SuBerryWiebeEtAl2021,BabbushHugginsBerryUngZhaoEtAlBaczewskiLee2023,RubinBerryKononovEtAlLeeNevenBabbushBaczewski2024,KuChenHuHsieh2025,BabbushWiebeMcCleanMcClainNevenChan2018,KivlichanMcCleanWiebeGidneyAspuru-GuzikChanBabbush2018,StroeksLentermanTerhalHerasymenko2024,LiuZhuLowLinYang2025}, each with its own advantages for simulation and algorithm design.

The unbounded nature of both the Laplacian operator and the Coulomb potential poses significant mathematical and algorithmic difficulties. This makes Trotterization (product formula methods) a particularly natural approach, as it decomposes the time evolution into a sequence of local unitary operations that are more friendly to implement on quantum hardware. Trotter methods remain among the most widely used simulation techniques due to their simplicity, compatibility with unbounded operators, and well-understood error structures~\cite{ChildsSuTranEtAl2020,Somma2015,SahinogluSomma2020,AnFangLin2021,SuHuangCampbell2021,ZhaoZhouShawEtAk2021,ChildsLengEtAl2022,FangTres2023,BornsWeilFang2022,HuangTongFangSu2023,ZengSunJiangZhao2022,GongZhouLi2023,LowSuTongTran2023,ZhaoZhouChilds2024,YuXuZhao2024,ChenXuZhaoYuan2024,FangQu2025,BeckerGalkeSalzmannLuijk2024,MizutaIkeda2025Fujii2025,MizutaKuwahara2025}. Compared with post-Trotter approaches~\cite{LowChuang2017,GilyenSuLowEtAl2019,BerryChildsCleveEtAl2015,KieferovaSchererBerry2019,LowWiebe2019,BerryChildsSuEtAl2020,AnFangLin2022,FangLiuSarkar2025,BornsweilFangZhang2025,FangLiuZhu2025,FangZhang2025} (e.g., truncated series, qubitization, quantum signal processing, and quantum singular value transformation), Trotterization executes entirely through unitary operations and hence avoids reintroducing operator-norm cost dependence in circuit implementations.

Even so, analyzing Trotter error in this setting, without introducing regularization or modifying the Coulomb singularity, remains highly nontrivial. The challenges are threefold: (i) the Hilbert space dimension grows exponentially with particle number; (ii) both kinetic and potential operators are unbounded; and (iii) the Coulomb potential is singular and non-smooth, violating the regularity assumptions used in commonly used many-body Trotter error analyses. In such a many-body analysis, it is important to determine both the best possible convergence rate of the error bound and the explicit dependence of the preconstant on the system size (the particle number).

While the dependence on system size appears in the Trotter error bound as a preconstant, it is important to emphasize that this is not just a constant! The scaling of this prefactor with the particle number is decisive in determining the efficiency of the algorithm in the many-body setting. From a computational complexity perspective, achieving only polynomial dependence on the particle number $N$ is essential.

In previous work~\cite{FangWuSoffer2025}, we rigorously analyzed first-order Trotter error bounds for many-body quantum systems with Coulomb interactions. We proved that first-order Trotter achieves a $1/4$ convergence rate, with a preconstant scaling polynomially as $N^{4.5}$. The rate matches the prior numerical studies~\cite{BurgarthFacchiHahnJohnssonYuasa2024}, such as hydrogen-atom simulations with the ground-state wavefunction as the initial state, confirming that this $1/4$ rate indeed governs the convergence. These results raise two natural and important questions:

\textit{1. What is the convergence rate and system-size dependence for the second-order Trotterization?
2. Can special classes of initial states
improve the convergence rate beyond the \REV{general} worst-case $1/4$ rate?}

This paper addresses both questions and makes \REV{three} main contributions.
Our first contribution is to prove a $1/4$ rate bound for the second-order Trotter formula for all initial conditions in the domain of the Hamiltonian.
We rigorously prove that for many-body Coulomb systems, the second-order Trotter has a worst-case convergence rate of $1/4$, the same as the first-order Trotter formula. \REV{We further establish exact one-step local error asymptotics of order $t^{5/4}$ for both formulas. This shows that the quarter-order exponent cannot be improved uniformly over all numbers of Trotter steps, although a fixed-time many-step lower bound, as the number of steps tends to infinity, remains open.} 
Importantly, the \REV{evidence of the }optimality of this 1/4 rate is supported by numerical results~\cite[Figure 6]{BurgarthFacchiHahnJohnssonYuasa2024}: $1/4$ rate is already observed for the physically most natural case -- the ground state $\Psi_{100}$ of the hydrogen atom -- demonstrating the practical relevance of our worst-case analysis. The rate for the ground state of the one-body case was also first proved in~\cite{BurgarthFacchiHahnJohnssonYuasa2024}. Here we consider general initial conditions \REV{in the domain of the Hamiltonian} and the $N$-body case. To our knowledge, this is the first rigorous proof of a $1/4$ rate for the second-order Trotter formula for all initial conditions in the domain of the Hamiltonian, even for one-body systems. Moreover, we also achieve an explicit polynomial dependence on $N$ in the many-body scenario.

Having established the general-case bounds, we further investigate conditions on the initial state that can lead to improved convergence rates. For systems with Coulomb singularities, we identify a set of physically meaningful \REV{angular-momentum and Sobolev regularity} conditions on the \REV{initial} wavefunction, which govern whether the $1/4$ bottleneck can be overcome. \REV{In the one-body and two-body case, we obtain a hierarchy of improved convergence rates depending on the lowest angular momentum sector of the initial state. When restricted to }the hydrogen atom, \REV{our result implies that }eigenstates with angular momentum \REV{quantum number at least $2$} \REV{recover the full first-order rate}, while \REV{those with angular momentum quantum number at least $4$ recover the full second-order rate}. Our rigorous results match previous numerical studies~\cite{BurgarthFacchiHahnJohnssonYuasa2024} as well as their physical intuition, and provide a unifying mathematical explanation for the observed behaviors.

\REV{Our third contribution establishes improved convergence rates for many-body fermionic Coulomb systems with initial data of different Sobolev regularity. In particular, even without imposing any additional regularity beyond $H^2\REV{(\mathbb{R}^{3N})}$ (the domain of the Hamiltonian), both the first- and second-order Trotter formulas converge with rate $1/2$, rather than the general $1/4$ rate.}

The organization of the rest of the paper is as follows:
In~\cref{sec:main_results}, we introduce the problem setup and notations, and present our main results, including both the $1/4$ convergence rate for general initial conditions \REV{in the domain of the Hamiltonian} and the improved rates under additional structural assumptions. \cref{sec:pf_main_result1}, \cref{sec:pf_main_result2_improved_rate} \REV{and \cref{sec:pf_main3_many-body-fermion}} are devoted to the proofs of the main results. \REV{The local
error lower bounds are established in
\cref{sec:local_error_lower_bounds}.} Finally,~\cref{sec:conclusion} concludes with a discussion of the main findings and directions for future research.

\section{Main Results}\label{sec:main_results}
We introduce the problem and notation in this section, followed by a presentation of the main results.

\subsection{Problem Setup and Notations}
Let $N \in \mathbb{N}^+$ denote the particle number (i.e., system size). We consider the $N$-body Schr\"odinger equation with Coulomb interactions:
\begin{equation}
    \begin{cases}
        i\partial_t \psi(t) = (- \Delta + V(x) ) \psi(t) =: H\psi(t) \\
        \psi(0) = \psi_0 \in  H^2(\mathbb{R}^{3N})
    \end{cases}\qquad\qquad  t\in \mathbb R,\label{N-SE}
\end{equation}
where the spatial degrees of freedom are denoted by
\[
x = (x_1,\ldots,x_N), \qquad x_j \in \mathbb{R}^3,
\]
so that the total spatial dimension is $3N$. The negative Laplacian operator is defined in the standard way by
$ -\Delta := -\sum_{j=1}^N \Delta_{x_j}$, and the Coulomb interaction potential $V(x)$ is given by
\begin{equation}
    V(x) = \sum_{1 \leq j < k \leq N} \frac{c_{jk}}{|x_j - x_k|},\label{eq:N-V_def}
\end{equation}
where the coupling constants $c_{jk} \in \mathbb{R}$, $1 \leq j < k \leq N$, may be either positive or negative, allowing for both repulsive and attractive interactions depending on the application. We assume that the coupling coefficients are uniformly bounded, namely,
\begin{equation} \label{con: c0}
    c_0 := \max_{1 \leq j < k \leq N} |c_{jk}| < \infty.
\end{equation}
Throughout this work, we consider the initial data $\psi_0 \in H^2(\mathbb{R}^{3N})$, which coincides with the domain of the (unbounded) Hamiltonian operator $H$. In other words, $H^2(\mathbb{R}^{3N})$ consists precisely of those states $\psi$ for which the action of the Schrödinger operator on the wavefunction is well defined, i.e., $H\psi \in L^2(\mathbb{R}^{3N})$. \REV{The identification of $H^2(\mathbb{R}^{3N})$ with the operator domain
follows from the classical Kato-Rellich theory. Indeed, the Hardy
inequality implies that the Coulomb potential $V$ is infinitesimally
$-\Delta$-bounded. Therefore, $H=-\Delta+V$ is self-adjoint and bounded
from below with
\begin{equation}
\operatorname{dom}(H)
=
H^2(\mathbb{R}^{3N});
\end{equation}
see, e.g.,
\cite[Theorems~X.12 and X.16]{ReedSimonII1975}. We also distinguish the operator domain
$H^2(\mathbb{R}^{3N})$ from $H^1(\mathbb{R}^{3N})$, which is the quadratic-form, or finite-energy, domain (see, e.g.,
\cite[Theorems~X.17--X.18]{ReedSimonII1975}).}

Throughout the sequel, the notation $\|\cdot\|$ is used to denote either the norm in $L^2 \equiv L^2(\mathbb{R}^n)$ of a wavefunction or the operator norm on $L^2(\mathbb{R}^n)$ of an operator, as determined by the context. When necessary, we write $\|\cdot\|_{\mathcal{H} \to \mathcal{H}}$ for the operator norm on a Hilbert space $\mathcal{H}$, and $\|\cdot\|_{\mathcal{H}_1 \to \mathcal{H}_2}$ for the operator norm of a bounded linear map from one Hilbert space $\mathcal{H}_1$ to another Hilbert space $\mathcal{H}_2$. \REV{Throughout, if $g$ is a function on $\mathbb{R}^n$, we abbreviate its Sobolev norm by
\begin{equation}
\norm{g}_{H^s}
:=
\norm{g}_{H^s(\mathbb{R}^n)},
\end{equation}
for notational convenience, where the dimension $n$ is determined by the underlying configuration space of $g$. This convention also applies when we refer to the $H^s$ norm or $H^s$ regularity in words. Whenever we refer to a Sobolev space itself, we specify the underlying Euclidean space explicitly.} We employ the following convention for the $H^2$ norm: for $g \in H^2\REV{(\mathbb{R}^n)}$,
\begin{equation}\label{eq:def_H^2_norm}
    \|g\|_{H^2} := \sqrt{\|(-\Delta)g\|^2 + \|g\|^2},
\end{equation}
which quantifies the second-order derivative behavior of a quantum state. We note that the setup and the notations are consistent with~\cite{FangWuSoffer2025}. 

We briefly recall the first- and second-order Trotter splitting schemes for the time evolution generated by a Hamiltonian of the form $H = A + B$. 
The first-order (Lie-Trotter) splitting~\cite{Trotter1959} approximates the exact propagator $e^{-iHt}$ by
\[
e^{-iHt} \;\approx\; e^{-iBt}\,e^{-iAt},
\]
while the second-order (Strang) splitting~\cite{Strang1968} is given by
\[
e^{-iHt} \;\approx\; e^{-iA t/2}\,e^{-iBt}\,e^{-iA t/2},
\]
where $t$ is the short Trotter time-step. 
In the present work, we adopt the decomposition 
\begin{equation} \label{def:A_and_B}
    A=-\Delta, \quad B=V(x),
\end{equation} 
corresponding to the kinetic and Coulomb interaction operators.

\subsection{Main Result 1: Trotter 2 for General Initial Conditions}\label{sec:main-result-1}

Our first main result concerns the convergence of the second-order (Strang) Trotter splitting for the many-body Schr\"odinger equation with Coulomb interactions. 
We prove a long-time error bound that remains finite directly in the continuum, without introducing any spatial discretization, and whose dependence on the system size is explicit and polynomial.

\begin{thm}[Long-time 2nd-order Trotter Error for General Initial States] \label{thm:main_trotter_long-main}
Let $H = A + B$ be the $N$-body Hamiltonian with Coulomb interactions given by~\cref{eq:N-V_def,con: c0,def:A_and_B}. \REV{Let $T>0$ and $L\in\mathbb{N}^+$ be such that the short Trotter
step size $t:=T/L$ satisfies $0<t\leq1$ and $t<T^{4/3}$.}
For any initial state $\psi_0 \in H^2(\mathbb{R}^{3N})$, the long-time second-order Trotter error over a total evolution time $T>0$, using $L$ time steps, satisfies
\begin{equation}
    \norm{\left( e^{-iH T} - \left( e^{-iA t/2} e^{-iB t} e^{-iA t/2} \right)^L \right) \psi_0}
    \;\leq\;
    \tilde{C}  N^{4.5}\; T\, t^{\frac{1}{4}}\;
    \norm{\psi_0}_{H^2},
\end{equation}
where $t = T/L$ denotes the short Trotter step size. 
Here, $\tilde{C}>0$ is an absolute constant depending only on the uniform bound $c_0$ of the Coulomb coefficients.
\end{thm}

As discussed above, our result applies to arbitrary initial states in $H^2\REV{(\mathbb{R}^{3N})}$, that is, any general initial conditions on which the Hamiltonian is well-defined. 
Moreover, the resulting error bound depends polynomially on the system size. We note that while prior significant results of Trotter analyses typically adopt a discretized formulation which would diverge in the continuum limit, our approach works directly at the level of the continuum Schr\"odinger equation, which is the natural formulation of the underlying PDE and remains finite as the number of spatial discretization degrees of freedom approaches infinity.

We remark that we do not attempt to optimize the constant appearing in the bound; rather, our primary goal is to establish the existence of an absolute constant with the stated properties.

Our result also shows that for general initial conditions \REV{in the domain of the Hamiltonian}, the convergence rate of the second-order (Strang) Trotter splitting with respect to the time step size is $1/4$. 
Notably, this rate coincides with previously reported numerical observations, where the ground state of the hydrogen atom was found to saturate such a quarter-order convergence rate (see~\cite{BurgarthFacchiHahnJohnssonYuasa2024}).
We further observe that, for general initial conditions \REV{in the domain of the Hamiltonian}, the first-order (Lie-Trotter) splitting also exhibits a convergence rate of $1/4$. 
This behavior was rigorously established in our prior work~\cite{FangWuSoffer2025} and is again consistent with numerical results in~\cite{BurgarthFacchiHahnJohnssonYuasa2024}. \REV{In addition, we prove exact one-step local error asymptotics of
order $t^{5/4}$ for both the first-order and second-order Trotter
formulas; see \cref{thm:sharp_local_error_asymptotics}. This shows that
increasing the order of the Trotter splitting does not improve the
convergence exponent uniformly over all numbers of time steps
$L\geq1$. However, this one-step result does not by itself rule
out sublinear accumulation of the local errors in particular cases,
and a fixed-time global lower bound in the many-step regime would
require a separate analysis.}

 \subsection{Main Result 2: Improved \REV{One-body} Convergence \REV{from Angular Momentum}}\label{sec:main-result-2}
 Given that the convergence rate for general initial conditions is $1/4$, which is lower than the rates usually expected in the bounded-operator case, it is natural to ask whether suitable regularity or structural assumptions on the initial quantum state can restore first-order convergence for the Lie-Trotter splitting and second-order convergence for the Strang splitting. We answer this question affirmatively for both the one-body and two-body cases.

We now turn to the one-body case. For completeness, we first specify the precise setting.
Consider the Schrödinger equation with a one-body Coulomb potential:
\begin{equation}\label{1-SE}
    \begin{cases}
        i\partial_t \psi(x,t) = \left(-\Delta \pm \dfrac{ c}{|x|}\right)\psi(x,t), \\[6pt]
        \psi(x,0) = \psi_0 \in H^2(\mathbb{R}^3),
    \end{cases}
    \qquad t \in \mathbb{R}, 
\end{equation}
where $-\Delta \equiv -\Delta_x$ is the Laplacian in $\mathbb{R}^3$, and $c \geq 0$. We note that this equation corresponds to the hydrogen atom after an appropriate change of coordinates; see the discussion of the two-body case in \cref{subsection:two-body-main} for further details. 

Before stating our main results in the one-body setting, we first recall several structural properties of the Coulomb Hamiltonian.
The Coulomb potential $\frac{c}{|x|}$ is spherically symmetric, and the Laplacian in $\mathbb{R}^3$ admits the following representation in spherical coordinates $(r,\omega)$, where $r = |x|$ and $\omega \in S^2$:
\begin{equation}\label{eq:lap_to_lap_S2}
    \Delta
= \partial_r^2 + \frac{2}{r}\partial_r + \frac{1}{r^2}\Delta_{S^2}.
\end{equation}
As a consequence, if the initial condition $\psi_0$ depends only on the radial variable $r$, then the corresponding solution $\psi(t)$ remains radial for all times $t$. More generally, the one-body Coulomb Hamiltonian
\begin{equation*}
H = -\partial_r^2 - \frac{2}{r}\partial_r - \frac{1}{r^2}\Delta_{S^2} \pm \frac{c}{r}
\end{equation*}
admits a separation-of-variables structure. This naturally motivates the spectral analysis of the angular operator $-\Delta_{S^2}$ on the unit sphere. Its eigenfunctions, the spherical harmonics, form an orthonormal basis of $L^2(S^2)$, and allow the full solution $\psi(t)$ to be expanded into angular momentum sectors.

Motivated by this, we let $\{ Y_{\ell,m} : -\ell \leq m \leq \ell \}$ be an orthonormal basis of the space $\mathcal{H}_\ell$ of spherical harmonics of degree $\ell$ in $\mathbb{R}^3$, for each $\ell \in \mathbb{N}$.  
We denote by $P_\ell$ the orthogonal projection onto $\mathcal{H}_\ell$, and denote
\begin{equation}
    P_{\ge \ell} := \sum_{k=\ell}^{\infty} P_k
\end{equation}
the orthogonal projection onto the space of all spherical harmonics of degree greater
than or equal to $\ell$ in $\mathbb{R}^3$. We are ready to describe the conditions for the initial states for the improved Trotter convergence rates.

\begin{assumption}[Angular-momentum and Sobolev Regularity]
\label{asp:angular_sobolev_condition}
Let $0 \leq \delta<\frac12$ and $\ell\in\mathbb{N}^+$. We assume that
\begin{equation}
\label{eq:angular_sobolev_condition}
\psi_0
=
P_{\geq\ell}\psi_0,
\qquad
\psi_0
\in
H^{1+\ell+\delta}(\mathbb{R}^3).
\end{equation}
\end{assumption}

The first condition in
\cref{asp:angular_sobolev_condition} excludes the angular momentum
sectors below $\ell$, while the second imposes additional Sobolev
regularity on the initial state. Since the Coulomb potential is radial,
the exact Coulomb evolution and the free Schr\"odinger evolution both
commute with $P_{\geq\ell}$. Therefore, the angular momentum condition
is preserved under both evolutions.

We also remark that the important aspect of this assumption is that it only applies to the initial condition, and the state itself does not need to be an eigenstate.

The case $\ell=0$ is not included in \cref{asp:angular_sobolev_condition}, since it does not exclude the radial sector responsible for the worst-case behavior. For general $H^2(\mathbb{R}^3)$ initial data, both Trotter formulas have the quarter-order rate discussed in \cref{sec:main-result-1,sec:local_error_lower_bounds}.

We also note that the role of angular momentum has been recognized in related studies of hydrogen eigenstates~\cite{BurgarthFacchiHahnJohnssonYuasa2024}. The present result complements these works by requiring only that the initial state satisfy \cref{asp:angular_sobolev_condition}, without assuming that $\psi_0$ is an eigenstate. For an eigenstate, the exact evolution preserves its spatial profile up to a phase factor, so no separate Sobolev propagation estimate is needed. For general initial data satisfying
\cref{asp:angular_sobolev_condition}, however, the spatial profile may evolve nontrivially, and the proof requires uniform propagation of the
relevant Sobolev and weighted estimates under the dynamics. Such propagation results are the main technical ingredients detailed in \cref{sec:improved_rate_3_ingredients} and proved in \cref{sec: key}. We record the corresponding eigenstate specialization in \cref{app:eigenstate_setting}.

To better understand the role of \cref{asp:angular_sobolev_condition}, we consider the bound-state eigenfunctions of the hydrogen atom Hamiltonian
\begin{equation}
\label{eq:H_one-body-hydrogen}
H
=
-\frac{1}{2}\Delta
-
\frac{1}{|x|}.
\end{equation}
We use the standard notation
\begin{equation}
\label{eq:Psi_nlm}
\Psi_{n\widetilde{\ell}m}(r,\omega)
=
R_{n\widetilde{\ell}}(r)
Y_{\widetilde{\ell},m}(\omega),
\end{equation}
where $n\in\mathbb{N}^+$ is the principal quantum number, $\widetilde{\ell}=0,1,\ldots,n-1$ is the orbital angular momentum quantum number, and
$m=-\widetilde{\ell},\ldots,\widetilde{\ell}$ is the magnetic quantum number. In particular,
\begin{equation}
P_{\widetilde{\ell}}
\Psi_{n\widetilde{\ell}m}
=
\Psi_{n\widetilde{\ell}m}.
\end{equation}

The radial function satisfies
\begin{equation}
R_{n\widetilde{\ell}}(r)
\sim
r^{\widetilde{\ell}}
\qquad
\text{as }
r\to0,
\end{equation}
and decays exponentially as $r\to\infty$. The hydrogen eigenfunctions therefore satisfy
\begin{equation}
\Psi_{n\widetilde{\ell}m}
\in
H^{1+\widetilde{\ell}+\delta}(\mathbb{R}^3)
\qquad
\text{for every }
0 \leq \delta<\frac12.
\end{equation}
Consequently, every eigenstate with $\widetilde{\ell}\geq1$ satisfies \cref{asp:angular_sobolev_condition} with $\ell=\widetilde{\ell}$.

The ground state is given by
\begin{equation}
\Psi_{100}(x)
=
\frac{1}{\sqrt{\pi}}e^{-|x|}.
\end{equation}
Since it is radial, it lies entirely in the zeroth angular momentum sector and does not satisfy the angular momentum condition in \cref{asp:angular_sobolev_condition}. This is consistent with the quarter-order convergence observed for the hydrogen ground state.

As a concrete nonradial example, consider
\begin{equation}
\Psi_{320}(r,\omega)
=
C r^2e^{-r/3}Y_{2,0}(\omega),
\end{equation}
where $C>0$ is a normalization constant. This state lies in the angular momentum sector $\widetilde{\ell}=2$ and belongs to
$H^{3+\delta}(\mathbb{R}^3)$ for every
$0<\delta<\frac12$. It therefore satisfies
\cref{asp:angular_sobolev_condition} with $\ell=2$.
The corresponding first-order and second-order convergence rates follow from the two theorems below.

We now present our main result 2 in the one-body case. For the first-order Trotter splitting, we have:

\begin{thm}[First-order Trotter rate for $P_{\geq \ell}$ data]
\label{thm: 1st-Trotter}
Let $H=A+B$ be the one-body Coulomb Hamiltonian defined in \cref{1-SE,def:A_and_B}. Suppose that \cref{asp:angular_sobolev_condition} holds. Let $T>0$ and the number of steps $L\in\mathbb{N}^+$, and set the step size $t=T/L\leq1$. 
Then
\begin{equation}
\label{eq:first_order_angular_rate}
\norm{
\left(
e^{- i HT}
-
\left(
e^{- i Bt}e^{- i At}
\right)^L
\right)
\psi_0
}
\leq
C_{\ell,\delta}\,
Tt^{\gamma_1(\ell,\delta)}
\norm{\psi_0}_{H^{1+\ell+\delta}},
\end{equation}
with the rate
\begin{equation}
\label{eq:def_gamma_1}
\gamma_1(\ell,\delta)
:=
\min
\left\{
1,\,
\frac{\ell+\delta}{2}
\right\},
\end{equation}
where $C_{\ell,\delta}>0$ depends only on $\ell$, $\delta$, and the coefficient $c$ in the Coulomb potential.
\end{thm}

In particular, when $\ell=1$, the first-order convergence exponent is
\begin{equation}
\gamma_1(1,\delta)
=
\frac{1+\delta}{2},
\end{equation}
which can be taken arbitrarily close to $3/4$ by choosing $\delta<1/2$ sufficiently close to $1/2$. When $\ell\geq2$, we have $\gamma_1(\ell,\delta)=1$, and the full first-order convergence rate is recovered.

We next consider the second-order Trotter formula.

\begin{thm}[Improved Second-order Trotter Rate]
\label{thm: 2-rd Trotter}
Under the same conditions as in \cref{thm: 1st-Trotter}, then the second-order Trotter error satisfies
\begin{equation}
\label{eq:second_order_angular_rate}
\norm{
\left(
e^{- i HT}
-
\left(
e^{- i At/2}
e^{- i Bt}
e^{- i At/2}
\right)^L
\right)
\psi_0
}
\leq
C_{\ell,\delta}\,
Tt^{\gamma_2(\ell,\delta)}
\norm{\psi_0}_{H^{1+\ell+\delta}},
\end{equation}
with the short step size $t = T/L \leq 1$ and the rate
\begin{equation}
\label{eq:def_gamma_2}
\gamma_2(\ell,\delta)
:=
\begin{cases}
\dfrac{1+\delta}{2},
& \ell=1,\\[5pt]
1,
& \ell=2,\\[3pt]
\dfrac{3+\delta}{2},
& \ell=3,\\[5pt]
2,
& \ell\geq4,
\end{cases}
\end{equation}
where $C_{\ell,\delta}>0$ depends only on $\ell$, $\delta$, and the coefficient $c$ in the Coulomb potential.
\end{thm}

We remark that the assumptions in \cref{thm: 2-rd Trotter} can be weakened when $\ell\geq4$. In fact, the proof of the second-order rate only requires
\begin{equation}
\psi_0
=
P_{\geq4}\psi_0,
\qquad
\psi_0\in H^5(\mathbb{R}^3).
\end{equation}
These conditions provide all the weighted derivative estimates needed to control the two double commutators. Consequently,
\begin{equation}
\norm{
\left(
e^{-iHT}
-
\left(
e^{-iAt/2}e^{-iBt}e^{-iAt/2}
\right)^L
\right)
\psi_0
}
\leq
C \,Tt^2\norm{\psi_0}_{H^5},
\end{equation}
Thus, increasing the angular momentum threshold beyond $4$ does not require any additional Sobolev regularity to retain the second-order convergence rate.

\subsection{Implication of Main Result 2 in the Two-body Case}\label{subsection:two-body-main}

In this section, we present the improved convergence rates for both the first and second-order Trotter splittings in the two-body case. The purpose of this section is to make transparent that the one-body result naturally extends to the two-body case, as the latter can be essentially reduced to the former after a change of coordinates and separation of variables.

Before proceeding, we note that the spatial notation used in this subsection differs slightly from that in the rest of the paper. To remain consistent with standard physical conventions, we denote the electron and proton positions by $r_e$ and $r_p$, respectively, rather than by a generic variable $x$. We further introduce the relative coordinate $r = r_e - r_p$. This notation is used only within the present subsection and should not be confused with the notation employed elsewhere in the paper.

For concreteness, we consider the hydrogen atom with one electron and one nucleus, where the first-principles Hamiltonian reads
\begin{equation}\label{eq:two-body-Ham}
H = -\frac{\hbar^2}{2m_e} \Delta_e - \frac{\hbar^2}{2m_p} \Delta_p - \frac{e^2}{|r_e - r_p|},
\end{equation}
where $r_e$ and $r_p$ are the electron and proton positions, and $m_e$ and $m_p$ are their masses, respectively. Following the usual route, we change the coordinate by considering the relative coordinates
\begin{equation}\label{eq: R}
  R = \frac{m_e r_e + m_p r_p}{m_e + m_p},
\end{equation}
and
\begin{equation}\label{eq: r}
  r = r_e - r_p.
\end{equation}
Then the system becomes
\begin{equation}\label{eq:two-body-Ham-hydrogen}
H = -\frac{\hbar^2}{2M} \Delta_R - \frac{\hbar^2}{2\mu} \Delta_r - \frac{e^2}{|r|},
\end{equation}
where $M = m_e + m_p$ is the total mass and $\mu = \frac{m_e m_p}{m_e + m_p}$ is the reduced mass. In the usual setting of electronic structure problems, one exploits the fact that $M  \gg 1$ (the Born-Oppenheimer approximation), and hence gets the effective one-body problem
\begin{equation}\label{eq: def Hrel}
H_{\text{rel}} = -\frac{\hbar^2}{2\mu} \Delta_r - \frac{e^2}{|r|},
\end{equation}
which we have analyzed in \cref{sec:main-result-2}. In this section we consider the case without such a Born-Oppenheimer approximation. It is straightforward to observe from \cref{eq:two-body-Ham-hydrogen} that the whole system can be treated via separation of variables in $r$ and $R$. In light of this, we have the following two-body result.

For each $\ell \in \mathbb{N}$, let $\{Y_{\ell,m} : -\ell \leq m \leq \ell\}$ be an orthonormal basis of the space $\mathcal{H}_\ell$ of spherical harmonics of degree $\ell$ in $\mathbb{R}^3$.  
For $y \in \mathbb{R}^3$, let $P_{\ell,y}$ denote the orthogonal projection onto $\mathcal{H}_\ell$ with respect to the variable $y$, and set
\[
P_{\ge \ell, y} := \sum_{k=\ell}^\infty P_{k,y},
\]
which is the orthogonal projection onto the subspace of all spherical harmonics of degree at least $\ell$. This definition coincides with that used in \cref{sec:main-result-2}, except that here we explicitly indicate the coordinate $y$ on which the projection acts.

\begin{assumption}\label{asp:condition two-body}
\REV{
Let $0 \leq \delta<\frac12$ and $\ell\in\mathbb{N}^+$. We assume that
\begin{equation}
\psi_0
=
P_{\geq\ell,r_e-r_p}\psi_0,
\qquad
\psi_0
\in
H^{1+\ell+\delta}(\mathbb{R}^6).
\end{equation}
}
\end{assumption}
Let $H = A+B$, where
\begin{equation}\label{eq:def_A_B_two-body}
A = -\frac{\hbar^2}{2m_e} \Delta_e - \frac{\hbar^2}{2m_p} \Delta_p  
, \quad
B =  - \frac{e^2}{|r_e - r_p|}.
\end{equation}
We have the following improved convergence rate for the Trotter splittings. 

\begin{thm}[First-order Trotter Error -- Two-body]
\label{thm: 1st-Trotter two-body}
Let $H=A+B$ be given by \cref{eq:def_A_B_two-body}. If Assumption~\ref{asp:condition two-body} holds, then the long-time first-order Trotter error over a total evolution time $T>0$, using $L$ time steps with the Trotter step size $t=T/L \REV{\leq 1}$, satisfies the bounds
\begin{equation}\label{eq: 1st ell1 two-body}
\norm{
\left(
e^{-iHT}
-
\left(e^{-iBt}e^{-iAt}\right)^L
\right)\psi_0
}
\leq
C_{1st}\,
T\REV{t^{\gamma_1(\ell,\delta)}
\norm{\psi_0}_{H^{1+\ell+\delta}}},
\end{equation}
for some absolute constant $C_{1st}=C_{1st}\REV{(\ell,\delta, m_e,m_p,\hbar)}>0$. The exponent $\gamma_1(\ell,\delta)$ is defined in
\cref{eq:def_gamma_1}.
\end{thm}

\begin{thm}[Second-order Trotter Error -- Two-body]
\label{thm: 2-rd Trotter two-body}
Under the same condition of \cref{thm: 1st-Trotter two-body}, the long-time second-order Trotter error over a total evolution time $T>0$, using $L$ time steps with step size $t=T/L \leq 1$, satisfies
\begin{equation}\label{eq: thm6 two-body}
\norm{
\left(
e^{-iHT}
-
\left(
e^{-iAt/2}e^{-iBt}e^{-iAt/2}
\right)^L
\right)\psi_0
}
\leq
C_{2nd}\,
Tt^{\gamma_2(\ell,\delta)}
\norm{\psi_0}_{H^{1+\ell+\delta}}.
\end{equation}
for some constant $C_{2nd}=C_{2nd}(\ell,\delta,m_e,m_p,\hbar)>0$. The exponent $\gamma_2(\ell,\delta)$ is defined in \cref{eq:def_gamma_2}.
\end{thm}

\subsection{Main Result 3: Improved Convergence for Many-body Fermionic States}
\label{sec:main-result-3-many-body}

We now extend the improved first-order estimate in the one-body case to a class of many-body initial states. The main observation is that
the one-body estimate corresponding to $\ell=1$ can be applied in each relative coordinate $x_j-x_k$. For a wave function that is fully antisymmetric in the spatial particle coordinates, the zeroth angular momentum component vanishes automatically in every such relative
coordinate. This leads to improved convergence rates for both the first-order and second-order Trotter formulas.

\begin{defn}[Fermionic Antisymmetry]
\label{defn:fermionic}
Let $N\geq2$. A wave function
$\Psi(x_1,\ldots,x_N)$ on $(\mathbb R^3)^N$ has fermionic antisymmetry if,
for every transposition $(j\,k)$,
\[
    \Psi(x_1,\ldots,x_j,\ldots,x_k,\ldots,x_N)
    =
    -\Psi(x_1,\ldots,x_k,\ldots,x_j,\ldots,x_N).
\]
\end{defn}

A key observation is that the fermionic antisymmetry leads to the projection property in the many-body case, which may be of independent interest.

\begin{prop}[Relative angular momentum for Fermionic states]\label{prop:Nbody_Pge1}
Let $\Psi$ have fermionic antisymmetry. Then, for every pair $1\leq j<k\leq N$,
\[
    \Psi=P_{\geq1,x_j-x_k}\Psi .
\]
Equivalently, the angular average of $\Psi$ in each relative variable $x_j-x_k$ vanishes.
\end{prop}

\begin{thm}[$N$-body first-order Trotter rate for Fermionic data]
\label{prop:Nbody-rate}
Consider an $N$-body Coulomb Hamiltonian
\[
    H=-\Delta+\sum_{1\leq j<k\leq N}\frac{c_{jk}}{|x_j-x_k|}
\]
on $\mathbb R^{3N}$. Assume that $H$ is invariant under particle exchanges (for example, $c_{jk}=c$ for all $j<k$). Let $0 \leq \delta<\frac12$. If $\Psi_0\in H^{2+\delta}(\mathbb R^{3N})$ has fermionic antisymmetry, then the first-order Trotter error satisfies, for fixed final time $T>0$,
\begin{equation}
  \norm{\left( e^{-iH T} - \left(e^{-iB t} e^{-iA t} \right)^L \right) \Psi_0} 
  \leq
C_{\delta,c_0}N^{13/2}T
       t^{(1+\delta)/2}
       \norm{\Psi_0}_{H^{2+\delta}},
\end{equation}
where $t = T/L \leq 1$ is the time step size, and $C_{\delta,c_0}$ is a constant depending only on $\delta$ and $c_0$ defined in~\cref{con: c0}. 

At the endpoint $\delta=0$, the dependence on the particle number can be improved as
\begin{equation}
\label{eq:Nbody_fermionic_half_order}
\norm{
\left(
e^{-iHT}
-
\left(e^{-iBt}e^{-iAt}\right)^L
\right)
\Psi_0
}
\leq
C_{c_0}N^{9/2}Tt^{1/2}
\norm{\Psi_0}_{H^2},
\end{equation}
where $C_{c_0}>0$ depends only on $c_0$.
\end{thm}

In particular, note that the endpoint $\delta=0$ case shows that, for a permutation-invariant $N$-body Coulomb Hamiltonian and spatially antisymmetric initial data (both of which hold in the fermionic case), the first-order Trotter formula has a global convergence rate of $1/2$ with respect to the time-step size for general $H^2(\mathbb{R}^{3N})$ initial data, improving upon the worst-case rate of $1/4$ in the general case.

We also have a similar improvement for the second-order Trotter formula. 
\begin{thm}[$N$-body Second-order Trotter Rate for Fermionic Data]
\label{thm:Nbody_fermionic_second_order_rate}
Under the same assumptions as in \cref{prop:Nbody-rate}, let $T>0$ and $L\in\mathbb{N}^+$, and set $t=T/L\leq1$. Then
\begin{equation}
\label{eq:Nbody_fermionic_second_order_rate}
\norm{
\left(
e^{- i HT}
-
\left(
e^{- i At/2}e^{- i Bt}e^{- i At/2}
\right)^L
\right)
\Psi_0
} 
\leq
C_{\delta,c_0}N^{13/2}Tt^{(1+\delta)/2}
\norm{\Psi_0}_{H^{2+\delta}},
\end{equation}
where $C_{\delta,c_0}$ is a constant depending only on $\delta$ and $c_0$ defined in~\cref{con: c0}. 
\end{thm}
The endpoint $\delta = 0$ case provides the same estimate as the right-hand side of \cref{eq:Nbody_fermionic_half_order}, which we omit here.

\subsection{Organization of the Proofs}

In this section, we outline the main ideas underlying the proofs of our results and explain how the remainder of the paper is organized. Rather than presenting full technical details at once, our goal here is to highlight the key mechanisms and ingredients that drive the analysis.

For the reader's convenience, we also provide a roadmap indicating where the proofs of the main results are located. In particular:
\begin{itemize}
    \item The proof of Main Result 1 (\cref{thm:main_trotter_long-main}) for general initial conditions \REV{in the domain of the Hamiltonian} is given in \cref{sec:pf_main_result1}.

    \item The proofs of Main Result 2 in the one-body case (\cref{thm: 1st-Trotter} and \cref{thm: 2-rd Trotter}), together with the implications for the two-body case (\cref{thm: 1st-Trotter two-body} and \cref{thm: 2-rd Trotter two-body}), are presented in \cref{sec:pf_main_result2_improved_rate}.

    \item \REV{The proof of the first-order and second-order estimates in Main Result 3 (\cref{prop:Nbody-rate}) is given in \cref{sec:pf_main3_many-body-fermion}.}

    \item \REV{The sharp one-step local error lower bounds in
    \cref{thm:sharp_local_error_asymptotics} are proved in
    \cref{sec:local_error_lower_bounds}.}

\end{itemize}

The proof of Main Result~1 uses results established in our previous work~\cite{FangWuSoffer2025}, including Sobolev norm estimates and first-order Trotter error bounds, which we briefly review in \cref{sec:pf_E1_E2_error_thm_prep}. The main new ingredient is a precise connection between the first- and second-order Trotter formulas in the presence of the Coulomb singularity (\cref{thm:main_trotter_long}). This result is proved in \cref{sec:pf_E1_E2_error_thm} and subsequently used to derive Main Result~1 (\cref{thm:main_trotter_long-main}) in \cref{sec:pf_long_time_Trotter2_general}.

The proof of Main Result~2 begins with the case $\ell=1$, where the main mechanism behind the improved rate is already visible. This direct argument uses the Hardy-type estimate in \cref{prop: S2H2} and the lower-order Sobolev propagation estimate in \cref{lem:lower-order-propagation}, whose proofs are given in \cref{sec: tool}. The extension to general $P_{\geq\ell}$ initial data uses the Sobolev and weighted propagation estimate in \cref{prop:Pgel-propagation}, the angular weighted estimate in \cref{prop:Pgel-weighted}, and the exact local error representation in \cref{thm:exact_err_rep_Trotter2_general}. These three ingredients are discussed in \cref{sec:improved_rate_3_ingredients}. The first two are proved in \cref{sec: key}, while the exact error representation is proved as part of the second-order analysis in \cref{sec:pf_improved_rate_Trotter2}. The first-order and second-order
estimates are proved in \cref{sec:pf_improved_rate_Trotter1,sec:pf_improved_rate_Trotter2}, respectively. The corresponding two-body results then follow by separating the center-of-mass and relative coordinates.

The proof of the first-order part of Main Result~3 uses fermionic antisymmetry to exclude the zeroth angular momentum sector in every relative coordinate. This allows the one-body weighted estimate to be applied pairwise, together with relative-momentum counting and many-body Sobolev propagation. The techniques and intermediate results for the many-body fermionic case may be of independent interest. The details are given in \cref{sec:pf_main3_many-body-fermion}.

\section{Proof of Main Result 1 (\cref{thm:main_trotter_long-main})}
\label{sec:pf_main_result1}

Let $H = A +B$ with \( A = -\Delta \) and \( B = V \), the same as before. For a total evolution time \( T > 0 \), we define the first-order and second-order Trotter errors with short-time Trotter step size \( t = T/L \) by
\begin{equation}
    E_{1,L}(t)f := \left( \left( e^{-iB t} e^{-iA t} \right)^L - e^{-i  H T} \right) f, \qquad \forall\, f \in H^2\REV{(\mathbb{R}^{3N})},
\end{equation}
and
\begin{equation}
    E_{2,L}(t)f := \left( \left( e^{-iA t/2} e^{-iBt} e^{-iA t/2} \right)^L - e^{-i H T} \right) f, \qquad \forall\, f \in H^2\REV{(\mathbb{R}^{3N})},
\end{equation}
respectively. 

One immediate relationship between the two is given by
\begin{equation}\label{eq:E1_E2_trotter_error_relation}
    \begin{aligned}
        E_{2,L+1}(t) \psi_0 = &  \left(E_{2,L+1}(t) - e^{-iAt/2} E_{1,L}(t) e^{-iBt} e^{-iAt/2} \right)\psi_0 
        \\ & +  e^{-iAt/2} E_{1,L}(t) \left( e^{-iBt} e^{-iAt} - e^{-iHt} \right) e^{iAt/2} \psi_0 
        \\ & + e^{-iAt/2} E_{1,L}(t) e^{-iHt} e^{iAt/2} \psi_0.
    \end{aligned}
\end{equation}
As proved in~\cite{FangWuSoffer2025}, both $e^{-iAt}$ and $e^{-iHt}$ map the $H^2\REV{(\mathbb{R}^{3N})}$ (the domain of the Hamiltonian) into itself, whereas $e^{-iBt}$ does not (see also~\cref{lem:sobolev_n-body} and \cite[Section 2.2]{FangWuSoffer2025}). Hence, for the second and third terms in \cref{eq:E1_E2_trotter_error_relation}, it is essential that no $e^{-iBt}$ factors occur on the right-hand side, so that we can pass along norms in the proper sense.

In order to prove our main result \cref{thm:main_trotter_long-main}, the only remaining ingredient is the following theorem, which controls the first term in \cref{eq:E1_E2_trotter_error_relation}.

\begin{thm}[Long-time First- and Second-order Trotter Errors] \label{thm:main_trotter_long}
Let \( H = A + B \) be the \( N \)-body Hamiltonian with Coulomb interactions as defined in~\cref{N-SE,eq:N-V_def,con: c0}, where \( A = -\Delta \) denotes the kinetic part and \( B = V(x) \) the Coulomb interaction potential. Then, for any initial state \( \psi_0 \in H^2\REV{(\mathbb{R}^{3N})} \), the long-time first- and second-order Trotter errors for a total evolution time \( T > 0 \) using \( L \) time steps satisfy
\begin{equation}\label{eq: goal}
    \left\| \left(E_{2,L+1}(t) - e^{-iAt/2} E_{1,L}(t) e^{-iBt} e^{-iAt/2} \right)\psi_0 \right\| 
    \leq \left(\tilde C_N t^{5/4} + \frac{C_{N}}{2} t\right) \|\psi_0\|_{H^2},
\end{equation}
where \( t = T/L \) is the short-time Trotter step size. The constants \( \tilde C_N = \mathcal{O}(N^{4.5}) \) and \( C_{N} = \mathcal{O}(N^3) \) are defined in~\cref{thm: short Trotter,lem: Ne}, respectively, and depend polynomially on the system size \( N \).

\end{thm}

Once this theorem is established, the proof of \cref{thm:main_trotter_long-main} follows straightforwardly. Indeed, the second term in \cref{eq:E1_E2_trotter_error_relation} reduces to a one-step first-order Trotter error, while the third term can be controlled by an estimate associated with the long-time first-order Trotter error operator.

The rest of this section is organized as follows. In \cref{sec:pf_E1_E2_error_thm_prep}, we recall several key estimates on the first-order Trotter estimate and solution properties proved in \cite{FangWuSoffer2025}, which will be used in this work. We then prove \cref{thm:main_trotter_long} in \cref{sec:pf_E1_E2_error_thm}. Finally, we use this result to establish \cref{thm:main_trotter_long-main} in \cref{sec:pf_long_time_Trotter2_general}.

\subsection{Auxiliary Estimates}\label{sec:pf_E1_E2_error_thm_prep}
In this section, we review a few core results proved in our previous work \cite{FangWuSoffer2025} on the first-order Trotter splitting for general initial conditions, which will be used in the proofs of our main results.

The first helpful result is the alternative exact error representation of the first-order Trotter local error operator (\cite[Lemma 9]{FangWuSoffer2025}):
\begin{equation}\label{eq:Trotter1_error_operator_exact_rep}
    e^{-iBt}e^{-iAt} - e^{-iHt} = i\int_0^t ds \, e^{-is B }[e^{-is A}, B]e^{-i(t-s)H}.
\end{equation}
In light of this, we define the local truncation error operator acting on the solution $e^{-i\sigma H}\psi_0$ at time $\sigma = t \REV{k} \in [0,T]$ by
\begin{equation}
e_\sigma(t)
:=
\int_0^t ds\,
e^{-is\REV{B}}
[e^{-is\REV{A}},\REV{B}]
e^{-i(t-s+\sigma)H}\psi_0,
\qquad
\psi_0\in H^2 \REV{(\mathbb{R}^{3N})}.
\end{equation}
There are two helpful results regarding this. The first is its accumulation gives the long-time first-order Trotter error, as proved in \cite[Equation (59) and Lemma 9]{FangWuSoffer2025}:
\begin{equation}\label{eq:long_term_error_by_lte_trotter1}
\norm{E_{1,L} \psi_0} \leq 
\sum_{\REV{k}=0}^{L-1} \norm{\int_0^t ds \, e^{-is B }[e^{-is A}, B]e^{-i(t-s+t \REV{k} )H} \psi_0} = \sum_{\REV{k} =0}^{L-1} \norm{e_{t \REV{k} } (t)}.
\end{equation}
The second is its estimate, as proved in \cite[Theorem 10]{FangWuSoffer2025}:
\begin{lemma}[$N$-body Short-time Trotter Error, \cite{FangWuSoffer2025}]\label{thm: short Trotter} $\forall N\in \mathbb{N}^+$. If the condition~\eqref{con: c0} holds, then for any time step size \( t \in (0,1] \),
\begin{equation}
   \sup_{\sigma \in [0, T]} \|e_\sigma(t)\| \leq  \tilde{C}_N\, t^{5/4} \|\psi_0\|_{H^2},
\end{equation}
where $ \tilde{C}_N$ is bounded by an absolute constant times $N^{4.5}$, precisely defined by 
\begin{equation}\label{eq:def_tilde_CN}
     \tilde{C}_N := C c_0\left( (N-1)N^{\frac{3}{2}} + (N-1)N^{\frac{1}{2}}(C_N - 1) \right),
\end{equation}
with $C$ an absolute (universal) constant and $C_N = \Theta(N^3)$ given in~\cref{eq:def_CN_sobolev}.

\end{lemma}

Another important estimate is the Sobolev norm estimate for the many-body system given by \cite[Theorem 7]{FangWuSoffer2025}:
\begin{lemma}[$N$-body Sobolev Norm~\cite{FangWuSoffer2025}]\label{lem:sobolev_n-body}
Under the same conditions of~\cref{thm: short Trotter} , the Sobolev norm of the solution $\psi(t) = e^{-iHt}\psi_0$ of~\cref{N-SE} at any time $t>0$ can be estimated as
\begin{equation}
    \|\psi(t)\|_{H^2}\leq C_{N}\|\psi_0\|_{H^2},
\end{equation}
uniformly in $t$, 
where $C_N$ is defined by
\begin{equation} \label{eq:def_CN_sobolev}
    C_N := 2 + 6 c_0  N^{3/2} + 8 c_0^2 N^3.
\end{equation} Moreover, the solution $\psi(t)$ also satisfies the estimate
\begin{equation} \label{thmN est}
    \|(-\Delta)\psi(t)\| \leq (C_N-1)\| \psi_0 \|_{H^2}.
\end{equation} 
\end{lemma}

\subsection{Proof of~\cref{thm:main_trotter_long}}\label{sec:pf_E1_E2_error_thm}

The proof of~\cref{thm:main_trotter_long} also requires the following lemma, whose proof is similar to that of~\cite[Theorem 10]{FangWuSoffer2025}.
\begin{lemma}\label{lem: Ne}Let $V, H, T$ and $t$ be as in~\cref{thm:main_trotter_long}. Then 
\begin{equation}\label{eq: est lem4 goal}
    \norm{\left(e^{-i(-\Delta)t/2}e^{-i(T+t)H}e^{i(-\Delta)t/2}-e^{-i(T+t)H}\right)}_{H^2\to L^2}\leq \frac{C_{N}}{2}t,
\end{equation}
where $C_N$ is the same constant as in~\cref{lem:sobolev_n-body}.
    
\end{lemma}
\begin{proof} 
Take $f \in H^2\REV{(\mathbb{R}^{3N})}$. By~\cref{lem:sobolev_n-body}, we have $e^{-itH}f \in H^2\REV{(\mathbb{R}^{3N})}$ for all $t > 0$. Then the relation
\begin{equation}\label{eq: comm eTH A}
\begin{aligned}
 f_T(t):=   &\left(e^{-i(-\Delta)t/2}e^{-i(T+t)H}e^{i(-\Delta)t/2} - e^{-i(T+t)H}\right)f \\
    = &i\int_0^{\frac{t}{2}} ds\, e^{-i(-\Delta)s} [e^{-i(T+t)H}, -\Delta] e^{i(-\Delta)s} f\\
    =& i\int_0^{\frac{t}{2}} ds\, e^{-i(-\Delta)s} e^{-i(T+t)H} (-\Delta) e^{i(-\Delta)s} f
    \\
    &
    -i\int_0^{\frac{t}{2}} ds\, e^{-i(-\Delta)s}(-\Delta) e^{-i(T+t)H}  e^{i(-\Delta)s} f
\end{aligned}
\end{equation}
is valid in $L^2$. By~\cref{lem:sobolev_n-body} and the unitarity of $e^{\pm i(-\Delta)s}$ and $e^{-i(T+t)H}$, we have
\begin{equation}\label{eq: est fT1}
    \|f_{T}(t)\| \leq \int_0^{\frac{t}{2}} ds\, \left(\|(-\Delta)f\|+\|(-\Delta)e^{-i(T+t)H}e^{i(-\Delta)s}f\|\right) \leq \frac{C_N t}{2} \|f\|_{H^2},
\end{equation}
where $C_N$ is the same constant as in~\cref{lem:sobolev_n-body}.
\end{proof}

We are now ready to prove~\cref{thm:main_trotter_long}.
\begin{proof}[Proof of~\cref{thm:main_trotter_long}] We observe that
\begin{equation}
    \left(e^{-iAt/2}e^{-iBt}e^{-iAt/2}\right)^{L+1} = e^{-iAt/2} \left(e^{-iBt}e^{-iAt}\right)^L e^{-iBt}e^{-iAt/2}.
\end{equation}
This identity yields
\begin{equation}\label{eq: def eL}
    \begin{aligned}
        e_L(t)\psi_0 
        &:= \left(E_{2,L+1}(t) - e^{-iAt/2} E_{1,L}(t) e^{-iBt} e^{-iAt/2} \right)\psi_0 \\
        &= \left(e^{-iAt/2} e^{-iTH} e^{-iBt} e^{-iAt/2} - e^{-i(T+t)H} \right)\psi_0.
    \end{aligned}
\end{equation}
To estimate \( e_L(t)\psi_0 \), we decompose it into two parts:
\begin{equation}\label{eq: eL}
    e_L(t)\psi_0 = e_{L1}(t)\psi_0 + e_{L2}(t)\psi_0,
\end{equation}
where \( e_{Lj}(t)\psi_0 \), for \( j = 1, 2 \), are defined as follows:
\begin{align}
    e_{L1}(t)\psi_0 &:= \left(e^{-iAt/2} e^{-iTH}  \right) \left( e^{-iBt}e^{-iAt}-e^{-itH}\right) e^{iAt/2} \psi_0, \\
    e_{L2}(t)\psi_0 &:= \left( e^{-iAt/2} e^{-i(T+t)H} e^{iAt/2} - e^{-i(T+t)H} \right) \psi_0.
\end{align}
For $e_{L1}(t)\psi_0$, by~\cref{thm: short Trotter} and again using the unitarity of $e^{-iAt/2}$ and $e^{-iTH}$, we have
\begin{equation}\label{est: eL2}
    \| e_{L1}(t)\psi_0 \| = \| e_\sigma(t) \vert_{\sigma = 0, \psi(0)=e^{iAt/2}\psi_0} \| \leq \tilde{C}_N t^{5/4} \|e^{iAt/2}\psi_0\|_{H^2}=\tilde{C}_N t^{5/4} \|\psi_0\|_{H^2},
\end{equation}
where we used the fact that $e^{-iAt}$ preserves the $H^2$ norm.
For \( e_{L2}(t)\psi_0 \), taking \(A=-\Delta\) and applying \cref{lem: Ne}, we obtain
\begin{equation}\label{est: eL3}
    \| e_{L2}(t)\psi_0 \| \leq \frac{C_{N}}{2} t \|\psi_0\|_{H^2}.
\end{equation}
Combining estimates~\eqref{est: eL2}, and~\eqref{est: eL3} with~\eqref{eq: def eL} and~\eqref{eq: eL} yields~\cref{eq: goal}.\end{proof}

\subsection{Proof of~\cref{thm:main_trotter_long-main}}\label{sec:pf_long_time_Trotter2_general}
\begin{proof}[Proof of~\cref{thm:main_trotter_long-main} (using \cref{thm:main_trotter_long})] 
    Taking the $L^2$ norm of \cref{eq:E1_E2_trotter_error_relation} gives
\begin{equation}\label{eq:est_sum3terms_E2L}
    \begin{aligned}
       & \norm{ E_{2,L+1}(t) \psi_0}
       \\
       \leq &  \norm{\left(E_{2,L+1}(t) - e^{-iAt/2} E_{1,L}(t) e^{-iBt} e^{-iAt/2} \right)\psi_0} 
        \\ & + \norm{ e^{-iAt/2}}_{L^2 \to L^2} \norm{E_{1,L}(t)\left(e^{-iBt}  e^{-iAt} - e^{-iHt} \right)  e^{iAt/2} \psi_0 }_{L^2}
        \\ 
        & + \norm{e^{-iAt/2}}_{L^2 \to L^2} \norm{E_{1,L}(t) e^{-iHt} e^{iAt/2} \psi_0}.
    \end{aligned}
\end{equation}   
For the first term on the right-hand side, we use~\cref{thm:main_trotter_long}. For the second term, we invoke
\begin{equation}
    \norm{E_{1,L}}_{L^2 \to L^2} \leq \norm{\left( e^{-iB t} e^{-iA t} \right)^L }_{L^2 \to L^2} +  \norm{ e^{-i T (A + B)}}_{L^2 \to L^2} \leq 2.
\end{equation}
As a result, the second term is reduced to a one-step error, for which we can apply \cref{thm: short Trotter} by setting $\sigma = 0$ and choosing the initial state as $e^{iAt/2} \psi_0$ (the same as in \cref{est: eL2}). Consequently, the second term of \cref{eq:est_sum3terms_E2L} is bounded by 
\begin{equation}\label{eq:E1L_term2}
\begin{aligned}
    & \norm{E_{1,L}(t) \left(e^{-iBt}  e^{-iAt} - e^{-iHt} \right)  e^{iAt/2} \psi_0 }_{L^2} 
    \\
 \leq  & \norm{E_{1,L}(t)}_{L^2 \to L^2} \norm{ \left(e^{-iBt}  e^{-iAt} - e^{-iHt} \right)  e^{iAt/2} \psi_0 }_{L^2}
   \leq  2 \tilde{C}_N t^{5/4} \norm{\psi_0}_{H^2}.
\end{aligned}
\end{equation}
The third term is reduced to the long-time first-order Trotter error operator $E_{1,L}$ acting on an $H^2\REV{(\mathbb{R}^{3N})}$ initial condition $e^{-iHt} e^{-iAt/2} \psi_0$. To be specific, by \cref{eq:long_term_error_by_lte_trotter1} we have
\begin{equation}\label{eq:E1L_term3}
\begin{aligned}
    \norm{E_{1,L}  e^{-iHt} e^{iAt/2} \psi_0 } 
   & \leq 
\sum_{\REV{k}=0}^{L-1} \norm{\int_0^t ds \, e^{-is B }[e^{-is A}, B]e^{-i(t-s+t \REV{k})H} e^{-iHt} e^{iAt/2} \psi_0}
\\
& \leq 
\sum_{\REV{k}=0}^{L-1}
\norm{ e_\sigma(t) \vert_{\sigma = t\REV{k} +t, \psi(0)=e^{iAt/2}\psi_0} }
\\
& \leq \sum_{\REV{k}=0}^{L-1} \tilde{C}_N t^{5/4} \|e^{iAt/2}\psi_0\|_{H^2}=T \tilde{C}_N t^{1/4} \|\psi_0\|_{H^2},
 \end{aligned}    
\end{equation}
where we used again the fact that $e^{-iAt}$ preserves the $H^2$ norm.
Combining estimates (\cref{eq:E1L_term3,eq:E1L_term2,eq:est_sum3terms_E2L}) together with \cref{thm:main_trotter_long} yields the desired bound
\begin{equation}
   \norm{ E_{2,L+1}(t) \psi_0} \leq  
   \left(3 \tilde C_N t^{5/4} + \frac{C_{N}}{2} t +  \tilde{C}_N T t^{1/4} \right) \|\psi_0\|_{H^2}. 
\end{equation}
Recall the definitions of $C_N = \Theta(N^3)$ in \cref{eq:def_CN_sobolev} and $\tilde C_N = \Theta(N^{4.5})$ in \cref{eq:def_tilde_CN}, and we have completed the proof of
\cref{thm:main_trotter_long-main}. 
\end{proof}

\section{Sufficient Conditions for Better Convergence (Main Result 2 Proofs)} \label{sec:pf_main_result2_improved_rate}

In this section, we prove the sufficient conditions on the initial data in the one-body case under which the convergence rates of the first- and second-order Trotter formulas are improved. We begin with the special case $\ell=1$. Besides being the simplest setting in which the improvement appears, this case is closely related to the many-body fermionic problem, since fermionic antisymmetry excludes the zeroth angular momentum sector in each relative coordinate. Its proof also makes clear the main mechanism behind the improved convergence rate.

We then introduce three technical ingredients used to extend the first-order estimate to general $P_{\geq\ell}$ initial data and to analyze the second-order Trotter formula. After proving the one-body results, we apply them to the two-body problem.

\subsection{\REV{The Case $\ell=1$}}
\label{sec:first_order_Pgeq1}

We first prove the first-order estimate for $P_{\geq1}$ initial data.
This special case contains the main mechanism of the general first-order argument and will also be used and generalized in the analysis of many-body fermionic systems in \cref{sec:pf_main3_many-body-fermion}.

To prepare for the proof, we first need the following two lemmas. The first is a Hardy-type inequality for the Laplace–Beltrami operator, which implies that $\tfrac{\Delta_{S^2}}{|x|^2}$ extends to a bounded operator from $H^2(\mathbb{R}^3)$ to $L^2(\mathbb{R}^3)$. We present a proof with constant $C_{SH}=22$ in \cref{sec:spherical_hardy_pf}, although this constant might not be optimal.

\begin{lemma}\label{prop: S2H2}
Let $f \in H^2(\mathbb{R}^3)$. Then
\begin{equation}\label{eq: Sr2}
    \Bigl\| \tfrac{\Delta_{S^2}}{|x|^2} f \Bigr\|
    \leq C_{SH}\, \|f\|_{H^2},
\end{equation}
where $\Delta_{S^2}$ denotes the Laplace–Beltrami operator on the unit sphere $S^2 \subset \mathbb{R}^3$ and $C_{SH} = 22$.
\end{lemma}

The second is the regularity propagation for Coulomb dynamics, which we prove in~\cref{sec:pf_lemma_coulomb_regularity_H_sigma}.

\begin{lemma}[Uniform propagation of $H^\sigma$ regularity]
\label{lem:lower-order-propagation}
Let
$H=-\Delta+\frac{c}{|x|}$ with $c\in\mathbb R$
be the self-adjoint Coulomb Hamiltonian on $L^2(\mathbb R^3)$. For every $0\leq\sigma\leq1$ there exists $C_{c,\sigma}>0$, independent of $t$, such that
\begin{equation}
    \norm{e^{- i tH}u}_{H^\sigma}
    \leq C_{c,\sigma}\norm{u}_{H^\sigma},
    \qquad t\in\mathbb R,\quad u\in H^\sigma(\mathbb R^3).
\end{equation}
\end{lemma}

We are now ready to prove the case of $\ell = 1$ of \cref{thm: 1st-Trotter}.

\begin{prop}[One-body first-order Trotter rate for $P_{\geq 1}$ data]
\label{prop:onebody-first}
Let
\[
    H=-\Delta+\frac{c}{|x|}, \qquad c\in\mathbb R\setminus\{0\},
\]
on $L^2(\mathbb R^3)$. Let $0 \leq \delta<\frac12$. Assume that $\psi_0\in H^{2+\delta}(\mathbb R^3)$ and $\psi_0=P_{\geq 1}\psi_0$. Then, for $0<t\leq 1$, the one-step first-order Trotter error
\begin{equation}
    E(t)\psi_0
    :=\left(e^{- i tV}e^{- i t(-\Delta)}-e^{- i tH}\right)\psi_0
\end{equation}
satisfies
\begin{equation}
    \norm{E(t)\psi_0}\leq C_\delta t^{(3+\delta)/2}
    \norm{\psi_0}_{H^{2+\delta}}.
\end{equation}
Consequently, for fixed final time $T>0$ and $L$ Trotter steps,
\begin{equation}
    \norm{\left(e^{- i (T/L)V}e^{- i (T/L)(-\Delta)}\right)^L\psi_0
        -e^{- i TH}\psi_0}
    \leq C_{\delta,T} L^{-(1+\delta)/2}
    \norm{\psi_0}_{H^{2+\delta}}.
\end{equation}
\end{prop}

\begin{proof}
By \cref{eq:Trotter1_error_operator_exact_rep},
\begin{equation}
    E(t)\psi_0
    = i\int_0^t e^{-i sV}
    \left[e^{-i s(-\Delta)},V\right]e^{-i(t-s)H}\psi_0\,ds .
\end{equation}
For each fixed $s\in(0,t)$, set
\begin{equation}
    f_s:=e^{-i(t-s)H}\psi_0 .
\end{equation}
It is therefore enough to prove, uniformly for $0<s\leq1$,
\begin{equation}
    \norm{\left[e^{-i s(-\Delta)},V\right]f_s}
    \leq C_\delta s^{(1+\delta)/2}\norm{\psi_0}_{H^{2+\delta}}.
\end{equation}
Since the Coulomb potential $V(x)=c/|x|$ is radial, both $-\Delta$ and $V$ commute with the angular momentum projections $P_\ell$. Hence
$H=-\Delta+V$, $e^{-i tH}$, and $e^{-i t(-\Delta)}$ commute with $P_{\geq1}$. Therefore $f_s=P_{\geq1}f_s$. By \cref{prop: S2H2},
\begin{equation}
    \norm{\frac{\Delta_{S^2}}{|x|^2}g}
    \leq C_{SH}\norm{g}_{H^2}.
\end{equation}
Thus, for every $g=P_{\geq1}g$,
\begin{equation}
    \norm{\frac{1}{|x|^2}g}
    \leq
    \norm{(-\Delta_{S^2})^{-1}P_{\geq1}}\,
    \norm{\frac{-\Delta_{S^2}}{|x|^2}g}
    \leq \frac{C_{SH}}{2}\norm{g}_{H^2}.
\end{equation}
The same Sobolev norm argument as in \cref{lem:sobolev_n-body}, applied to the one-body Coulomb Hamiltonian, gives
\begin{equation}
    \sup_{t\in\mathbb R}\norm{e^{-i tH}\psi_0}_{H^2}
    \leq C\norm{\psi_0}_{H^2}.
\end{equation}
Therefore,
\begin{equation}
    \norm{\frac{1}{|x|^2}f_s}
    \leq C C_{SH}\norm{\psi_0}_{H^2}.
\end{equation}
The improvement below uses the stronger regularity $\psi_0\in H^{2+\delta}(\mathbb{R}^3)$.
We claim that
\begin{equation}
\label{eq:weighted_H2delta}
\sup_{\tau,\tilde t\in\mathbb R}
\norm{
\frac{1}{|x|^{2+\delta}}
e^{-i\tau(-\Delta)}
e^{-i\tilde tH}\psi_0
}
\leq
C_\delta\norm{\psi_0}_{H^{2+\delta}}.
\end{equation}
We first establish the corresponding static estimate. Let $g=P_{\geq1}g\in H^{2+\delta}(\mathbb R^3)$. The following computation is first carried out for smooth functions with finite angular expansions, and the resulting estimate is then extended by approximation in $H^{2+\delta}(\mathbb{R}^3)$. We may use the identity
\begin{equation}
    g=(-\Delta_{S^2})^{-1}P_{\geq1}(-\Delta_{S^2})g
\end{equation}
before introducing any fractional derivative. Hence
\begin{equation}
    |x|^{-2-\delta}g
    =
    (-\Delta_{S^2})^{-1}P_{\geq1}|x|^{-\delta}
    \frac{-\Delta_{S^2}}{|x|^2}g .
\end{equation}
Since $\norm{(-\Delta_{S^2})^{-1}P_{\geq1}}\leq\frac12$, it remains to bound $\norm{|x|^{-\delta}\frac{\Delta_{S^2}}{|x|^2}g}$. After this angular
inversion, \cref{eq:lap_to_lap_S2} gives
\begin{equation}
    \frac{\Delta_{S^2}}{|x|^2}
    =\Delta-\partial_r^2-\frac{2}{r}\partial_r .
\end{equation}
For the pure Laplacian term, we insert $1=|p|^{-\delta}|p|^\delta$ and use the fractional Hardy-Littlewood-Sobolev bound:
\begin{equation}
    \norm{|x|^{-\delta}|p|^{-\delta}|p|^\delta\Delta g}
    \leq C_\delta\norm{g}_{H^{2+\delta}}.
\end{equation}
For the radial terms, equation (206) of~\cite{FangWuSoffer2025} gives
\begin{equation}
    \partial_r
    =
    \frac{x_1}{r}\partial_{x_1}
    +\frac{x_2}{r}\partial_{x_2}
    +\frac{x_3}{r}\partial_{x_3},
\end{equation}
where the angular coefficients are bounded by $1$. Moreover,
\begin{equation}
    \partial_r\left(\frac{x_j}{r}\right)=0,\qquad j=1,2,3,
\end{equation}
so $\partial_r$ commutes with these angular coefficients. Therefore
\begin{equation}
    \partial_r^2g
    =
    \sum_{j,k=1}^3
    \frac{x_jx_k}{r^2}\partial_{x_j}\partial_{x_k}g,
    \qquad
    \frac{2}{r}\partial_rg
    =
    2\sum_{j=1}^3
    \frac{x_j}{r^2}\partial_{x_j}g .
\end{equation}
For the second radial derivative,
\begin{equation}
\begin{aligned}
    \norm{|x|^{-\delta}\partial_r^2g}
    &\leq
    \sum_{j,k=1}^3
    \norm{|x|^{-\delta}\partial_{x_j}\partial_{x_k}g}  \\
    &\leq
    \sum_{j,k=1}^3
    \norm{|x|^{-\delta}|p|^{-\delta}}\,
    \norm{|p|^\delta\partial_{x_j}\partial_{x_k}g}
    \leq C_\delta\norm{g}_{H^{2+\delta}}.
\end{aligned}
\end{equation}
For the first radial derivative,
\begin{equation}
\begin{aligned}
    \norm{|x|^{-\delta}\frac{2}{r}\partial_rg}
    &\leq
    2\sum_{j=1}^3
    \norm{|x|^{-1-\delta}\partial_{x_j}g}  \\
    &\leq
    2\sum_{j=1}^3
    \norm{|x|^{-1-\delta}|p|^{-1-\delta}}\,
    \norm{|p|^{1+\delta}\partial_{x_j}g}
    \leq C_\delta\norm{g}_{H^{2+\delta}}.
\end{aligned}
\end{equation}
Here we used
\begin{equation}
    \norm{|x|^{-\delta}|p|^{-\delta}}<\infty,\qquad
    \norm{|x|^{-1-\delta}|p|^{-1-\delta}}<\infty
\end{equation}
for $0\leq\delta<\frac12$. Combining the three pieces gives
\begin{equation}
    \norm{|x|^{-2-\delta}g}
    \leq C_\delta\norm{g}_{H^{2+\delta}}.
\end{equation}
It remains to control $\norm{g}_{H^{2+\delta}}$ uniformly for $g=e^{-i\tilde tH}\psi_0$. We do this directly from
\begin{equation}
\begin{aligned}
    -\Delta g=Hg-\frac{c}{|x|}g .
\end{aligned}
\end{equation}
By the definition of the Sobolev norm,
\begin{equation}
\begin{aligned}
    \norm{g}_{H^{2+\delta}}
    &\leq C\bigl(\norm{g}+\norm{|p|^\delta(-\Delta)g}\bigr) \\
    &\leq C\bigl(\norm{\psi_0}
        +\norm{|p|^\delta Hg}
        +|c|\norm{|p|^\delta(|x|^{-1}g)}\bigr).
\end{aligned}
\end{equation}
We estimate the two nontrivial terms separately. Since $H$ commutes with its flow,
\begin{equation}
    Hg=e^{-i\tilde tH}H\psi_0 .
\end{equation}
Applying \cref{lem:lower-order-propagation} with $\sigma=\delta$ gives
\begin{equation}
    \norm{|p|^\delta Hg}
    \leq \norm{Hg}_{H^\delta}
    \leq C_\delta\norm{H\psi_0}_{H^\delta}.
\end{equation}
Using $H\psi_0=-\Delta\psi_0+c|x|^{-1}\psi_0$ and the estimate for $|x|^{-1}\psi_0$ proved below, we get
\begin{equation}
\begin{aligned}
    \norm{H\psi_0}_{H^\delta}
    &\leq \norm{\psi_0}_{H^{2+\delta}}
      + |c|\,\norm{|x|^{-1}\psi_0}_{H^\delta}  \\
    &\leq C_\delta\norm{\psi_0}_{H^{2+\delta}}.
\end{aligned}
\end{equation}
The Coulomb term in the evolution is handled in the same way. Since $0\leq\delta<\frac12$, the pointwise inequality
$|\xi|^\delta\leq\langle\xi\rangle$ shows that the constant in the following estimate may be chosen independently of $\delta$:
\begin{equation}
\begin{aligned}
    \norm{|p|^\delta(|x|^{-1}g)}
    &\leq C\norm{|x|^{-1}g}_{H^1}  \\
    &\leq C\left(
        \norm{|x|^{-1}g}
        +\norm{\nabla(|x|^{-1}g)}
    \right).
\end{aligned}
\end{equation}
For the first term,
\begin{equation}
    \norm{|x|^{-1}g}
    \leq \norm{|x|^{-1}|p|^{-1}}\norm{|p|g}
    \leq C\norm{\langle p\rangle^{1+\delta}g}.
\end{equation}
For the derivative term,
\begin{equation}
    \partial_{x_j}(|x|^{-1}g)
    =
    -\frac{x_j}{|x|^3}g+|x|^{-1}\partial_{x_j}g,
\end{equation}
and hence
\begin{equation}
\begin{aligned}
    \norm{\nabla(|x|^{-1}g)}
    &\leq C\norm{|x|^{-2}g}
       + C\sum_{j=1}^3\norm{|x|^{-1}\partial_{x_j}g}  \\
    &\leq C\norm{|x|^{-2}g}
       + C\sum_{j=1}^3
       \norm{|x|^{-1}|p|^{-1}}\norm{|p|\partial_{x_j}g} \\
    &\leq C\norm{|x|^{-2}g}
       + C\norm{g}_{H^2}.
\end{aligned}
\end{equation}
Thus
\begin{equation}
    \norm{|p|^\delta(|x|^{-1}g)}
    \leq C\left(
        \norm{|x|^{-2}g}
        +\norm{g}_{H^2}
    \right).
\end{equation}
The first term is controlled by the $P_{\geq1}$ estimate above, and the second term is controlled by the same norm estimate, so
\begin{equation}
    \norm{|p|^\delta(|x|^{-1}g)}
    \leq C\norm{\psi_0}_{H^{2+\delta}}.
\end{equation}
Combining the estimates gives
\begin{equation}
    \sup_{\tilde t\in\mathbb R}
    \norm{e^{-i\tilde tH}\psi_0}_{H^{2+\delta}}
    \leq C_\delta\norm{\psi_0}_{H^{2+\delta}}.
\end{equation}
Therefore
\begin{equation}
    \sup_{\tilde t\in\mathbb R}
    \norm{|x|^{-2-\delta}e^{-i\tilde tH}\psi_0}
    \leq C_\delta\norm{\psi_0}_{H^{2+\delta}}.
\end{equation}
For $\tau\neq0$, the same argument is applied to $g=e^{-i \tau(-\Delta)}h$, where $h=e^{-i\tilde tH}\psi_0$. The operators $\Delta$, $\partial_{x_j}$, and $|p|^\alpha$ commute with the free propagator, so, for example,
\begin{equation}
    |p|^\delta\Delta g
    =
    e^{-i \tau(-\Delta)}|p|^\delta\Delta h,
\end{equation}
and similarly for the Cartesian first and second derivatives. Thus the same insertions
\begin{equation}
    1=|p|^{-\delta}|p|^\delta,\qquad
    1=|p|^{-1-\delta}|p|^{1+\delta}
\end{equation}
reduce all terms to the corresponding estimates for $h$. For the highest derivative of $h$, we again use
\begin{equation}
    -\Delta h=Hh-\frac{c}{|x|}h .
\end{equation}
Since the free evolution is unitary on every $H^s(\mathbb{R}^3)$, the same bound holds uniformly in $\tau$, proving \cref{eq:weighted_H2delta}.

We now apply the step-size–dependent smooth cutoff technique introduced in~\cite{FangWuSoffer2025}. In particular, we introduce a smooth cutoff decomposition of the potential that depends on the short-time Trotter step size $s \in (0,1]$:
\begin{equation}
    V(x) \;=\; V_{\mathrm{reg}}(x,s) + V_{\mathrm{sin}}(x,s),
\end{equation}
where
\begin{equation}\label{eq: def Vreg}
    V_{\mathrm{reg}}(x,s) \;:=\; F\!\left( \frac{|x|}{s^\beta} > 1 \right) V(x),
\end{equation}
and
\begin{equation}\label{eq: def Vsin}
    V_{\mathrm{sin}}(x,s) \;:=\; F\!\left( \frac{|x|}{s^\beta} \leq 1 \right) V(x).
\end{equation}
Here $\beta >0$ will be chosen later and $F$ is any smooth cutoff defined by
$ F(\cdot \leq 1) $ and $F(\cdot > 1) := 1 - F(\cdot \leq 1)$, such that
\begin{equation}
    F(\lambda \leq 1) = 
    \begin{cases}
        1 & \text{for } \lambda \leq \tfrac{1}{2}, \\
        0 & \text{for } \lambda \geq 1.
    \end{cases}
\end{equation}
It is convenient to observe that
\[
F(\lambda>1)\le \chi(\lambda>1/2),
\]
where $\chi(z\in I)$ denotes the indicator function of the interval $I$.

The choice of this smooth cutoff function $F$ is not unique, and affects only the absolute constants in the estimate. To make things concrete, we choose the same cutoff function $F$ as~\cite[Eqs.~(76)–(77)]{FangWuSoffer2025}:
\begin{equation} \label{eq:F_particular}
    F(\lambda \leq 1) = 
\begin{cases}
1 & \lambda \leq 1/2 \\
C_0\int_{\lambda}^1 e^{-\frac{1}{(r-1/2)(1-r)}}dr
& \lambda \in (1/2, 1) \\
0 & \lambda \geq 1
\end{cases}
\end{equation}
with the normalization constant
\begin{equation}
    C_0:=\dfrac{1}{\int_{\frac{1}{2}}^1 e^{-\frac{1}{(r-1/2)(1-r)}}dr}. 
\end{equation}
For the regular part, since
\begin{equation}
    [-\Delta,V_{\mathrm{reg}}]g
    =-(\Delta V_{\mathrm{reg}})g
    -2\nabla V_{\mathrm{reg}}\cdot\nabla g,
\end{equation}
the argument in Eqs. (26)-(30) and Step 2 of
\cite{FangWuSoffer2025} gives
\begin{equation}
    \norm{[-\Delta,V_{\mathrm{reg}}(x,s)]g}
    \leq \frac{C_\delta}{s^{(1-\delta)\beta}}
    \left(
    \norm{\frac{1}{|x|^{2+\delta}}g}
    +
    \sum_{j=1}^3
    \norm{\frac{1}{|x|^{1+\delta}}\partial_{x_j}g}
    \right).
\end{equation}
Here the fractional Hardy inequality gives
\begin{equation}
    \norm{|x|^{-1-\delta}\partial_{x_j}g}
    \leq C_\delta\norm{g}_{H^{2+\delta}}.
\end{equation}
Using the variation of constants formula for the commutator,
\begin{equation}
\begin{aligned}
    \norm{[e^{-i s(-\Delta)},V_{\mathrm{reg}}(x,s)]f_s}
    &\leq
    \int_0^s
    \norm{[-\Delta,V_{\mathrm{reg}}(x,s)]
        e^{-i u(-\Delta)}f_s}\,du  \\
    &\leq
    C_\delta s^{1-(1-\delta)\beta}\norm{\psi_0}_{H^{2+\delta}},
\end{aligned}
\end{equation}
where we used \cref{eq:weighted_H2delta} and the uniform
$H^{2+\delta}$ norm estimate above.

For the singular part, the support property of
$V_{\mathrm{sin}}(x,s)$ similarly gives
\begin{equation}
\begin{aligned}
    \norm{[e^{-i s(-\Delta)},V_{\mathrm{sin}}(x,s)]f_s}
    &\leq C s^{(1+\delta)\beta}
    \left(
    \norm{\frac{1}{|x|^{2+\delta}}f_s}
    +
    \norm{\frac{1}{|x|^{2+\delta}}e^{-i s(-\Delta)}f_s}
    \right) \\
    &\leq C_\delta s^{(1+\delta)\beta}\norm{\psi_0}_{H^{2+\delta}}.
\end{aligned}
\end{equation}
Combining the two bounds, we obtain
\begin{equation}
    \norm{[e^{-i s(-\Delta)},V]f_s}
    \leq C_\delta\left(s^{1-(1-\delta)\beta}
    +s^{(1+\delta)\beta}\right)\norm{\psi_0}_{H^{2+\delta}}.
\end{equation}
Balancing the two powers gives $\beta=\frac12$, hence
\begin{equation}
    \norm{[e^{-i s(-\Delta)},V]f_s}
    \leq C_\delta s^{(1+\delta)/2}\norm{\psi_0}_{H^{2+\delta}}.
\end{equation}
Substitution into the local error representation gives
\begin{equation}
    \norm{E(t)\psi_0}
    \leq C_\delta\int_0^t s^{(1+\delta)/2}\,ds\,
    \norm{\psi_0}_{H^{2+\delta}}
    \leq C_\delta t^{(3+\delta)/2}\norm{\psi_0}_{H^{2+\delta}}.
\end{equation}
For the global estimate, the bound in
\cref{eq:long_term_error_by_lte_trotter1}
reduces the problem to applying the one-step estimate to
$e^{-iktH}\psi_0$ for $0\leq k\leq L-1$. These states remain in
$\operatorname{Ran}(P_{\geq1})$, and their $H^{2+\delta}$ norms are
uniformly bounded by the propagation estimate above. Summing the local
bound over $L$ steps with $t=T/L\leq1$ gives the global
$L^{-(1+\delta)/2}$ convergence rate. If $T/L>1$, the same conclusion
follows from unitarity after enlarging $C_{\delta,T}$.
\end{proof}

\subsection{Three New Technical Ingredients}\label{sec:improved_rate_3_ingredients}

The preceding proof treats the case $\ell=1$ directly. To extend this argument to general $P_{\geq\ell}$ initial data, there are three important technical ingredients we proved and used in the proofs of our main result 2.

The first and most important ingredient is the following key observation, a property of the Coulomb system that may be of independent interest. We defer its proof to~\cref{sec: key} to avoid interrupting the proof of the main results.

\begin{prop}[Key Observation]
\label{prop:Pgel-propagation}
Let $0\leq\delta<\frac12$, $\ell\geq1$, and let
$\psi_0\in H^{1+\ell+\delta}(\mathbb R^3)$ satisfy
$\psi_0=P_{\geq\ell}\psi_0$. Then
\begin{equation}\label{eq:Pgel-propagation}
    \sup_{\tau,\tilde t\in\mathbb R}
    \norm{e^{- i \tau(-\Delta)}e^{- i \tilde tH}\psi_0}_{H^{1+\ell+\delta}}
    \leq C_{\ell,\delta}\norm{\psi_0}_{H^{1+\ell+\delta}} .
\end{equation}
Moreover,
\begin{equation}\label{eq:Pgel-weighted-dynamic}
    \sup_{\tau,\tilde t\in\mathbb R}
    \norm{\frac{1}{|x|^{1+\ell+\delta}}
      e^{- i \tau(-\Delta)}e^{- i \tilde tH}\psi_0}
    \leq C_{\ell,\delta}\norm{\psi_0}_{H^{1+\ell+\delta}} .
\end{equation}

\end{prop}

The first estimate shows that the $H^{1+\ell+\delta}$ regularity is preserved under both the exact Coulomb evolution and the free evolution. Combining this estimate with the result below (\cref{prop:Pgel-weighted}) gives the second estimate, which provides the weighted control near the Coulomb singularity needed in the local error analysis.

The second ingredient is an angular Hardy-type estimate. Its proof uses the spectral gap of $-\Delta_{S^2}$ on $\operatorname{Ran}(P_{\geq\ell})$, together with the fractional Hardy inequality. The estimate is also proved in \cref{sec: key}.

\begin{prop}[Weighted estimate for $P_{\geq\ell}$ data]
\label{prop:Pgel-weighted}
Let $0\leq\delta<\frac12$ and $\ell\geq1$. If $g\in H^{1+\ell+\delta}(\mathbb R^3)$ satisfies $g=P_{\geq\ell}g$, then
\begin{equation}\label{eq:Pgel-weighted}
    \norm{\frac{1}{|x|^{1+\ell+\delta}}g}
    \leq C_{\ell,\delta}\norm{g}_{H^{1+\ell+\delta}},
\end{equation}
and the constants are uniformly bounded: $\sup_{\ell\geq1}C_{\ell,\delta} =:C_\delta^\ast<\infty$.
\end{prop}

This estimate converts the exclusion of the angular momentum sectors below $\ell$ into additional weighted control near the Coulomb singularity. Combining \cref{prop:Pgel-weighted} with the propagation of the ordinary Sobolev norm gives \cref{eq:Pgel-weighted-dynamic}. In particular, this estimate provides the singular weights needed in the commutator bounds without requiring the propagation of a weighted Sobolev norm.

The third ingredient, used only in the proof of the improved second-order Trotter rate (\cref{thm: 2-rd Trotter}), is an (alternative) exact local error representation for the second-order Trotter formula (i.e. Strang splitting). This representation holds formally for general operators $\mathcal{L} = \mathcal{L}_1 + \mathcal{L}_2$ that are not necessarily anti-Hermitian (or anti-self-adjoint). For general unbounded operators, of course, one needs to carefully check the domain of both sides and interpret the identity on admissible functions in their common domain. When applying this representation to the Coulomb case (with $\mathcal{L}_1 = -iA = -i(-\Delta)$ and $\mathcal{L}_2 = -iB = -i V$), we can make sense of the terms, as the error operator will be acting on the solution states that satisfy the property~\cref{prop:Pgel-propagation}.  Its proof is straightforward and is given in \cref{sec:pf_improved_rate_Trotter2}.

\begin{thm}[Exact Local Error Representation for Trotter2]\label{thm:exact_err_rep_Trotter2_general}
    Let $\mathcal{L} = \mathcal{L}_1 + \mathcal{L}_2$. For every admissible $f$, the Strang splitting has the following exact error representation
    \begin{equation}
    \begin{aligned}
        & \left( e^{\mathcal{L}_1 t/2} e^{\mathcal{L}_2 t}e^{\mathcal{L}_1 t/2}  - e^{\mathcal{L}t}\right) f
        \\
        = &  
            \frac{1}{2} \int_0^t \, ds \, 
    \int_0^s \, du \,  
  e^{\mathcal{L}_1 s/2}  
  e^{\mathcal{L}_2 (s-u)} 
    \left[ e^{\mathcal{L}_2 u}e^{\mathcal{L}_1 (s-u)/2}  ,\left[ \mathcal{L}_2,  \mathcal{L}_1 \right]\right] e^{\mathcal{L}_1 u/2}
  e^{\mathcal{L} (t-s)} f.
\end{aligned}
    \end{equation}

\end{thm}

In the presence of the Coulomb potential, it is crucial to derive an error representation in which the unitary evolutions generated by $H$ and $-\Delta$ appear on the right-hand side, thereby deferring the appearance of $e^{-iVt}$ as much as possible. More specifically, the unitary generated by the Coulomb interaction $V$ does not preserve $H^2\REV{(\mathbb{R}^{3N})}$ (the domain of the Hamiltonian operator); see \cite[Section 2.2]{FangWuSoffer2025} for a detailed discussion.
To illustrate this point, consider for simplicity a one-body model with $V(x) = |x|^{-1}$ near $x = 0$, and take $\psi \in C_c^\infty(\mathbb{R}^3) \subset H^2(\mathbb{R}^3)$ with $\psi(0) \neq 0$. A direct computation shows that derivatives of $e^{-iVs}\psi$ involve terms of the form $|x|^{-3}\psi$, which are not square-integrable. Consequently, $e^{-iVs}\psi \notin H^2\REV{(\mathbb{R}^3)}$.

By contrast, in finite-dimensional or bounded-operator settings, the ordering of unitaries in the error representation is largely immaterial. For example, in~\cite{ChildsSuTranEtAl2020}, one may place $e^{-iBs}$ on the right-hand side, yet different representations lead to the same commutator-based error scaling. Indeed, one may expand commutators using identities such as
\begin{equation}
\left[ e^{\mathcal{L}_2 u}e^{\mathcal{L}_1 (s-u)/2}  ,\left[ \mathcal{L}_2,  \mathcal{L}_1 \right]\right]
= 
e^{\mathcal{L}_2 u} \left[ e^{\mathcal{L}_1 (s-u)/2}  ,\left[ \mathcal{L}_2,  \mathcal{L}_1 \right]\right]
+
\left[ e^{\mathcal{L}_2 u}  ,\left[ \mathcal{L}_2,  \mathcal{L}_1 \right]\right] e^{\mathcal{L}_1 (s-u)/2}
    ,
\end{equation}
together with
\begin{equation}
 [e^{s\mc{L}_2}, \mc{L}_1] =  \int_0^s e^{(s-\tau) \mc{L}_2} [\mc{L}_2, \mc{L}_1] e^{ \tau \mc{L}_2} \, d\tau,
\end{equation}
which allow one to rewrite the error in different but equivalent forms, ultimately yielding the well-known commutator scaling in terms of the Hamiltonian components.

However, in the presence of unbounded operators such flexibility breaks down. While $e^{-iAs}$ preserves $H^2\REV{(\mathbb{R}^{3N})}$, the unitary $e^{-iBs}$ associated with the Coulomb potential does not map $H^2\REV{(\mathbb{R}^{3N})}$ into itself; in particular, one may view $\|e^{-iBs}\|_{H^2 \to H^2} = \infty$. As a result, the precise ordering of operators in the error representation becomes essential, since otherwise the remaining terms cannot be controlled within the $H^2\REV{(\mathbb{R}^{3N})}$ framework. \REV{In the rest of this section, we consider $N = 1$ case of this formula.}

We also note that exact error representations constitute a fundamental tool in numerical analysis, and have more recently played an important role in the analysis of quantum algorithms (see, e.g., \cite{DescombesThalhammer2010,Lubich2008book,BlanesCasas2017book,LasserLubich2020,ChildsSuTranEtAl2020,AnFangLin2021,FangLiuSarkar2025,FangLiuZhu2025,BeckerGalkeSalzmannLuijk2024}).

\subsection{Proof of~\cref{thm: 1st-Trotter}}\label{sec:pf_improved_rate_Trotter1}

We now extend the preceding first-order argument to general $P_{\geq\ell}$ initial data. We first derive the intermediate weighted estimates needed in the cutoff argument.

\begin{lemma}[All intermediate weights]
\label{lem:interpolated-weights}
Let $0\leq\delta<\frac12$ and $\ell\geq1$. If $g=P_{\geq\ell}g\in H^{1+\ell+\delta}(\mathbb{R}^3)$, then, for every $0\leq a\leq1+\ell+\delta$,
\begin{equation}
\label{eq:interpolated-weights}
\norm{\frac{g}{|x|^a}}
\leq
\norm{g}^{1-\frac{a}{1+\ell+\delta}}
\norm{\frac{g}{|x|^{1+\ell+\delta}}}^{\frac{a}{1+\ell+\delta}}
\leq
C_{\ell,\delta}\norm{g}_{H^{1+\ell+\delta}}.
\end{equation}
\end{lemma}

\begin{proof}
The two endpoint cases are immediate. For $0<a<1+\ell+\delta$, let $\theta=a/(1+\ell+\delta)$. We write
\begin{equation}
|x|^{-2a}|g|^2
=
\left(|g|^2\right)^{1-\theta}
\left(
|x|^{-2(1+\ell+\delta)}|g|^2
\right)^\theta.
\end{equation}
The first inequality in \cref{eq:interpolated-weights} follows from H\"older's inequality, while the second follows from
\cref{prop:Pgel-weighted}.
\end{proof}

Combining \cref{lem:interpolated-weights} with \cref{prop:Pgel-propagation}, and using the fact that both evolutions preserve $\operatorname{Ran}(P_{\geq\ell})$, we obtain
\begin{equation}
\label{eq:intermediate-weight-propagation}
\sup_{\tau,\tilde t\in\mathbb{R}}
\norm{
\frac{1}{|x|^a}
e^{-i\tau(-\Delta)}e^{-i\tilde tH}\psi_0
}
\leq
C_{\ell,\delta}\norm{\psi_0}_{H^{1+\ell+\delta}},
\qquad
0\leq a\leq1+\ell+\delta.
\end{equation}

For the regular part of the commutator, we also need an estimate with one derivative. Recall the angular-momentum relation
\begin{equation}
\partial_{x_j}P_{\geq\ell}
=
P_{\geq\ell-1}\partial_{x_j}P_{\geq\ell},
\end{equation}
which is proved in \cref{lem:derivative-lowers-one}. If $\ell\geq2$, applying \cref{lem:interpolated-weights} to $\partial_{x_j}g$ with angular index $\ell-1$ and weight $a=2$, and then using
\cref{prop:Pgel-propagation}, gives
\begin{equation}
\label{eq:weighted-derivative-propagation}
\sup_{\tau,\tilde t\in\mathbb{R}}
\sum_{j=1}^3
\norm{
\frac{1}{|x|^2}\partial_{x_j}
e^{-i\tau(-\Delta)}e^{-i\tilde tH}\psi_0
}
\leq
C_{\ell,\delta}\norm{\psi_0}_{H^{1+\ell+\delta}}.
\end{equation}

We next record the regular and singular cutoff estimates in the form needed below.

\begin{lemma}[Cutoff commutator bounds with general weight]
\label{lem:general-weight-commutator}
Let $0<s\leq1$ and $0<\beta\leq1$, and let $V=V_{\mathrm{reg}}(x,s)+V_{\mathrm{sin}}(x,s)$ be the decomposition in \cref{eq: def Vreg,eq: def Vsin}. Then, for $1\leq a\leq3$,
\begin{equation}
\label{eq:reg-general}
\norm{[-\Delta,V_{\mathrm{reg}}(x,s)]h}
\leq
\frac{C}{s^{(3-a)\beta}}
\left(
\norm{\frac{h}{|x|^a}}
+
\sum_{j=1}^3
\norm{\frac{\partial_{x_j}h}{|x|^{a-1}}}
\right).
\end{equation}
Moreover, for every $a\geq1$,
\begin{equation}
\label{eq:sin-general}
\norm{[e^{-is(-\Delta)},V_{\mathrm{sin}}(x,s)]h}
\leq
Cs^{(a-1)\beta}
\sup_{u\in\mathbb{R}}
\norm{\frac{1}{|x|^a}e^{-iu(-\Delta)}h}.
\end{equation}
\end{lemma}

\begin{proof}
The cutoff bounds give
\begin{equation}
|\Delta V_{\mathrm{reg}}(x,s)|
\leq
C\chi_{\{|x|\geq s^\beta/2\}}|x|^{-3},
\qquad
|\nabla V_{\mathrm{reg}}(x,s)|
\leq
C\chi_{\{|x|\geq s^\beta/2\}}|x|^{-2}.
\end{equation}
On this region, for $1\leq a\leq3$,
\begin{equation}
|x|^{-3}
\leq
Cs^{-(3-a)\beta}|x|^{-a},
\qquad
|x|^{-2}
\leq
Cs^{-(3-a)\beta}|x|^{-(a-1)}.
\end{equation}
Using
\begin{equation}
[-\Delta,V_{\mathrm{reg}}]h
=
-(\Delta V_{\mathrm{reg}})h
-2\nabla V_{\mathrm{reg}}\cdot\nabla h,
\end{equation}
we obtain \cref{eq:reg-general}.

On the support of $V_{\mathrm{sin}}(x,s)$, we have
\begin{equation}
\frac{1}{|x|}
\leq
Cs^{(a-1)\beta}\frac{1}{|x|^a}.
\end{equation}
Expanding the commutator and using the unitarity of the free evolution, we obtain
\begin{equation}
\begin{aligned}
\norm{[e^{-is(-\Delta)},V_{\mathrm{sin}}(x,s)]h}
&\leq
\norm{V_{\mathrm{sin}}(x,s)h}
+
\norm{V_{\mathrm{sin}}(x,s)e^{-is(-\Delta)}h} \\
&\leq
Cs^{(a-1)\beta}
\sup_{u\in\mathbb{R}}
\norm{\frac{1}{|x|^a}e^{-iu(-\Delta)}h},
\end{aligned}
\end{equation}
which proves \cref{eq:sin-general}.
\end{proof}

We are now ready to prove \cref{thm: 1st-Trotter}.
\begin{proof}[Proof of~\cref{thm: 1st-Trotter}]
For $\ell=1$, we apply the one-step estimate in \cref{prop:onebody-first} to $e^{-iktH}\psi_0$ at the $k$th time step. These states remain in $\operatorname{Ran}(P_{\geq1})$, and their $H^{2+\delta}$ norms are uniformly bounded by the propagation estimate proved there. Therefore, \cref{eq:long_term_error_by_lte_trotter1} gives the rate $Tt^{(1+\delta)/2}$. It remains to consider $\ell\geq2$.

Recall the one-step error operator $E(t)$ defined in \cref{prop:onebody-first}. For $0<t\leq1$, let
\begin{equation}
f_s
:=
e^{-i(t-s)H}\psi_0.
\end{equation}
By \cref{eq:Trotter1_error_operator_exact_rep},
\begin{equation}
E(t)\psi_0
=
i\int_0^t e^{-isV}
\left[e^{-is(-\Delta)},V\right]f_s\,ds.
\end{equation}
Thus, it suffices to prove
\begin{equation}
\label{eq:target-commutator}
\norm{\left[e^{-is(-\Delta)},V\right]f_s}
\leq
C_{\ell,\delta}s\norm{\psi_0}_{H^{1+\ell+\delta}},
\qquad
0<s\leq1.
\end{equation}
We use the cutoff decomposition in \cref{eq: def Vreg,eq: def Vsin} with
\begin{equation}
\beta
:=
\frac{1}{\ell+\delta}
\in
\left(0,\frac12\right].
\end{equation}

For the regular part, the variation of constants formula and \cref{eq:reg-general} with $a=3$ give
\begin{equation}
\begin{aligned}
& \norm{[e^{-is(-\Delta)},V_{\mathrm{reg}}(x,s)]f_s}
\\
\leq &
\int_0^s
\norm{[-\Delta,V_{\mathrm{reg}}(x,s)]
e^{-iu(-\Delta)}f_s}\,du 
\leq
C_{\ell,\delta}s\norm{\psi_0}_{H^{1+\ell+\delta}},
\end{aligned}
\end{equation}
where in the second inequality, we used \cref{eq:intermediate-weight-propagation,eq:weighted-derivative-propagation}.

For the singular part, applying \cref{eq:sin-general} with $a=1+\ell+\delta$ gives
\begin{equation}
\norm{[e^{-is(-\Delta)},V_{\mathrm{sin}}(x,s)]f_s}
\leq
Cs^{(\ell+\delta)\beta}
\sup_{u\in\mathbb{R}}
\norm{
\frac{1}{|x|^{1+\ell+\delta}}
e^{-iu(-\Delta)}f_s
}\leq
C_{\ell,\delta}s\norm{\psi_0}_{H^{1+\ell+\delta}},
\end{equation}
where we used $(\ell+\delta)\beta=1$ and \cref{eq:intermediate-weight-propagation}. Combining the regular and singular estimates proves \cref{eq:target-commutator}. Substituting this estimate into the local error representation gives
\begin{equation}
\label{eq:first-order-local-Pgeq-ell}
\norm{E(t)\psi_0}
\leq
C_{\ell,\delta}t^2\norm{\psi_0}_{H^{1+\ell+\delta}}.
\end{equation}

For the global estimate, we apply \cref{eq:first-order-local-Pgeq-ell} with the initial state replaced by
$e^{-iktH}\psi_0$. Since the exact evolution preserves $\operatorname{Ran}(P_{\geq\ell})$, and \cref{prop:Pgel-propagation} gives
\begin{equation}
\norm{e^{-iktH}\psi_0}_{H^{1+\ell+\delta}}
\leq
C_{\ell,\delta}\norm{\psi_0}_{H^{1+\ell+\delta}},
\end{equation}
the telescoping estimate in \cref{eq:long_term_error_by_lte_trotter1} yields
\begin{equation}\label{eq:first-order-global-Pgeq-ell}
\begin{aligned}
\norm{
\left(
e^{-iHT}
-
\left(e^{-iBt}e^{-iAt}\right)^L
\right)\psi_0
}
&\leq
\sum_{k=0}^{L-1}
\norm{E(t)e^{-iktH}\psi_0} \\
&\leq
C_{\ell,\delta}Lt^2
\norm{\psi_0}_{H^{1+\ell+\delta}} =
C_{\ell,\delta}Tt
\norm{\psi_0}_{H^{1+\ell+\delta}}.
\end{aligned}
\end{equation}
Together with the case $\ell=1$ and the definition of $\gamma_1(\ell,\delta)$ in \cref{eq:def_gamma_1}, this proves the desired estimate stated in \cref{eq:first_order_angular_rate}.
\end{proof}

\begin{rem}[The gain saturates at $\ell=2$]
\label{rem:saturation}
Since $\operatorname{Ran}(P_{\geq\ell}) \subset\operatorname{Ran}(P_{\geq2})$ and $H^{1+\ell+\delta}(\mathbb{R}^3)\subset H^{3+\delta}(\mathbb{R}^3)$ for $\ell\geq2$, the binding case of \cref{thm: 1st-Trotter} is $\ell=2$:
\begin{equation}
\psi_0=P_{\geq2}\psi_0\in H^{3+\delta}(\mathbb{R}^3)
\quad\Longrightarrow\quad
\norm{E(t)\psi_0}
\leq
C_\delta t^2\norm{\psi_0}_{H^{3+\delta}}.
\end{equation}
Higher $\ell$ yields no further improvement in the convergence exponent. This is not merely an artifact of the proof: for initial
data for which $[-\Delta,V]\psi_0$ is well defined and nonzero, the leading term in the first-order local error is of order $t^2$.
Therefore, the exponent $2$ in \cref{eq:first-order-local-Pgeq-ell} and the first-order $t$ rate in
\cref{eq:first-order-global-Pgeq-ell} cannot in general be improved within this splitting. Thus the angular ladder in \cref{prop:Pgel-weighted} converts the Coulomb singularity from an obstruction into a non-issue at $\ell=2$: the almost-optimal exponent $(1+\delta)/2$ in \cref{prop:onebody-first}, which approaches $3/4$ as $\delta\to(1/2)^-$, is a genuine $\ell=1$ phenomenon, while one additional unit of angular momentum restores the full first-order rate.
\end{rem}

\subsection{Proofs of the First Two Technical Ingredients}
\label{sec: key}

We first establish the angular decomposition and induction argument needed for both results. We then prove the weighted estimate in \cref{prop:Pgel-weighted}, and use it to prove the Sobolev propagation result in \cref{prop:Pgel-propagation}.

Throughout this section we use the angular decomposition of $L^2(\mathbb R^3)$: for $\ell'\in\mathbb N_0$ let $P_{\ell'}$ be the orthogonal projection onto the sector spanned by the degree-$\ell'$ spherical harmonics (with arbitrary radial profile), and set
\[
    P_{\geq\ell}:=\sum_{\ell'\geq\ell}P_{\ell'},
    \qquad
    P_{\geq\ell}:=I \text{ for } \ell\leq0 .
\]
Every radial multiplier $|x|^{-a}$, as well as $\Delta$ and $-\Delta_{S^2}$, commutes with each $P_{\ell'}$, and
\begin{equation}\label{eq:angular-inverse-norm}
    \norm{(-\Delta_{S^2})^{-1}P_{\geq\ell}}
    =\frac{1}{\ell(\ell+1)},\qquad \ell\geq1,
\end{equation}
since $-\Delta_{S^2}$ has eigenvalue $\ell'(\ell'+1)$ on $\operatorname{Ran} P_{\ell'}$.

\begin{lemma}[A derivative lowers the angular index by at most one]
\label{lem:derivative-lowers-one}
For $j=1,2,3$ and every $\ell\in\mathbb Z$,
\[
    \partial_{x_j}P_{\geq\ell}
    =P_{\geq\ell-1}\,\partial_{x_j}P_{\geq\ell}
\]
as operators on $H^1(\mathbb R^3)$. Consequently, if $g=P_{\geq\ell}g$ then $\partial_{x_j}g=P_{\geq\ell-1}\partial_{x_j}g$; and if moreover $g\in H^2(\mathbb R^3)$, then $\partial_{x_j}\partial_{x_k}g=P_{\geq\ell-2}\partial_{x_j}\partial_{x_k}g$.
\end{lemma}

\begin{proof}
Both sides are bounded $H^1(\mathbb R^3)\to L^2(\mathbb R^3)$, so it suffices to argue on $C_c^\infty(\mathbb R^3\setminus\{0\})$, which is dense in $H^1(\mathbb R^3)$ because a point has vanishing $H^1$-capacity in dimension three.

We use the angular momentum operators
\begin{equation}\label{eq:Omega-def}
    \Omega_{jk}:=x_j\partial_{x_k}-x_k\partial_{x_j},
    \qquad 1\leq j,k\leq3 ,
\end{equation}
and the radial derivative $\partial_r=\sum_{m}\frac{x_m}{r}\partial_{x_m}$ of equation (206) of~\cite{FangWuSoffer2025}. Contracting the elementary identity $x_k\Omega_{jk}=x_jx_k\partial_{x_k}-x_k^2\partial_{x_j}$ over $k$ and using $\sum_kx_k\partial_{x_k}=r\partial_r$ gives the decomposition of a Cartesian derivative into a radial and a purely angular part,
\begin{equation}\label{eq:cartesian-split}
    \partial_{x_j}
    =\frac{x_j}{r}\,\partial_r
     -\frac{1}{r}\sum_{k=1}^3\frac{x_k}{r}\,\Omega_{jk},
     \qquad j=1,2,3 .
\end{equation}
This is the linear expression for $\partial_{x_j}$ dual to the one used for $\partial_r$ in the proof of \cref{prop:onebody-first}, and it is all we shall need.

The two structural facts about \eqref{eq:cartesian-split} are the following.
First, each $\Omega_{jk}$ is tangential, $\Omega_{jk}r=0$, and hence commutes with every radial multiplier, in particular with $1/r$. Moreover, $\partial_r$ is rotation invariant, so $[\Omega_{jk},\partial_r]=0$. Notice that $x_m/r$ is an angular multiplier and does not in general commute with $\Omega_{jk}$; no such commutation is used below. Second, the Casimir relation
\begin{equation}\label{eq:casimir}
    \sum_{1\leq j<k\leq3}\Omega_{jk}^2=\Delta_{S^2}
\end{equation}
holds, where $\Delta_{S^2}$ is the Laplace-Beltrami operator on $S^2$, acting in the angular variable only. Consequently $\Omega_{jk}$ commutes with $\Delta_{S^2}$, so it preserves each eigenspace of $-\Delta_{S^2}$ and therefore
\begin{equation}\label{eq:Omega-preserves-sector}
    \Omega_{jk}P_{\ell'}=P_{\ell'}\,\Omega_{jk}
    \qquad\text{for every }\ell' .
\end{equation}
In words: the angular momentum operators do not move the angular index at all.

Now let $u=P_{\ell'}u$. In \eqref{eq:cartesian-split} the coefficients $x_j/r$ and $x_k/r$ are the degree-one spherical harmonics $\omega_j$,
multiplication by which shifts the angular index by exactly one,
\begin{equation}\label{eq:omega-shift}
    \omega_jP_{\ell'}=(P_{\ell'-1}+P_{\ell'+1})\,\omega_jP_{\ell'} ,
\end{equation}
this being the classical selection rule for the product of a degree-one harmonic with a degree-$\ell'$ one. Both $\partial_r$ and, by \eqref{eq:Omega-preserves-sector}, the operators $\frac1r\Omega_{jk}$ leave the sector $\ell'$ unchanged.
Hence, in each term of \cref{eq:cartesian-split}, the entire angular shift is carried by its single degree-one multiplier $x_j/r$ or $x_k/r$, and \cref{eq:omega-shift} gives
\[
    \partial_{x_j}P_{\ell'}
    =(P_{\ell'-1}+P_{\ell'+1})\,\partial_{x_j}P_{\ell'} .
\]
Summing over $\ell'\geq\ell$ and using $\bigcup_{\ell'\geq\ell}\{\ell'-1,\ell'+1\}\subset\{m:m\geq\ell-1\}$ yields the stated identity.

The first consequence is immediate. For the second, $g\in H^2 (\mathbb{R}^3)$ guarantees $\partial_{x_k}g\in H^1(\mathbb{R}^3)$, so the identity may be applied a second time, giving $\partial_{x_j}\partial_{x_k}g\in\operatorname{Ran}P_{\geq\ell-2}$.
\end{proof}

\begin{rem}
\eqref{eq:cartesian-split} also explains the bookkeeping in the induction of \cref{prop:Pgel-weighted}: the angular index drops by one per Cartesian derivative, and it does so \emph{only} through the coefficient $\omega_j$, never through $\partial_r$ or $\Omega_{jk}$. This is the precise sense in which ``$\partial_r$ commutes with the angular coefficients'' is used there, and it is why the radial term $\frac2r\partial_rg$ sits at level $\ell-1$ while $\partial_r^2g$ sits at level $\ell-2$.
\end{rem}

\begin{lemma}[Hardy base case]\label{lem:hardy-base}
For $0\leq a<\frac32$ there is $C_a>0$ with $\norm{|x|^{-a}u}\leq C_a\norm{u}_{H^a}$ for all $u\in H^a(\mathbb R^3)$.
No angular condition is required.
\end{lemma}

\begin{proof}
This is the fractional Hardy inequality $\norm{|x|^{-a}|p|^{-a}}<\infty$ for $0\leq a<3/2$ in dimension $3$, already used in the proof of \cref{prop:onebody-first} for $a=\delta$ and $a=1+\delta$. At $a=1$ the sharp constant is available: $\norm{|x|^{-1}|p|^{-1}}\leq2$.
\end{proof}

\begin{rem}
Only the two base cases $m=-1,0$ of \cref{prop:Pgel-weighted} need $a\notin\{0,1\}$, namely $a=\delta$ and $a=1+\delta$; every other appeal to \cref{lem:hardy-base} below is at $a=1$ and could be replaced by $\norm{|x|^{-1}|p|^{-1}}\leq2$. The fractional range cannot be dispensed with in those two base cases. Splitting $|x|^{-1-\delta}\leq|x|^{-1}+|x|^{-2}$ fails outright, since $|x|^{-2}$ is not controlled by any $H^s$ norm without an angular condition -- this is precisely why the weight $|x|^{-2}$ is handled throughout by the angular inversion \eqref{eq:angular-inverse-norm}. Splitting $|x|^{-\delta}\leq1+|x|^{-1}$ is valid but replaces $\norm{u}_{H^\delta}$ by $\norm{u}_{H^1}$, which propagates through step (c) of the induction and costs a full extra derivative in \eqref{eq:Pgel-weighted}: for $\ell=1$ the hypothesis
$\psi_0\in H^{2+\delta}(\mathbb{R}^3)$ of \cref{prop:onebody-first} would become $\psi_0\in H^3(\mathbb{R}^3)$. The restriction $\delta<\frac12$ in \cref{prop:onebody-first} and the requirement $a<\frac32$ here are the same constraint, $1+\delta<\frac32$, and the almost-optimal exponent $(1+\delta)/2\to\frac34$ is exactly the limit $a\to\frac32$.
\end{rem}

We are now ready to prove the main ingredient 2, as stated in \cref{prop:Pgel-weighted}.

\begin{proof}[Proof of \cref{prop:Pgel-weighted}]
For $m\geq-1$ let $\mathcal H(m)$ denote the statement
\[
    \norm{|x|^{-(1+m+\delta)}h}
    \leq C_m\norm{h}_{H^{1+m+\delta}}
    \qquad\text{for all } h=P_{\geq m}h\in H^{1+m+\delta}(\mathbb{R}^3).
\]
The proposition is $\mathcal H(\ell)$ for $\ell\geq1$.

\emph{Base cases.} For $m=-1$ and $m=0$ the projection is the identity and the weights are $|x|^{-\delta}$ and $|x|^{-1-\delta}$, with $\delta<1+\delta<\frac32$. So $\mathcal H(-1)$ and $\mathcal H(0)$ are exactly \cref{lem:hardy-base}.

\emph{Induction step.} Let $\ell\geq1$ and assume $\mathcal H(\ell-1)$ and $\mathcal H(\ell-2)$. Let $g=P_{\geq\ell}g$. By \eqref{eq:angular-inverse-norm} we may write $g=(-\Delta_{S^2})^{-1}P_{\geq\ell}(-\Delta_{S^2})g$; since the radial weight commutes with $(-\Delta_{S^2})^{-1}P_{\geq\ell}$,
\[
    \frac{1}{|x|^{1+\ell+\delta}}g
    =(-\Delta_{S^2})^{-1}P_{\geq\ell}
     \left[\frac{1}{|x|^{\ell-1+\delta}}
     \frac{-\Delta_{S^2}}{|x|^2}g\right],
\]
first for $g=P_{\geq\ell}\varphi$ with $\varphi\in C_c^\infty(\mathbb R^3)$. Such functions are dense in $H^{1+\ell+\delta}(\mathbb{R}^3)\cap\operatorname{Ran}P_{\geq\ell}$, since $P_{\geq\ell}$ is bounded on every Sobolev space. The estimate obtained below is uniform in this dense class. Applying it to the difference of two approximating functions shows that the weighted sequence is Cauchy in $L^2$; its limit agrees with $|x|^{-1-\ell-\delta}g$ away from the origin.
This justifies the passage to a general $g$.
Using \eqref{eq:angular-inverse-norm} and, as in the proof of \cref{prop:onebody-first},
\[
    \frac{-\Delta_{S^2}}{|x|^2}
    =-\Delta+\partial_r^2+\frac{2}{r}\partial_r ,
\]
we obtain
\begin{equation}\label{eq:three-terms}
\begin{aligned}
    \norm{\frac{1}{|x|^{1+\ell+\delta}}g}
    \leq\frac{1}{\ell(\ell+1)}
    \Bigl[&\norm{\frac{\Delta g}{|x|^{\ell-1+\delta}}}
     +\norm{\frac{\partial_r^2g}{|x|^{\ell-1+\delta}}} \\
     &+2\norm{\frac{\partial_rg}{|x|^{\ell+\delta}}}\Bigr].
\end{aligned}
\end{equation}

\emph{(a) Laplacian term.} $\Delta g=P_{\geq\ell}\Delta g$, hence also $\Delta g=P_{\geq\ell-2}\Delta g$, and the weight is $|x|^{-(1+(\ell-2)+\delta)}$. So $\mathcal H(\ell-2)$ applies:
\[
    \norm{|x|^{-(\ell-1+\delta)}\Delta g}
    \leq C_{\ell-2}\norm{\Delta g}_{H^{\ell-1+\delta}}
    \leq C_{\ell-2}\norm{g}_{H^{1+\ell+\delta}} .
\]

\emph{(b) Second radial derivative.} As in the proof of \cref{prop:onebody-first}, $\partial_r$ commutes with the coefficients $x_j/r$, so $\partial_r^2g=\sum_{j,k}\frac{x_jx_k}{r^2}\partial_{x_j}\partial_{x_k}g$ with $|x_jx_k/r^2|\leq1$. By \cref{lem:derivative-lowers-one}, each $\partial_{x_j}\partial_{x_k}g$ lies in $\operatorname{Ran}P_{\geq\ell-2}$, and the weight is again $|x|^{-(1+(\ell-2)+\delta)}$, so $\mathcal H(\ell-2)$ applies to every summand:
\[
    \norm{|x|^{-(\ell-1+\delta)}\partial_r^2g}
    \leq 9\,C_{\ell-2}\norm{g}_{H^{1+\ell+\delta}} .
\]

\emph{(c) First radial derivative.} Here $\partial_rg=\sum_j\frac{x_j}{r}\partial_{x_j}g$ and the weight is
$|x|^{-(\ell+\delta)}=|x|^{-(1+(\ell-1)+\delta)}$. By \cref{lem:derivative-lowers-one}, $\partial_{x_j}g\in\operatorname{Ran} P_{\geq\ell-1}$, which is exactly the hypothesis of $\mathcal H(\ell-1)$:
\[
    2\norm{|x|^{-(\ell+\delta)}\partial_rg}
    \leq 6\,C_{\ell-1}\norm{g}_{H^{1+\ell+\delta}} .
\]
Note that the weight and the projection level match with no room to spare; this is the reason the recursion has two steps, $\ell-1$ from the first-order radial term and $\ell-2$ from the second-order one.

Inserting (a)-(c) into \eqref{eq:three-terms} proves $\mathcal H(\ell)$ with
\begin{equation}\label{eq:constant-recursion}
    C_\ell\leq\frac{16}{\ell(\ell+1)}
    \max\{C_{\ell-1},C_{\ell-2}\}.
\end{equation}

\emph{Uniformity in $\ell$.} Let $D_\ell:=\max_{-1\leq m\leq\ell}C_m$. For $\ell\geq4$ we have $16/(\ell(\ell+1))\leq16/20<1$, so \eqref{eq:constant-recursion} gives $C_\ell\leq D_{\ell-1}$ and therefore $D_\ell=D_{\ell-1}$. Hence $D_\ell=D_3$ for all $\ell\geq3$, and $\sup_{\ell\geq1}C_\ell\leq D_3<\infty$.
\end{proof}

\begin{lemma}[Coulomb multiplication on $P_{\geq\ell}$ data]
\label{lem:coulomb-mult-Pgel}
Let $m\geq1$ be an integer and $g=P_{\geq m}g\in H^{m+1}(\mathbb R^3)$. Then
\[
    \norm{\frac{g}{|x|}}_{H^m}\leq C_m\norm{g}_{H^{m+1}} .
\]
\end{lemma}

\begin{proof}
Since $\norm{h}_{H^m}$ is equivalent to $\sum_{|\alpha|\leq m}\norm{\partial^\alpha h}$, it suffices to bound $\norm{\partial^\alpha(|x|^{-1}g)}$ for every multi-index $|\alpha|\leq m$.
On $\mathbb R^3\setminus\{0\}$, Leibniz's rule gives
\[
    \partial^\alpha\!\left(\frac{g}{|x|}\right)
    =\sum_{\beta\leq\alpha}\binom{\alpha}{\beta}
     \left(\partial^{\alpha-\beta}\frac{1}{|x|}\right)\partial^\beta g,
\]
and $|\partial^{\alpha-\beta}|x|^{-1}|\leq C|x|^{-1-|\alpha-\beta|}$. Writing $k=|\beta|$ and $j=|\alpha|-k\leq m-k$, each term is bounded by a constant times
\[
    \norm{\frac{1}{|x|^{1+j}}\,\partial^\beta g},
    \qquad 0\leq j\leq m-k .
\]
Since $|x|^{-1-j}\leq |x|^{-1-(m-k)}+|x|^{-1}$ pointwise, and $|x|^{-1}$ is covered by \cref{lem:hardy-base}, it is enough to treat $j=m-k$. By \cref{lem:derivative-lowers-one}, $\partial^\beta g\in\operatorname{Ran} P_{\geq m-k}$, so \cref{prop:Pgel-weighted} with $\delta=0$ and angular index $m-k$ (or \cref{lem:hardy-base} when $m-k\leq0$) applies to the weight $|x|^{-(1+(m-k))}$ and yields
\[
    \norm{\frac{\partial^\beta g}{|x|^{1+m-k}}}
    \leq C_{m-k}\norm{\partial^\beta g}_{H^{1+m-k}}
    \leq C_{m-k}\norm{g}_{H^{m+1}} .
\]
The estimates above show that all weak derivatives of $|x|^{-1}g$ on $\mathbb R^3\setminus\{0\}$ up to order $m$ belong to $L^2(\mathbb R^3)$. Moreover, $|x|^{-1}g\in L^2(\mathbb R^3)$ by \cref{lem:hardy-base}. Since a point is removable for $H^1(\mathbb{R}^3)$ in three dimensions, applying this fact successively to the derivatives of orders $0,\ldots,m-1$ shows that these are also the global weak derivatives on $\mathbb R^3$. Therefore, $|x|^{-1}g\in H^m(\mathbb R^3)$. Summing over the finitely many $(\alpha,\beta)$ yields the desired estimate.
\end{proof}

We are now ready to prove the main ingredient 1, as stated in \cref{prop:Pgel-propagation}.
\begin{proof}[Proof of \cref{prop:Pgel-propagation}]

Since $V=c/|x|$ is radial, $-\Delta$ and $V$ commute with every $P_{\ell'}$, hence so do $H$, $e^{- i tH}$ and $e^{- i \tau(-\Delta)}$. In particular the constraint $P_{\geq\ell}$ is preserved by both flows, and $e^{- i \tau(-\Delta)}$
is an isometry on every $H^s(\mathbb{R}^3)$, so it suffices to prove the bound for $g:=e^{- i \tilde tH}\psi_0$ and then apply \cref{prop:Pgel-weighted} to $e^{- i \tau(-\Delta)}g=P_{\geq\ell} e^{- i \tau(-\Delta)}g$.

\emph{Step 1: integer orders.} We show by induction on the integer $n\geq0$ that
\begin{equation}\label{eq:integer-propagation}
    \sup_{t\in\mathbb R}\norm{e^{- i tH}u}_{H^n}
    \leq B_n\norm{u}_{H^n}
    \qquad\text{for all }u= P_{\geq n-2}u \in H^n (\mathbb{R}^3).
\end{equation}
For $n=0$ this is unitarity, and for $n=1$ it is \cref{eq:H1-propagation}; by our convention on $P_{\geq m}$, neither case requires an angular restriction.

Let $n\geq2$ and assume \cref{eq:integer-propagation} for all orders below $n$. Let $u=P_{\geq n-2}u$ and $G:=e^{- i tH}u$, so that $G=P_{\geq n-2}G$. From $-\Delta G=HG-VG$ and the equivalence $\norm{G}_{H^n}\leq C(\norm{G}+\norm{\Delta G}_{H^{n-2}})$, we have
\[
    \norm{G}_{H^n}
    \leq C\left(
      \norm{u}+\norm{HG}_{H^{n-2}}+\norm{VG}_{H^{n-2}}
    \right).
\]
Since $H$ commutes with $P_{\geq n-2}$, $Hu=P_{\geq n-2}Hu$. Moreover,
\[
    \norm{Hu}_{H^{n-2}}
    \leq \norm{\Delta u}_{H^{n-2}}
      +|c|\norm{|x|^{-1}u}_{H^{n-2}} 
    \leq C\norm{u}_{H^n},
\]
where for $n\geq3$ we used \cref{lem:coulomb-mult-Pgel} with $m=n-2$, while for $n=2$ we used \cref{lem:hardy-base}. The inductive hypothesis at order $n-2$ requires only the angular threshold $P_{\geq n-4}$ and therefore applies to $Hu$. Consequently,
\[
    \norm{HG}_{H^{n-2}}
    =\norm{e^{- i tH}Hu}_{H^{n-2}}
    \leq B_{n-2}\norm{Hu}_{H^{n-2}}
    \leq C\norm{u}_{H^n}.
\]
For the potential term, the inductive hypothesis at order $n-1$ applies to $u$, since it requires only $u=P_{\geq n-3}u$. Hence
\[
    \norm{G}_{H^{n-1}}
    \leq B_{n-1}\norm{u}_{H^{n-1}}.
\]
Applying \cref{lem:coulomb-mult-Pgel} with $m=n-2$ to $G=P_{\geq n-2}G$ for $n\geq3$, and \cref{lem:hardy-base} when $n=2$, gives
\[
    \norm{VG}_{H^{n-2}}
    \leq C\norm{G}_{H^{n-1}}
    \leq CB_{n-1}\norm{u}_{H^{n-1}}
    \leq C\norm{u}_{H^n}.
\]
This proves \cref{eq:integer-propagation} at order $n$.

\emph{Step 2: fractional order.} Fix $\ell\geq1$ and consider the family $e^{- i tH}P_{\geq\ell}$. At the integer order $n=1+\ell$, the angular threshold in \cref{eq:integer-propagation} is $P_{\geq\ell-1}$, while at $n=2+\ell$ it is $P_{\geq\ell}$. Therefore, uniformly in $t$,
\[
    e^{- i tH}P_{\geq\ell}:H^{1+\ell}(\mathbb{R}^3) \longrightarrow H^{1+\ell} (\mathbb{R}^3),
    \qquad
    e^{- i tH}P_{\geq\ell}:H^{2+\ell}(\mathbb{R}^3) \longrightarrow H^{2+\ell} (\mathbb{R}^3).
\]
Complex interpolation and $[H^{1+\ell}\REV{(\mathbb{R}^3)}, H^{2+\ell}\REV{(\mathbb{R}^3)}]_\delta = H^{1+\ell+\delta}\REV{(\mathbb{R}^3)}$ then yield
\[
    \sup_{t\in\mathbb R}
    \norm{e^{- i tH}P_{\geq\ell}u}_{H^{1+\ell+\delta}}
    \leq C_{\ell,\delta}\norm{u}_{H^{1+\ell+\delta}}.
\]
Taking $u=\psi_0=P_{\geq\ell}\psi_0$ proves
\cref{eq:Pgel-propagation}. Estimate \cref{eq:Pgel-weighted-dynamic} then
follows from \cref{prop:Pgel-weighted}, since
\[
    e^{- i \tau(-\Delta)}e^{- i \tilde tH}\psi_0
    =P_{\geq\ell}e^{- i \tau(-\Delta)}e^{- i \tilde tH}\psi_0
\]
by the commutation noted at the beginning of the proof.

\end{proof}

\subsection{Proof of~\cref{thm: 2-rd Trotter}}
\label{sec:pf_improved_rate_Trotter2}

Recall that the one-step error of the second-order Trotter formula is defined by
\begin{equation}
\label{eq: def e2t}
e_2(t)
:=
e^{-iAt/2}e^{-iBt}e^{-iAt/2}-e^{-itH}.
\end{equation}
We will prove the uniform local estimate
\begin{equation}
\label{eq:second_order_angular_local_rate}
\sup_{v\in\mathbb{R}}
\norm{e_2(t)e^{-ivH}\psi_0}
\leq
C_{\ell,\delta}
t^{1+\gamma_2(\ell,\delta)}
\norm{\psi_0}_{H^{1+\ell+\delta}},
\qquad 0<t\leq1.
\end{equation}
The long-time estimate in \cref{eq:second_order_angular_rate} will then follow from the standard telescoping argument.

We first specialize \cref{thm:exact_err_rep_Trotter2_general} to the Coulomb Hamiltonian. For $\ell\geq3$, all identities involving Coulomb commutators are first computed for smooth $P_{\geq\ell}$ functions compactly supported in $\mathbb{R}^3\setminus\{0\}$. Since $0\leq\delta<1/2$, radial cutoffs followed by mollification show that this class is dense in $H^{1+\ell+\delta}(\mathbb{R}^3)\cap \operatorname{Ran}(P_{\geq\ell})$, together with the weighted norms used below. We use the commutators below as the corresponding closed extensions.

\begin{lemma}
Under the assumptions of \cref{thm: 2-rd Trotter}, let $\ell\geq3$ and $f=e^{-ivH}\psi_0$ for some $v\in\mathbb{R}$. Then the one-step error in \cref{eq: def e2t} satisfies
\begin{equation}
\label{eq:e2tf-commutator}
\begin{aligned}
e_2(t)f
= &
\int_0^t\!ds\int_0^s\!du\,
e^{-iAs/2}e^{-iV(s-u)}
\left[
\left[V,\frac{A}{2}\right],
e^{-iuV}e^{-iA(s-u)/2}
\right] \\
&\hspace{42mm}\times
e^{-iAu/2}e^{-iH(t-s)}f.
\end{aligned}
\end{equation}
\end{lemma}

We now prove the general representation stated in \cref{thm:exact_err_rep_Trotter2_general}.

\begin{proof}[Proof of~\cref{thm:exact_err_rep_Trotter2_general}]
Consider
\begin{equation}
e^{\mathcal{L}_1s/2}e^{\mathcal{L}_2s}
e^{\mathcal{L}_1s/2}e^{\mathcal{L}(t-s)}.
\end{equation}
The difference of this operator at $s=t$ and $s=0$ is precisely the local Strang error. Therefore, the fundamental theorem of calculus gives
\begin{equation}
\label{eq:strang_L_1}
\begin{aligned}
&e^{\mathcal{L}_1t/2}e^{\mathcal{L}_2t}
e^{\mathcal{L}_1t/2}-e^{\mathcal{L}t} \\
= &
\int_0^t\!ds\,
e^{\mathcal{L}_1s/2}
\left[
\frac{\mathcal{L}_1}{2}+\mathcal{L}_2,
e^{\mathcal{L}_2s}e^{\mathcal{L}_1s/2}
\right]
e^{\mathcal{L}(t-s)}.
\end{aligned}
\end{equation}
For any admissible operators $\mathcal{A}$ and $\mathcal{B}$, we have
\begin{equation}
\label{eq:comm_A_expB}
[\mathcal{A},e^{s\mathcal{B}}]
=
\int_0^s\!du\,
e^{(s-u)\mathcal{B}}[\mathcal{A},\mathcal{B}]
e^{u\mathcal{B}}.
\end{equation}
Indeed, the right-hand side is obtained by differentiating $e^{(s-u)\mathcal{B}}\mathcal{A}e^{u\mathcal{B}}$ and evaluating the result at $u=s$ and $u=0$. Applying \cref{eq:comm_A_expB}, we obtain
\begin{equation}
\label{eq:strang_L_2}
\begin{aligned}
&\left[
\frac{\mathcal{L}_1}{2}+\mathcal{L}_2,
e^{\mathcal{L}_2s}e^{\mathcal{L}_1s/2}
\right] \\
= &
\frac12\int_0^s\!du\,
e^{\mathcal{L}_2s}e^{\mathcal{L}_1(s-u)/2}
[\mathcal{L}_2,\mathcal{L}_1]
e^{\mathcal{L}_1u/2} \\
&+
\frac12\int_0^s\!du\,
e^{\mathcal{L}_2(s-u)}
[\mathcal{L}_1,\mathcal{L}_2]
e^{\mathcal{L}_2u}e^{\mathcal{L}_1s/2} \\
= &
\frac12\int_0^s\!du\,
e^{\mathcal{L}_2(s-u)}
\left[
e^{\mathcal{L}_2u}e^{\mathcal{L}_1(s-u)/2},
[\mathcal{L}_2,\mathcal{L}_1]
\right]
e^{\mathcal{L}_1u/2}.
\end{aligned}
\end{equation}
Substituting \cref{eq:strang_L_2} into \cref{eq:strang_L_1} proves \cref{thm:exact_err_rep_Trotter2_general}. Taking $\mathcal{L}_1=-iA$ and $\mathcal{L}_2=-iV$ then gives \cref{eq:e2tf-commutator}.
\end{proof}

We are now ready to prove the improved error estimate for the second-order Trotter formula, as stated in \cref{thm: 2-rd Trotter}.
\begin{proof}[Proof of~\cref{thm: 2-rd Trotter}]
We first consider the cases $\ell=1,2$. Applying \cref{prop:onebody-first,eq:first-order-local-Pgeq-ell} to the evolved
initial state $e^{-ivH}\psi_0$, and then using \cref{prop:Pgel-propagation}, gives
\begin{equation}
\label{eq:uniform-first-order-orbit}
\sup_{v\in\mathbb{R}}
\norm{
\left(e^{-iB\tau}e^{-iA\tau}-e^{-iH\tau}\right)
e^{-ivH}\psi_0
}
\leq
C_{\ell,\delta}|\tau|^{1+\gamma_2(\ell,\delta)}
\norm{\psi_0}_{H^{1+\ell+\delta}}
\end{equation}
for $0<|\tau|\leq1$. For negative $\tau$, this follows by complex conjugation from the corresponding positive-time estimate. Indeed, $A$, $B$, and $H$ have real coefficients, while complex conjugation preserves both $\operatorname{Ran}(P_{\geq\ell})$ and the Sobolev norm.

Let $h=t/2$. A direct calculation gives
\begin{equation}
\label{eq:strang-two-first-order-errors}
\begin{aligned}
e_2(t)
=&
e^{-iAh}e^{-iBh}
\left(e^{-iBh}e^{-iAh}-e^{-iHh}\right) \\
&-
e^{-iAh}e^{-iBh}
\left(e^{iBh}e^{iAh}-e^{iHh}\right)e^{-iHt}.
\end{aligned}
\end{equation}
Since the factors outside the parentheses are unitary, \cref{eq:uniform-first-order-orbit} with $\tau=h$ and $\tau=-h$ proves \cref{eq:second_order_angular_local_rate} for $\ell=1,2$.

We next consider $\ell\geq3$.
We first record the weighted derivative estimates needed below. By \cref{lem:derivative-lowers-one}, a derivative of order $|\alpha|$ maps $\operatorname{Ran}(P_{\geq\ell})$ into $\operatorname{Ran}(P_{\geq\ell-|\alpha|})$. Applying \cref{prop:Pgel-weighted} with angular index $\ell-|\alpha|$, followed by the uniform Sobolev propagation estimate, gives
\begin{equation}
\label{eq:strang-weighted-derivative-orbit}
\sup_{\tau,v\in\mathbb{R}}
\sum_{|\alpha|\leq2}
\norm{
\frac{
\partial^\alpha e^{-i\tau A}e^{-ivH}\psi_0
}{
|x|^{1+\ell+\delta-|\alpha|}
}
}
\leq
C_{\ell,\delta}
\norm{\psi_0}_{H^{1+\ell+\delta}}.
\end{equation}
The corresponding estimates with smaller powers of $|x|^{-1}$ follow from \cref{lem:interpolated-weights}.

Let $\sigma=(s-u)/2$. Since $A=-\Delta$ and $B=V$, the commutator in \cref{eq:e2tf-commutator} can be written as
\begin{equation}
\left[
\left[V,\frac{A}{2}\right],
e^{-iuV}e^{-i\sigma A}
\right]
=
\frac12
\left[e^{-iuV}e^{-i\sigma A},[A,V]\right].
\end{equation}
The outer commutator on the right-hand side satisfies
\begin{equation}
\label{eq:strang_outer_commutator_split}
\left[e^{-iuV}e^{-i\sigma A},[A,V]\right]
=
e^{-iuV}\left[e^{-i\sigma A},[A,V]\right] 
+
\left[e^{-iuV},[A,V]\right]e^{-i\sigma A}.
\end{equation}
The second term on the right-hand side must be estimated separately.
A direct calculation gives
\begin{equation}
\label{eq:strang_potential_phase_commutator}
\left[e^{-iuV},[A,V]\right]
=
-2iu\,e^{-iuV}|\nabla V|^2.
\end{equation}
Since $|\nabla V|^2\leq C|x|^{-4}$, it follows from
\cref{eq:strang-weighted-derivative-orbit} that
\begin{equation}
\label{eq:strang_potential_phase_bound}
\norm{
\left[e^{-iuV},[A,V]\right]
e^{-i\tau A}e^{-ivH}\psi_0
}
\leq
C_{\ell,\delta}u
\norm{\psi_0}_{H^{1+\ell+\delta}}.
\end{equation}

It remains to estimate the first term on the right-hand side of \cref{eq:strang_outer_commutator_split}. We use the decomposition $V=V_{\mathrm{reg}}+V_{\mathrm{sin}}$ in \cref{eq: def Vreg,eq: def Vsin}, now with $s=t$ and $\beta=1/2$.
The cutoff bounds imply that
\begin{equation}
\begin{aligned}
|\partial^\alpha V_{\mathrm{reg}}(x,t)|
&\leq
C_\alpha
\chi_{\{|x|\geq t^{1/2}/2\}}|x|^{-1-|\alpha|}, \\
|\partial^\alpha V_{\mathrm{sin}}(x,t)|
&\leq
C_\alpha
\chi_{\{|x|\leq t^{1/2}\}}|x|^{-1-|\alpha|},
\qquad x\neq0,\quad |\alpha|\leq4.
\end{aligned}
\end{equation}

For a smooth function $b$, we have
\begin{equation}
\label{eq:cartesian_double_commutator}
\left[A,[A,b]\right]g
=
(\Delta^2b)g+4\nabla(\Delta b)\cdot\nabla g 
+4\sum_{j,k=1}^3
(\partial_{jk}b)\partial_{jk}g.
\end{equation}
When $\ell=3$, the three terms on the right-hand side require the weights $|x|^{-5}$, $|x|^{-4}$, and $|x|^{-3}$, respectively. On the support of $V_{\mathrm{reg}}$, each of these loses at most the factor $t^{-(1-\delta)/2}$ relative to the corresponding estimate in \cref{eq:strang-weighted-derivative-orbit}. When $\ell\geq4$, the required weights are already available. The region $|x|\geq1$ is controlled directly by the unweighted Sobolev norm. We therefore have
\begin{equation}
\label{eq:strang_regular_double_commutator}
\norm{[A,[A,V_{\mathrm{reg}}]]g}
\leq
\begin{cases}
C_{3,\delta}\,t^{-(1-\delta)/2}
\norm{\psi_0}_{H^{4+\delta}},
& \ell=3,\\[3pt]
C_{\ell,\delta}
\norm{\psi_0}_{H^{1+\ell+\delta}},
& \ell\geq4,
\end{cases}
\end{equation}
whenever $g=e^{-i\tau A}e^{-ivH}\psi_0$.

For the singular part, we use
\begin{equation}
[A,b]g
=
-(\Delta b)g-2\nabla b\cdot\nabla g.
\end{equation}
For $b=V_{\mathrm{sin}}$, this identity is first understood on the punctured smooth functions described above. Moreover, \cref{eq:strang-weighted-derivative-orbit} and Sobolev embedding show that the orbit states have zero trace at $x=0$. Consequently, the
point-mass contribution from $\Delta(c/|x|)=-4\pi c\delta_0$ also vanishes after taking the closed
extension. Since $V_{\mathrm{sin}}$ is supported in $|x|\leq t^{1/2}$, \cref{eq:strang-weighted-derivative-orbit} gives
\begin{equation}
\label{eq:strang_singular_first_commutator}
\norm{[A,V_{\mathrm{sin}}]g}
\leq
C_{\ell,\delta}
t^{(\ell-2+\delta)/2}
\norm{\psi_0}_{H^{1+\ell+\delta}}.
\end{equation}

For the regular part, the variation of constants formula gives
\begin{equation}
\left[e^{-i\sigma A},[A,V_{\mathrm{reg}}]\right]
=
-i\int_0^\sigma
e^{-i(\sigma-w)A}[A,[A,V_{\mathrm{reg}}]]
e^{-iwA}\,dw.
\end{equation}
For the singular part, we use
\begin{equation}
\begin{aligned}
\norm{\left[e^{-i\sigma A},[A,V_{\mathrm{sin}}]\right]g}
\leq
\norm{[A,V_{\mathrm{sin}}]g}
+
\norm{[A,V_{\mathrm{sin}}]e^{-i\sigma A}g}.
\end{aligned}
\end{equation}
Combining these estimates with
\cref{eq:strang_potential_phase_bound}, we obtain
\begin{equation}
\label{eq:strang_outer_commutator_bound}
\norm{\left[e^{-iuV}e^{-i\sigma A},[A,V]\right]g}
\leq
\begin{cases}
C_{3,\delta}
\left(
u+\sigma t^{-(1-\delta)/2}+t^{(1+\delta)/2}
\right)
\norm{\psi_0}_{H^{4+\delta}},
& \ell=3,\\[3pt]
C_{\ell,\delta}
\left(
u+\sigma+t^{(\ell-2+\delta)/2}
\right)
\norm{\psi_0}_{H^{1+\ell+\delta}},
& \ell\geq4.
\end{cases}
\end{equation}
for every $g=e^{-i\tau A}e^{-ivH}\psi_0$.

After pairing with a function in the same punctured smooth class, the two applications of the fundamental theorem of calculus in the proof of \cref{thm:exact_err_rep_Trotter2_general} are valid in the weak $L^2$ sense. The bound in \cref{eq:strang_outer_commutator_bound} shows that the resulting
integrands are integrable in the $L^2$ norm. Hence, by density, \cref{eq:e2tf-commutator} holds for $f=e^{-ivH}\psi_0$.

Every function to which the outer commutator in \cref{eq:e2tf-commutator} is applied has this form. Hence, by \cref{eq:strang_outer_commutator_bound} and unitarity,
\begin{equation}
\sup_{v\in\mathbb{R}}
\norm{e_2(t)e^{-ivH}\psi_0}
\leq
\begin{cases}
C_{3,\delta}
\left(
t^3+t^3t^{-(1-\delta)/2}+t^2t^{(1+\delta)/2}
\right)
\norm{\psi_0}_{H^{4+\delta}},
& \ell=3,\\[3pt]
C_{\ell,\delta}
\left(
t^3+t^2t^{(\ell-2+\delta)/2}
\right)
\norm{\psi_0}_{H^{1+\ell+\delta}},
& \ell\geq4.
\end{cases}
\end{equation}
For $\ell=3$, the last two terms are both $t^{(5+\delta)/2}$, while for $\ell\geq4$,
\begin{equation}
2+\frac{\ell-2+\delta}{2}\geq3.
\end{equation}
Since $0<t\leq1$, this proves \cref{eq:second_order_angular_local_rate} for every $\ell\geq3$.

Finally, the telescoping estimate gives
\begin{equation}
\begin{aligned}
\norm{E_{2,L}(t)\psi_0}
&\leq
\sum_{k=0}^{L-1}
\norm{e_2(t)e^{-iktH}\psi_0} \\
&\leq
C_{\ell,\delta}Lt^{1+\gamma_2(\ell,\delta)}
\norm{\psi_0}_{H^{1+\ell+\delta}} 
=
C_{\ell,\delta}Tt^{\gamma_2(\ell,\delta)}
\norm{\psi_0}_{H^{1+\ell+\delta}}.
\end{aligned}
\end{equation}
This proves \cref{eq:second_order_angular_rate} and completes the proof.
\end{proof}

\subsection{On Two-body Case}
To prove \cref{thm: 1st-Trotter two-body} and \cref{thm: 2-rd Trotter two-body}, we reduce the two-body evolution $e^{-itH}\psi_0$ to an effective one-body problem by introducing the center-of-mass coordinate $R$ and the relative coordinate $r$ (see \cref{eq: R,eq: r}). With this change of variables, we have
\begin{equation}
e^{-itH}\psi_0
= e^{-it(-\tfrac{\hbar^{2}}{2M}\Delta_{R})}\,
  e^{-itH_{\mathrm{rel}}}\psi_0,
\end{equation}
where $H_{\mathrm{rel}}$ is defined in \cref{eq: def Hrel}. A key feature of this decomposition is that the operators $-\tfrac{\hbar^{2}}{2M}\Delta_R$ and $H_{\mathrm{rel}}$ commute.

In these coordinates, the kinetic and potential parts take the form
\begin{equation}
A = -\frac{\hbar^2}{2m_e} \Delta_e - \frac{\hbar^2}{2m_p} \Delta_p
  = -\frac{\hbar^2}{2M} \Delta_R - \frac{\hbar^2}{2\mu} \Delta_r, 
\qquad 
B = - \frac{e^2}{|r_e - r_p|} = - \frac{e^2}{|r|}.
\end{equation}
Accordingly, the first-order Trotter formula can be written as
\begin{equation}
e^{-iBt} e^{-iAt} \psi_0
= e^{-it(-\tfrac{\hbar^2}{2M}\Delta_R)}\,
  e^{-iBt}\,
  e^{-it(- \tfrac{\hbar^2}{2\mu}\Delta_r)} \psi_0.
\end{equation}
Since the center-of-mass evolution decouples, the Trotter error reduces to that of the corresponding one-body problem governed by $H_{\mathrm{rel}}$.

We thus obtain the two-body results by reducing to the corresponding one-body problem and applying the one-body results established in prior sections, identifying
\[
e^{-itH_{\mathrm{rel}}}\psi_0
= \psi\!\left(x, \tfrac{\hbar^{2} t}{2\mu}\right)
\]
for the choice $c = \tfrac{2e^{2} \REV{\mu}}{\hbar^{2}}$. An analogous reduction applies to the second-order Trotter formula.

\section{Better Convergence for Many-body Fermionic Cases (Main Result 3 Proofs)}\label{sec:pf_main3_many-body-fermion}

As discussed in \cref{sec:main-result-3-many-body}, the key ingredient allowing this many-body improvement is the relative angular momentum property in \cref{prop:Nbody_Pge1}. In this section, we first prove \cref{prop:Nbody_Pge1} and establish the relative-coordinate estimates and Sobolev norm propagation estimates needed for the argument. We then prove Main Result 3 by deriving the pairwise estimates that allow us to apply the one-body argument for $P_{\geq1}$ data to each Coulomb interaction.

\subsection{Relative-coordinate Estimates}

We begin with the proof of the key observation.
\begin{proof}[Proof of~\cref{prop:Nbody_Pge1}]
It is enough to prove the claim for the pair $(1,2)$. Write
\begin{equation}
y=x_1-x_2,
\qquad
z=x_1+x_2,
\end{equation}
and regard
\begin{equation}
\Psi(x_1,x_2,x_3,\ldots,x_N)
=
\widetilde{\Psi}(y,z,x_3,\ldots,x_N).
\end{equation}
Since
\begin{equation}
x_1=\frac{z+y}{2},
\qquad
x_2=\frac{z-y}{2},
\end{equation}
exchanging $x_1$ and $x_2$ replaces $y$ by $-y$ and leaves $z$ unchanged. Hence fermionic antisymmetry gives
\begin{equation}
\widetilde{\Psi}(-y,z,x_3,\ldots,x_N)
=
-\widetilde{\Psi}(y,z,x_3,\ldots,x_N).
\end{equation}

For fixed $r=|y|$, $z$, and $x_3,\ldots,x_N$, the zeroth angular mode is
\begin{equation}
\int_{S^2}
\widetilde{\Psi}(r\omega,z,x_3,\ldots,x_N)
\,d\omega.
\end{equation}
Changing variables $\nu=-\omega$ gives
\begin{equation}
\begin{aligned}
\int_{S^2}
\widetilde{\Psi}(r\omega,z,x_3,\ldots,x_N)
\,d\omega
&=
\int_{S^2}
\widetilde{\Psi}(-r\nu,z,x_3,\ldots,x_N)
\,d\nu \\
&=
-
\int_{S^2}
\widetilde{\Psi}(r\nu,z,x_3,\ldots,x_N)
\,d\nu.
\end{aligned}
\end{equation}
Thus the angular average vanishes, which proves $\Psi=P_{\geq1,x_1-x_2}\Psi$. The same argument applies to every pair $(j,k)$.

\end{proof}

For the remainder of this section, we write $p_j:=-i\nabla_{x_j}$ and $\langle p_j-p_k\rangle:=(1+|p_j-p_k|^2)^{1/2}$. The proof of \cref{prop:Nbody-rate} uses the following three auxiliary lemmas.

\begin{lemma}[Relative-momentum counting]
\label{lem:relative-counting}
Let $0\leq\delta<\frac12$ and $\Phi\in H^{2+\delta}(\mathbb{R}^{3N})$. Define
\begin{equation}
\label{eq:def_relative_momentum_sum}
\mathcal{R}_{2+\delta}(\Phi)
:=\sum_{1\leq j<k\leq N}
\norm{\langle p_j-p_k\rangle^{2+\delta}\Phi}.
\end{equation}
Then
\begin{equation}
\label{eq:relative_momentum_counting}
\mathcal{R}_{2+\delta}(\Phi)
\leq C_\delta\left(
N^2\norm{\Phi}_{H^\delta}
+N^{3/2}\norm{\Phi}_{H^{2+\delta}}
\right).
\end{equation}
\end{lemma}

\begin{proof}
By the Cauchy-Schwarz inequality in the pair indices,
\begin{equation}
\mathcal{R}_{2+\delta}(\Phi)
\leq \left(\frac{N(N-1)}{2}\right)^{1/2}\left(
\sum_{1\leq j<k\leq N}
\norm{\langle p_j-p_k\rangle^{2+\delta}\Phi}^2
\right)^{1/2}.
\end{equation}
For $\xi=(\xi_1,\ldots,\xi_N)\in\mathbb{R}^{3N}$, we have
\begin{equation}
\sum_{1\leq j<k\leq N}
\langle\xi_j-\xi_k\rangle^4
\leq C\left(N^2+N|\xi|^4\right),
\end{equation}
where we used
\begin{equation}
\begin{aligned}
\sum_{1\leq j<k\leq N}|\xi_j-\xi_k|^2
&=
N\sum_{j=1}^N|\xi_j|^2
-
\left|\sum_{j=1}^N\xi_j\right|^2
\leq N|\xi|^2, \\
\sum_{1\leq j<k\leq N}|\xi_j-\xi_k|^4
&\leq CN|\xi|^4.
\end{aligned}
\end{equation}
Moreover, $\langle\xi_j-\xi_k\rangle^{2\delta} \leq C_\delta\langle\xi\rangle^{2\delta}$, and hence
\begin{equation}
\sum_{1\leq j<k\leq N}
\langle\xi_j-\xi_k\rangle^{4+2\delta}
\leq C_\delta\left(
N^2\langle\xi\rangle^{2\delta}
+N\langle\xi\rangle^{4+2\delta}
\right).
\end{equation}
The desired result now follows from Plancherel's theorem.
\end{proof}

Before stating and proving the lemmas below, we recall the many-body Hardy estimate proved in
\cite[Lemma~8]{FangWuSoffer2025}: for every
$\Phi\in H^1(\mathbb{R}^{3N})$,
\begin{equation}
\label{eq:Nbody_hardy_estimate}
\norm{V\Phi}
\leq
Cc_0N^{3/2}\norm{\nabla\Phi}.
\end{equation}

\begin{lemma}[$N$-body lower-order propagation]
\label{lem:Nbody-lower-propagation}
Let $0\leq\delta\leq1$ and $\Phi\in H^\delta(\mathbb{R}^{3N})$. Then
\begin{equation}
\label{eq:Nbody_lower_order_propagation}
\sup_{t\in\mathbb{R}}
\norm{e^{-itH}\Phi}_{H^\delta}
\leq C_{\delta,c_0}\left(1+N^3\right)^{\delta/2}
\norm{\Phi}_{H^\delta}.
\end{equation}
\end{lemma}

\begin{proof}
The many-body Hardy estimate (\cref{eq:Nbody_hardy_estimate}) gives
\begin{equation}
\left|\left\langle \Phi,V\Phi\right\rangle\right|
\leq Cc_0N^{3/2}\norm{\nabla\Phi}\norm{\Phi}
\leq \frac12\norm{\nabla\Phi}^2
+Cc_0^2N^3\norm{\Phi}^2.
\end{equation}
Choose $\lambda=C(1+c_0^2N^3)$ sufficiently large and let
\begin{equation}
q_\lambda[\Phi]
:=\norm{\nabla\Phi}^2
+\left\langle\Phi,V\Phi\right\rangle
+\lambda\norm{\Phi}^2.
\end{equation}
The preceding estimate gives
\begin{equation}
\frac12\norm{\Phi}_{H^1}^2
\leq q_\lambda[\Phi]
\leq C\left(1+c_0^2N^3\right)\norm{\Phi}_{H^1}^2.
\end{equation}
Since $q_\lambda[e^{-itH}\Phi]=q_\lambda[\Phi]$, it follows that
\begin{equation}
\sup_{t\in\mathbb{R}}
\norm{e^{-itH}\Phi}_{H^1}
\leq C_{c_0}\left(1+N^3\right)^{1/2}\norm{\Phi}_{H^1}.
\end{equation}
Interpolating this estimate with the unitarity of $e^{-itH}$ on $L^2$ yields the desired result.
\end{proof}

\begin{lemma}[The $H^1$ norm of the Coulomb term]
\label{lem:Nbody-VH1}
Let $\Phi\in H^2(\mathbb{R}^{3N})$ have fermionic antisymmetry. Then
\begin{equation}
\label{eq:Nbody_VH1}
\norm{V\Phi}_{H^1}
\leq Cc_0N^2\norm{\Phi}_{H^2}.
\end{equation}
\end{lemma}

\begin{proof}
The many-body Hardy estimate (\cref{eq:Nbody_hardy_estimate}) gives
\begin{equation}
\norm{V\Phi}\leq Cc_0N^{3/2}\norm{\Phi}_{H^1}.
\end{equation}
For the derivative term, we use
\begin{equation}
\norm{\nabla(V\Phi)}
\leq \norm{(\nabla V)\Phi}+\norm{V\nabla\Phi}.
\end{equation}
By \cref{prop:Nbody_Pge1} and \cref{prop:Pgel-weighted}, applied in each relative coordinate, and \cref{lem:relative-counting} with $\delta=0$, we have
\begin{equation}
\begin{aligned}
\norm{(\nabla V)\Phi}
&\leq
Cc_0
\sum_{1\leq j<k\leq N}
\norm{\frac{\Phi}{|x_j-x_k|^2}} \\
&\leq
Cc_0\mathcal{R}_2(\Phi)
\leq
Cc_0N^2\norm{\Phi}_{H^2}.
\end{aligned}
\end{equation}
Applying the same many-body Hardy estimate (\cref{eq:Nbody_hardy_estimate}) to the first derivatives of $\Phi$ also gives
\begin{equation}
\norm{V\nabla\Phi}\leq Cc_0N^{3/2}\norm{\Phi}_{H^2}.
\end{equation}
Combining the three estimates yields the desired result.
\end{proof}

\subsection{Proof of~\cref{prop:Nbody-rate}}

\begin{proof}[Proof of~\cref{prop:Nbody-rate}]

By the exchange-invariance assumption, $H$ commutes with every particle transposition operator. Hence $e^{-i\tau H}$ also commutes with every particle transposition operator and therefore preserves fermionic antisymmetry for every $\tau\in\mathbb{R}$.

By \cref{prop:Nbody_Pge1}, every fermionic state has no zeroth angular mode in each relative coordinate. We first record the corresponding pairwise weighted estimate. For every fermionic $\Phi\in H^{2+\delta}(\mathbb{R}^{3N})$ and every pair $j<k$,
\begin{equation}
\label{eq:pairwise-weighted}
\norm{\frac{\Phi}{|x_j-x_k|^{2+\delta}}}
+
\norm{\frac{(p_j-p_k)\Phi}{|x_j-x_k|^{1+\delta}}}
\leq
C_\delta
\norm{\langle p_j-p_k\rangle^{2+\delta}\Phi}.
\end{equation}
Indeed, fix the pair $(j,k)$ and set
\begin{equation}
y:=x_j-x_k,
\qquad
z:=x_j+x_k.
\end{equation}
We write the remaining variables as $x'$ and denote by $\widetilde{\Phi}(y,z,x')$ the wave function in these coordinates.
This is a linear change of variables with constant Jacobian, and
\begin{equation}
|x_j-x_k|=|y|,
\qquad
p_j-p_k=2p_y,
\qquad
p_y:=-i\nabla_y.
\end{equation}
For almost every fixed $(z,x')$, the corresponding function of $y$ belongs to $\operatorname{Ran}(P_{\geq1})$ by \cref{prop:Nbody_Pge1}. Applying \cref{prop:Pgel-weighted} with $\ell=1$ to each fiber gives
\begin{equation}
\norm{\frac{\widetilde{\Phi}}{|y|^{2+\delta}}}_{L^2_y}
\leq
C_\delta
\norm{\langle p_y\rangle^{2+\delta}\widetilde{\Phi}}_{L^2_y}
\leq
C_\delta
\norm{\langle 2p_y\rangle^{2+\delta}\widetilde{\Phi}}_{L^2_y}.
\end{equation}
For the derivative term, we apply \cref{lem:hardy-base} with $a=1+\delta<3/2$ componentwise to $p_y\widetilde{\Phi}$ and obtain
\begin{equation}
\begin{aligned}
\norm{\frac{(p_j-p_k)\widetilde{\Phi}}{|y|^{1+\delta}}}_{L^2_y}
&=
2\norm{\frac{p_y\widetilde{\Phi}}{|y|^{1+\delta}}}_{L^2_y} \\
&\leq
C_\delta
\norm{p_y\widetilde{\Phi}}_{H^{1+\delta}_y}
\leq
C_\delta
\norm{\langle 2p_y\rangle^{2+\delta}\widetilde{\Phi}}_{L^2_y}.
\end{aligned}
\end{equation}
This also includes the endpoint $\delta=0$. Squaring and integrating these estimates over $(z,x')$ proves \cref{eq:pairwise-weighted}; the constant Jacobian is absorbed into $C_\delta$.

We next estimate the propagation of the $H^{2+\delta}$ norm. By \cref{lem:Nbody-VH1} and $H^1(\mathbb{R}^{3N})\subset H^\delta(\mathbb{R}^{3N})$, for every fermionic $\Phi\in H^{2+\delta}(\mathbb{R}^{3N})$ we have
\begin{equation}
\norm{H\Phi}_{H^\delta}
\leq
\norm{-\Delta\Phi}_{H^\delta}
+
\norm{V\Phi}_{H^1}
\leq
C_{\delta,c_0}N^2\norm{\Phi}_{H^{2+\delta}}.
\end{equation}

Let $\Phi\in H^{2+\delta}(\mathbb{R}^{3N})$ be fermionic and set $G=e^{-i\tau H}\Phi$. Since $H$ commutes with particle exchanges, $G$ remains fermionic. Using $-\Delta G=HG-VG$, we obtain
\begin{equation}
\norm{G}_{H^{2+\delta}}
\leq
C\left(
\norm{G}
+
\norm{HG}_{H^\delta}
+
\norm{VG}_{H^\delta}
\right).
\end{equation}
Since $HG=e^{-i\tau H}H\Phi$, it follows from \cref{lem:Nbody-lower-propagation} that
\begin{equation}
\begin{aligned}
\norm{HG}_{H^\delta}
&\leq
C_{\delta,c_0}\left(1+N^3\right)^{\delta/2}
\norm{H\Phi}_{H^\delta} \\
&\leq
C_{\delta,c_0}N^2\left(1+N^3\right)^{\delta/2}
\norm{\Phi}_{H^{2+\delta}}.
\end{aligned}
\end{equation}
Moreover, since $0\leq\delta<1$, we may use \cref{lem:Nbody-VH1,lem:sobolev_n-body} to obtain
\begin{equation}
\begin{aligned}
\norm{VG}_{H^\delta}
&\leq
\norm{VG}_{H^1}
\leq
Cc_0N^2\norm{G}_{H^2} \\
&\leq
Cc_0N^2C_N\norm{\Phi}_{H^{2+\delta}}.
\end{aligned}
\end{equation}
The estimate in \cref{lem:sobolev_n-body} is unchanged when $\tau$ is replaced by $-\tau$. Since $C_N=\mathcal{O}(N^3)$ by \cref{eq:def_CN_sobolev}, the preceding estimates give
\begin{equation}
\label{eq:Nbody-H2delta-propagation}
\begin{aligned}
\sup_{\tau\in\mathbb{R}}
\norm{e^{-i\tau H}\Phi}_{H^{2+\delta}}
&\leq
C_{\delta,c_0}
\left[
1
+N^2\left(1+N^3\right)^{\delta/2}
+N^2C_N
\right]
\norm{\Phi}_{H^{2+\delta}} \\
&\leq
C_{\delta,c_0}N^5
\norm{\Phi}_{H^{2+\delta}}.
\end{aligned}
\end{equation}

For $v\in\mathbb{R}$, set $\Phi_v=e^{-ivH}\Psi_0$. Applying \cref{lem:relative-counting} to $\Phi_v$ gives
\begin{equation}
\mathcal{R}_{2+\delta}(\Phi_v)
\leq
C_\delta\left(
N^2\norm{\Phi_v}_{H^\delta}
+
N^{3/2}\norm{\Phi_v}_{H^{2+\delta}}
\right).
\end{equation}
The two Sobolev norms on the right-hand side are controlled by \cref{lem:Nbody-lower-propagation,eq:Nbody-H2delta-propagation}.
Therefore,
\begin{equation}
\label{eq:evolved-relative-final}
\begin{aligned}
\sup_{v\in\mathbb{R}}
\mathcal{R}_{2+\delta}(e^{-ivH}\Psi_0)
&\leq
C_{\delta,c_0}
\left[
N^2\left(1+N^3\right)^{\delta/2}
+N^{13/2}
\right]
\norm{\Psi_0}_{H^{2+\delta}} \\
&\leq
C_{\delta,c_0}N^{13/2}
\norm{\Psi_0}_{H^{2+\delta}}.
\end{aligned}
\end{equation}

Both $e^{-i\tau A}$ and $e^{-ivH}$ preserve fermionic antisymmetry, and the free evolution commutes with every $p_j-p_k$. Applying \cref{eq:pairwise-weighted} to $e^{-i\tau A}e^{-ivH}\Psi_0$, summing over all pairs, and using \cref{eq:evolved-relative-final}, we obtain
\begin{equation}
\label{eq:Nbody-relative-weighted-orbit}
\begin{aligned}
&\sup_{\tau,v\in\mathbb{R}}
\sum_{1\leq j<k\leq N}
\left(
\norm{
\frac{e^{-i\tau A}e^{-ivH}\Psi_0}
{|x_j-x_k|^{2+\delta}}
}
+
\norm{
\frac{(p_j-p_k)e^{-i\tau A}e^{-ivH}\Psi_0}
{|x_j-x_k|^{1+\delta}}
}
\right) \\
&\qquad\leq
C_{\delta,c_0}N^{13/2}
\norm{\Psi_0}_{H^{2+\delta}}.
\end{aligned}
\end{equation}

We now apply the cutoff decomposition in \cref{eq: def Vreg,eq: def Vsin} to each pair potential $|x_j-x_k|^{-1}$.

For a fixed pair $(j,k)$, set $y=x_j-x_k$ and $z=x_j+x_k$ as above.
Since $-\Delta_{x_j}-\Delta_{x_k}=-2\Delta_y-2\Delta_z$ and all the remaining kinetic terms commute with $V_{\mathrm{reg}}(y,s)$, we have
\begin{equation}
[A,V_{\mathrm{reg}}(x_j-x_k,s)]
=
2[-\Delta_y,V_{\mathrm{reg}}(y,s)].
\end{equation}
Hence \cref{eq:reg-general} of \cref{lem:general-weight-commutator}, applied in the variable $y$ with $a=2+\delta$, gives
\begin{equation}
\label{eq:Nbody-regular-generator-bound}
\begin{aligned}
\norm{[A,V_{\mathrm{reg}}(x_j-x_k,s)]\Phi}
&\leq
Cs^{-(1-\delta)\beta}
\left(
\norm{\frac{\Phi}{|x_j-x_k|^{2+\delta}}}
+
\norm{\frac{(p_j-p_k)\Phi}{|x_j-x_k|^{1+\delta}}}
\right) \\
&\leq
C_\delta s^{-(1-\delta)\beta}
\norm{\langle p_j-p_k\rangle^{2+\delta}\Phi}.
\end{aligned}
\end{equation}
Here the fixed factors coming from $p_j-p_k=2p_y$ and the preceding commutator identity are absorbed into the constant.

The variation of constants formula gives
\begin{equation}
[e^{-isA},V_{\mathrm{reg}}(x_j-x_k,s)]\Phi
=
-i\int_0^s e^{-i(s-u)A}
[A,V_{\mathrm{reg}}(x_j-x_k,s)]e^{-iuA}\Phi\,du.
\end{equation}
Since the free evolution commutes with $p_j-p_k$ and preserves fermionic antisymmetry, \cref{eq:Nbody-regular-generator-bound} therefore yields
\begin{equation}
\label{eq:Nbody-regular-commutator}
\norm{[e^{-isA},V_{\mathrm{reg}}(x_j-x_k,s)]\Phi}
\leq
C_\delta s^{1-(1-\delta)\beta}
\norm{\langle p_j-p_k\rangle^{2+\delta}\Phi}.
\end{equation}

For the singular part, expanding the commutator as in the proof of \cref{eq:sin-general} in \cref{lem:general-weight-commutator} and using the support of the cutoff gives
\begin{equation}
\norm{[e^{-isA},V_{\mathrm{sin}}(x_j-x_k,s)]\Phi} 
\leq
Cs^{(1+\delta)\beta}
\sup_{u\in\mathbb{R}}
\norm{
\frac{e^{-iuA}\Phi}{|x_j-x_k|^{2+\delta}}
}.
\end{equation}
Applying \cref{eq:pairwise-weighted} to $e^{-iuA}\Phi$ and using the commutation of $e^{-iuA}$ with $p_j-p_k$, we obtain
\begin{equation}
\label{eq:Nbody-singular-commutator}
\norm{[e^{-isA},V_{\mathrm{sin}}(x_j-x_k,s)]\Phi}
\leq
C_\delta s^{(1+\delta)\beta}
\norm{\langle p_j-p_k\rangle^{2+\delta}\Phi}.
\end{equation}
Choosing $\beta=1/2$ in \cref{eq:Nbody-regular-commutator,eq:Nbody-singular-commutator}, we obtain the pairwise commutator estimate
\begin{equation}
\label{eq:Nbody-pair-commutator}
\norm{[e^{-isA},|x_j-x_k|^{-1}]\Phi}
\leq
C_\delta s^{(1+\delta)/2}
\norm{\langle p_j-p_k\rangle^{2+\delta}\Phi}.
\end{equation}

We now apply the exact local error representation in \cref{eq:Trotter1_error_operator_exact_rep}. Using \cref{eq:Nbody-pair-commutator}, summing over all pairs, and recalling that $|c_{jk}|\leq c_0$, we obtain
\begin{equation}
\label{eq:Nbody-first-order-local-orbit}
\begin{aligned}
&\sup_{\sigma\in\mathbb{R}}
\norm{
\left(e^{-iBt}e^{-iAt}-e^{-iHt}\right)e^{-i\sigma H}\Psi_0
} \\
\leq & \
C_\delta c_0\int_0^t s^{(1+\delta)/2}
\sup_{v\in\mathbb{R}}
\mathcal{R}_{2+\delta}(e^{-ivH}\Psi_0)\,ds \\
\leq & \
C_{\delta,c_0}N^{13/2}t^{(3+\delta)/2}
\norm{\Psi_0}_{H^{2+\delta}}.
\end{aligned}
\end{equation}

Finally, the telescoping argument in \cref{eq:long_term_error_by_lte_trotter1}, together with $Lt=T$, gives
\begin{equation}
\begin{aligned}
&\norm{
\left(
e^{-iHT}-\left(e^{-iBt}e^{-iAt}\right)^L
\right)\Psi_0
} \\
\leq & \
L\sup_{\sigma\in\mathbb{R}}
\norm{
\left(e^{-iBt}e^{-iAt}-e^{-iHt}\right)e^{-i\sigma H}\Psi_0
} \\
\leq & \ 
C_{\delta,c_0}N^{13/2}Tt^{(1+\delta)/2}
\norm{\Psi_0}_{H^{2+\delta}},
\end{aligned}
\end{equation}
which proves the first estimate in \cref{prop:Nbody-rate}.

At the endpoint $\delta=0$, it is not necessary to use \cref{eq:Nbody-H2delta-propagation}. Applying \cref{lem:relative-counting} to $e^{-ivH}\Psi_0$ and using \cref{lem:sobolev_n-body}, we obtain directly
\begin{equation}
\begin{aligned}
\sup_{v\in\mathbb{R}}
\mathcal{R}_2(e^{-ivH}\Psi_0)
&\leq
C\left(N^2+N^{3/2}C_N\right)
\norm{\Psi_0}_{H^2} \\
&\leq
C_{c_0}N^{9/2}\norm{\Psi_0}_{H^2}.
\end{aligned}
\end{equation}
Substituting this estimate into the same local error argument gives
\begin{equation}
\label{eq:Nbody-first-order-local-endpoint}
\sup_{\sigma\in\mathbb{R}}
\norm{
\left(e^{-iBt}e^{-iAt}-e^{-iHt}\right)e^{-i\sigma H}\Psi_0
}
\leq
C_{c_0}N^{9/2}t^{3/2}\norm{\Psi_0}_{H^2}.
\end{equation}
Telescoping once more proves \cref{eq:Nbody_fermionic_half_order} and completes the proof.

\end{proof}

\subsection{Proof of~\cref{thm:Nbody_fermionic_second_order_rate}}

\begin{proof}[Proof of~\cref{thm:Nbody_fermionic_second_order_rate}]
We first note that the local first-order estimate in \cref{eq:Nbody-first-order-local-orbit} also holds for negative time.
Indeed, this follows by complex conjugation, since $A$, $B$, and $H$ have real coefficients, while complex conjugation preserves fermionic antisymmetry and the $H^{2+\delta}$ norm. Therefore, for every $0<|\tau|\leq1$,
\begin{equation}
\label{eq:Nbody-first-order-local-both-directions}
\sup_{\sigma\in\mathbb{R}}
\norm{
\left(
e^{-iB\tau}e^{-iA\tau}-e^{-iH\tau}
\right)
e^{-i\sigma H}\Psi_0
} 
\leq
C_{\delta,c_0}N^{13/2}
|\tau|^{(3+\delta)/2}
\norm{\Psi_0}_{H^{2+\delta}}.
\end{equation}

Let $h=t/2$. The algebraic identity in \cref{eq:strang-two-first-order-errors} is unchanged for the present many-body operators and gives
\begin{equation}
\begin{aligned}
e_2(t)
={}&
e^{-iAh}e^{-iBh}
\left(
e^{-iBh}e^{-iAh}-e^{-iHh}
\right) \\
&-
e^{-iAh}e^{-iBh}
\left(
e^{iBh}e^{iAh}-e^{iHh}
\right)e^{-iHt}.
\end{aligned}
\end{equation}
Applying this identity to $e^{-ivH}\Psi_0$, using the unitarity of the outer factors, and then applying \cref{eq:Nbody-first-order-local-both-directions} with $\tau=h$ and $\tau=-h$, we obtain
\begin{equation}
\label{eq:Nbody-second-order-local-orbit}
\sup_{v\in\mathbb{R}}
\norm{
e_2(t)e^{-ivH}\Psi_0
}
\leq
C_{\delta,c_0}N^{13/2}
t^{(3+\delta)/2}
\norm{\Psi_0}_{H^{2+\delta}}.
\end{equation}
The standard telescoping argument now gives
\begin{equation}
\begin{aligned}
\norm{E_{2,L}(t)\Psi_0}
&\leq
\sum_{k=0}^{L-1}
\norm{
e_2(t)e^{-iktH}\Psi_0
} \\
&\leq
C_{\delta,c_0}N^{13/2}
Lt^{(3+\delta)/2}
\norm{\Psi_0}_{H^{2+\delta}} 
=
C_{\delta,c_0}N^{13/2}
Tt^{(1+\delta)/2}
\norm{\Psi_0}_{H^{2+\delta}},
\end{aligned}
\end{equation}
which proves \cref{eq:Nbody_fermionic_second_order_rate}.

At the endpoint $\delta=0$, we instead use \cref{eq:Nbody-first-order-local-endpoint}. The same argument gives
\begin{equation}
\sup_{v\in\mathbb{R}}
\norm{
e_2(t)e^{-ivH}\Psi_0
}
\leq
C_{c_0}N^{9/2}t^{3/2}
\norm{\Psi_0}_{H^2}.
\end{equation}
Telescoping over the $L$ time steps therefore yields
\begin{equation}
\norm{
\left(
e^{-iHT}
-
\left(
e^{-iAt/2}e^{-iBt}e^{-iAt/2}
\right)^L
\right)
\Psi_0
}
\leq
C_{c_0}N^{9/2}Tt^{1/2}
\norm{\Psi_0}_{H^2},
\end{equation}
which proves the endpoint estimate and completes the proof.
\end{proof}

\section{Local Error Lower Bounds}\label{sec:local_error_lower_bounds}

In this section, we establish the lower bound for the local truncation error for both the first-order and second-order Trotter formulas in the context of the Coulomb potential. We show that for general initial data in $H^2(\mathbb{R}^3)$, both local error rates are no better than $t^{5/4}$. This implies that no global rate better than $t^{1/4}$ holds for all numbers of time steps $L \geq 1$. 

It is sufficient to consider the one-body Coulomb problem in \cref{1-SE}. More precisely, we exhibit a fixed normalized initial state in $H^2(\mathbb{R}^3)$ for which both local errors have an exact leading term of order $t^{5/4}$. We work with the first-order local error $E(t)$ from \cref{prop:onebody-first}  and the second-order local error $e_2(t)$ defined in \cref{eq: def e2t}.

\begin{thm}[Sharp Local Error Asymptotics]
\label{thm:sharp_local_error_asymptotics}
Consider the one-body Hamiltonian in \cref{1-SE} with $V(x)=\pm c/|x|$, where $c>0$, and let
\begin{equation}
\label{eq:lower_bound_initial_state}
\psi_0(x)
:=
\frac{1}{\sqrt{\pi}}e^{-|x|}.
\end{equation}
Then $\psi_0$ is normalized and belongs to $H^2(\mathbb{R}^3)$. Moreover, as $t\to0^+$,
\begin{equation}
\label{eq:first_order_local_error_asymptotics}
\norm{ E (t)\psi_0}
=
\frac{4c\,2^{1/4}}{\sqrt{5\sqrt{\pi}}}
t^{5/4}
+
o(t^{5/4}),
\end{equation}
and
\begin{equation}
\label{eq:second_order_local_error_asymptotics}
\norm{e_2(t)\psi_0}
=
4c
\left(
\frac{5-2\sqrt{2}}{15\sqrt{\pi}}
\right)^{1/2}
t^{5/4}
+
o(t^{5/4}).
\end{equation}
In particular, the leading constants in both estimates are strictly positive.
\end{thm}

Before proving \cref{thm:sharp_local_error_asymptotics}, we first collect two estimates that will be used to control the remainder terms.

\begin{lemma}[Remainder Estimates]
\label{lem:local_lower_bound_remainder_estimates}
Let $A=-\Delta$, $B=V$, and $H=A+B$ as in \cref{def:A_and_B}. For every $v\in H^2(\mathbb{R}^3)$ and $0<t\leq1$,
\begin{equation}
\label{eq:local_lower_bound_propagator_increment}
\norm{B(e^{- i tA}-I)v}
+
\norm{B(e^{- i tH}-I)v}
\leq
C t^{1/2}\norm{v}_{H^2},
\end{equation}
and
\begin{equation}
\label{eq:local_lower_bound_multiplication_remainder}
\norm{(e^{- i tB}-I+ i tB)v}
\leq
C t^{3/2}\norm{v}_{H^2},
\end{equation}
where $C>0$ depends only on $c$.
\end{lemma}

\begin{proof}
We first prove \cref{eq:local_lower_bound_propagator_increment}. For $G=A$ or $G=H$, the spectral theorem gives
\begin{equation}
\norm{(e^{- i tG}-I)v}
\leq
t\norm{Gv}
\leq
Ct\norm{v}_{H^2}.
\end{equation}
The free evolution is unitary on $H^2(\mathbb{R}^3)$, and the evolution generated by $H$ is uniformly bounded on $H^2(\mathbb{R}^3)$. Therefore,
\begin{equation}
\norm{(e^{- i tG}-I)v}_{H^2}
\leq
C\norm{v}_{H^2}.
\end{equation}
Interpolating the last two estimates and applying the one-body Hardy inequality in \cref{eq:HLS_one-body}, we obtain
\begin{equation}
\norm{B(e^{- i tG}-I)v}
\leq
2c\norm{\nabla(e^{- i tG}-I)v}
\leq
Ct^{1/2}\norm{v}_{H^2}.
\end{equation}

To prove \cref{eq:local_lower_bound_multiplication_remainder}, we use the elementary estimate
\begin{equation}
\left|e^{- i y}-1+ i y\right|
\leq
C\min\{|y|,|y|^2\},
\qquad y\in\mathbb{R}.
\end{equation}
Set $b=ct$. Since $|B(x)|=c/|x|$, it follows that
\begin{equation}
\left|(e^{- i tB}-I+ i tB)v(x)\right|
\leq
C\left(
\frac{b}{|x|}\boldsymbol{1}_{\{|x|\leq b\}}
+
\frac{b^2}{|x|^2}\boldsymbol{1}_{\{|x|>b\}}
\right)|v(x)|.
\end{equation}
By the Sobolev embedding $H^2(\mathbb{R}^3)\subset L^\infty(\mathbb{R}^3)$ and spherical integration,
\begin{equation}
\begin{aligned}
\norm{(e^{- i tB}-I+ i tB)v}^2
&\leq
C\norm{v}_\infty^2
\left(
b^2\int_0^b dr
+
b^4\int_b^\infty r^{-2}\,dr
\right) \\
&\leq
Cb^3\norm{v}_{H^2}^2.
\end{aligned}
\end{equation}
This proves the lemma.
\end{proof}

The following elementary Fourier integrals determine the leading constants in \cref{thm:sharp_local_error_asymptotics}.

\begin{lemma}[Fourier Profile Integrals]
\label{lem:local_lower_bound_fourier_profiles}
For every $z\geq0$,
\begin{equation}
\label{eq:local_lower_bound_profile_bounds}
\left|
\frac{1-e^{- i z}}{z}-\I
\right|
\leq
C\min\{z,1\},
\qquad
\left|
\frac{1-e^{- i z}}{z}- i e^{- i z/2}
\right|
\leq
C\min\{z^2,1\},
\end{equation}
where the expressions at $z=0$ are understood by continuity. Moreover,
\begin{equation}
\label{eq:first_order_profile_integral}
\int_0^\infty
\left|
\frac{1-e^{- i y^2}}{y^2}-\I
\right|^2
\frac{dy}{y^2}
=
\frac{2\sqrt{2\pi}}{5},
\end{equation}
and
\begin{equation}
\label{eq:second_order_profile_integral}
\int_0^\infty
\left|
\frac{1-e^{- i y^2}}{y^2}- i e^{- i y^2/2}
\right|^2
\frac{dy}{y^2}
=
\frac{2\sqrt{\pi}}{15}
\left(5-2\sqrt{2}\right).
\end{equation}
\end{lemma}

\begin{proof}
The two bounds in \cref{eq:local_lower_bound_profile_bounds} follow from
\begin{equation}
\frac{1-e^{- i z}}{z}-\I
=
 i \int_0^1(e^{- i uz}-1)\,du
\end{equation}
and
\begin{equation}
\frac{1-e^{- i z}}{z}- i e^{- i z/2}
=
 i e^{- i z/2}
\left(
\frac{\sin(z/2)}{z/2}-1
\right).
\end{equation}
To evaluate the integrals, we use the identity
\begin{equation}
\label{eq:local_lower_bound_cosine_integral}
\int_0^\infty
\frac{1-\cos(as)}{s^{3/2}}\,ds
=
\sqrt{2\pi a},
\qquad a\geq0,
\end{equation}
which follows by scaling and integration by parts.

For the first-order profile, the substitution $s=y^2$ gives
\begin{equation}
\begin{aligned}
&\int_0^\infty
\left|
\frac{1-e^{- i y^2}}{y^2}-\I
\right|^2
\frac{dy}{y^2} \\
&\qquad=
\frac12\int_0^\infty\int_0^1\int_0^1
\frac{
1-\cos(us)+1-\cos(vs)-[1-\cos((u-v)s)]
}{s^{3/2}}
\,du\,dv\,ds.
\end{aligned}
\end{equation}
Applying \cref{eq:local_lower_bound_cosine_integral} to the three terms on the right-hand side, we obtain
\begin{equation}
\frac{\sqrt{2\pi}}{2}
\int_0^1\int_0^1
\left(\sqrt{u}+\sqrt{v}-\sqrt{|u-v|}\right)\,du\,dv
=
\frac{2\sqrt{2\pi}}{5}.
\end{equation}
This proves \cref{eq:first_order_profile_integral}.

For the second-order profile, its absolute value is
\begin{equation}
\left|
1-
\frac{\sin(y^2/2)}{y^2/2}
\right|.
\end{equation}
Using $s=y^2/2$ and $\sin s/s=\int_0^1\cos(us)\,du$, we have
\begin{equation}
\begin{aligned}
&\int_0^\infty
\left|
\frac{1-e^{- i y^2}}{y^2}- i e^{- i y^2/2}
\right|^2
\frac{dy}{y^2} \\
&\qquad=
\frac{1}{2\sqrt{2}}
\int_0^1\int_0^1\int_0^\infty
\frac{(1-\cos(us))(1-\cos(vs))}{s^{3/2}}
\,ds\,du\,dv.
\end{aligned}
\end{equation}
By \cref{eq:local_lower_bound_cosine_integral}, the inner integral is
\begin{equation}
\sqrt{2\pi}
\left(
\sqrt{u}+\sqrt{v}
-\frac12\sqrt{u+v}
-\frac12\sqrt{|u-v|}
\right).
\end{equation}
Finally,
\begin{equation}
\begin{aligned}
\int_0^1\int_0^1(\sqrt{u}+\sqrt{v})\,du\,dv
&=
\frac43, \\
\int_0^1\int_0^1\sqrt{u+v}\,du\,dv
&=
\frac{16\sqrt{2}-8}{15}, \\
\int_0^1\int_0^1\sqrt{|u-v|}\,du\,dv
&=
\frac{8}{15}.
\end{aligned}
\end{equation}
Substituting these identities proves \cref{eq:second_order_profile_integral}.
\end{proof}

\begin{proof}[Proof of \cref{thm:sharp_local_error_asymptotics}]
We use the unitary Fourier transform
\begin{equation}
\label{eq:local_lower_bound_fourier_transform}
\widehat{v}(\xi)
:=
(2\pi)^{-3/2}
\int_{\mathbb{R}^3}e^{- i x\cdot\xi}v(x)\,dx.
\end{equation}
For the initial state in \cref{eq:lower_bound_initial_state}, direct calculation gives
\begin{equation}
\label{eq:local_lower_bound_fourier_data}
\widehat{\psi_0}(\xi)
=
\frac{2\sqrt{2}}{\pi(1+|\xi|^2)^2},
\qquad
\widehat{V\psi_0}(\xi)
=
\frac{\pm c\sqrt{2}}{\pi(1+|\xi|^2)}.
\end{equation}
The first formula also shows directly that $\psi_0\in H^2(\mathbb{R}^3)$.

We first consider the first-order Trotter formula. The variation of constants formula on $H^2(\mathbb{R}^3)$ gives
\begin{equation}
\label{eq:local_lower_bound_duhamel}
e^{- i tH}\psi_0
=
e^{- i tA}\psi_0
-
 i \int_0^t e^{-\I(t-s)A}V e^{- i sH}\psi_0\,ds.
\end{equation}
Expanding only the multiplication operator $e^{- i tV}$ and using \cref{eq:local_lower_bound_duhamel}, we obtain the exact decomposition
\begin{equation}
\label{eq:first_order_local_error_decomposition}
\begin{aligned}
 E (t)\psi_0
=  &
- i tV\psi_0
+
 i \int_0^t e^{-\I(t-s)A}V\psi_0\,ds \\
&-
 i tV(e^{- i tA}-I)\psi_0
+
(e^{- i tV}-I+ i tV)e^{- i tA}\psi_0 \\
&+
 i \int_0^t e^{-\I(t-s)A}V(e^{- i sH}-I)\psi_0\,ds.
\end{aligned}
\end{equation}
By \cref{lem:local_lower_bound_remainder_estimates} and the unitarity of $e^{- i tA}$ on $H^2(\mathbb{R}^3)$, the last three terms satisfy
\begin{equation}
\label{eq:first_order_local_error_remainder}
\begin{aligned}
&t\norm{V(e^{- i tA}-I)\psi_0}
+
\norm{(e^{- i tV}-I+ i tV)e^{- i tA}\psi_0} \\
&\qquad+
\int_0^t\norm{V(e^{- i sH}-I)\psi_0}\,ds
\leq
Ct^{3/2}\norm{\psi_0}_{H^2}.
\end{aligned}
\end{equation}

The Fourier transform of the first line of \cref{eq:first_order_local_error_decomposition} is
\begin{equation}
t
\left(
\frac{1-e^{- i t|\xi|^2}}{t|\xi|^2}
-
\I
\right)
\widehat{V\psi_0}(\xi).
\end{equation}
Using \cref{eq:local_lower_bound_fourier_data}, polar coordinates, and the change of variables $y=\sqrt{t}\,|\xi|$, we find
\begin{equation}
\label{eq:first_order_local_error_scaled_limit}
\begin{aligned}
&t^{-5/2}
\norm{
- i tV\psi_0
+
 i \int_0^t e^{-\I(t-s)A}V\psi_0\,ds
}^2 \\
&\qquad=
\frac{8c^2}{\pi}
\int_0^\infty
\left|
\frac{1-e^{- i y^2}}{y^2}-\I
\right|^2
\frac{y^2}{(t+y^2)^2}\,dy.
\end{aligned}
\end{equation}
The first estimate in \cref{eq:local_lower_bound_profile_bounds} gives an integrable bound independent of $t$: the integrand is bounded by
$Cy^2$ for $0<y\leq1$ and by $Cy^{-2}$ for $y\geq1$. Therefore, dominated convergence and \cref{eq:first_order_profile_integral} imply
\begin{equation}
\lim_{t\to0^+}
t^{-5/2}
\norm{
- i tV\psi_0
+
 i \int_0^t e^{-\I(t-s)A}V\psi_0\,ds
}^2
=
\frac{16c^2\sqrt{2}}{5\sqrt{\pi}}.
\end{equation}
Taking the square root and combining this limit with \cref{eq:first_order_local_error_decomposition,eq:first_order_local_error_remainder} proves \cref{eq:first_order_local_error_asymptotics}.

We next consider the second-order Trotter formula. Expanding the middle multiplication operator in \cref{eq: def e2t} and using \cref{eq:local_lower_bound_duhamel}, we have
\begin{equation}
\label{eq:second_order_local_error_decomposition}
\begin{aligned}
e_2(t)\psi_0
=  &
- i t e^{- i tA/2}V\psi_0
+
 i \int_0^t e^{-\I(t-s)A}V\psi_0\,ds \\
&-
 i t e^{- i tA/2}V(e^{- i tA/2}-I)\psi_0 \\
&+
e^{- i tA/2}(e^{- i tV}-I+ i tV)e^{- i tA/2}\psi_0 \\
&+
 i \int_0^t e^{-\I(t-s)A}V(e^{- i sH}-I)\psi_0\,ds.
\end{aligned}
\end{equation}
The last three terms are again bounded by
\begin{equation}
\label{eq:second_order_local_error_remainder}
C t^{3/2}\norm{\psi_0}_{H^2}
=
o(t^{5/4})
\end{equation}
by \cref{lem:local_lower_bound_remainder_estimates}. The Fourier transform of the first line of \cref{eq:second_order_local_error_decomposition} is
\begin{equation}
t
\left(
\frac{1-e^{- i t|\xi|^2}}{t|\xi|^2}
-
 i e^{- i t|\xi|^2/2}
\right)
\widehat{V\psi_0}(\xi).
\end{equation}
Proceeding as in \cref{eq:first_order_local_error_scaled_limit}, we obtain
\begin{equation}
\begin{aligned}
&t^{-5/2}
\norm{
- i t e^{- i tA/2}V\psi_0
+
 i \int_0^t e^{-\I(t-s)A}V\psi_0\,ds
}^2 \\
&\qquad=
\frac{8c^2}{\pi}
\int_0^\infty
\left|
\frac{1-e^{- i y^2}}{y^2}
-
 i e^{- i y^2/2}
\right|^2
\frac{y^2}{(t+y^2)^2}\,dy.
\end{aligned}
\end{equation}
The second bound in \cref{eq:local_lower_bound_profile_bounds} gives an integrable majorant. Hence, dominated convergence and \cref{eq:second_order_profile_integral} yield
\begin{equation}
\lim_{t\to0^+}
t^{-5/2}
\norm{
- i t e^{- i tA/2}V\psi_0
+
 i \int_0^t e^{-\I(t-s)A}V\psi_0\,ds
}^2
=
\frac{16c^2}{15\sqrt{\pi}}
\left(5-2\sqrt{2}\right).
\end{equation}
Combining this limit with \cref{eq:second_order_local_error_decomposition,eq:second_order_local_error_remainder} proves \cref{eq:second_order_local_error_asymptotics} and completes the proof.
\end{proof}

\begin{cor}[Local Error Lower Bounds]
\label{cor:local_error_lower_bounds}
Under the conditions of \cref{thm:sharp_local_error_asymptotics}, there exist $C>0$ and $t_0>0$ such that
\begin{equation}
\norm{ E (t)\psi_0}
\geq
Ct^{5/4},
\qquad
\norm{e_2(t)\psi_0}
\geq
Ct^{5/4}
\end{equation}
for every $0<t\leq t_0$. Consequently, no estimate of the form
\begin{equation}
\sup_{\norm{v}_{H^2}\leq1}
\norm{E(t)v}
\leq
Ct^p,
\quad
\sup_{\norm{v}_{H^2}\leq1}
\norm{e_2(t)v}
\leq
Ct^p,
\qquad
p>\frac54,
\end{equation}
can hold for all sufficiently small $t$.
\end{cor}

\begin{proof}
The first conclusion follows from the two asymptotic formulas in \cref{thm:sharp_local_error_asymptotics}. The second follows by applying the assumed estimate to $\psi_0/\norm{\psi_0}_{H^2}$.
\end{proof}

\begin{rem}
For the attractive potential $V(x)=-2/|x|$, the initial state in \cref{eq:lower_bound_initial_state} is the normalized ground state of $H=-\Delta-2/|x|$ and satisfies $H\psi_0=-\psi_0$. Thus, the lower bounds above already hold for a physically natural Coulomb state.
\end{rem}

\begin{rem}
The local lower bound also rules out a global estimate of the form
\begin{equation}
\label{eq:impossible_uniform_global_rate}
\norm{E_{j,L}(t)\psi_0}
\leq
C T t^\alpha,
\qquad
T=Lt,
\qquad
\alpha>\frac14,
\end{equation}
that holds uniformly for all $L\geq1$ and all sufficiently small $t$, where $j=1,2$. Indeed, taking $L=1$ gives $T=t$, so that the right-hand
side of \cref{eq:impossible_uniform_global_rate} is $Ct^{1+\alpha}=o(t^{5/4})$, contradicting \cref{cor:local_error_lower_bounds}.

This conclusion should be distinguished from a fixed-time global lower bound in the many-step regime. For a fixed $T>0$ and $L=T/t\to\infty$, the local errors generated at different time steps may cancel or accumulate sublinearly in $L$. Therefore, the result does not by itself establish a lower bound of order $t^{1/4}$ for the accumulated error at a fixed final time with $L \to \infty$. Such a result would require a separate analysis of the propagation and accumulation of the local errors.
\end{rem}

\section{Conclusion and Discussions}\label{sec:conclusion}

In this work, we developed a sequence of rigorous analyses of Trotter error for many-body quantum systems with Coulomb interactions. The primary mathematical challenges arise from both the many-body nature of the problem and the singular, long-ranged structure of the Coulomb interaction itself.

Our first main result establishes that the second-order Trotter formula achieves a convergence rate of $1/4$, together with an explicit polynomial dependence of the error prefactor on the system size, for general initial states in the domain of the Hamiltonian. To the best of our knowledge, this $1/4$ rate for general initial states is new even in the one-body setting.
We further establish exact one-step local error asymptotics of order $t^{5/4}$, and hence matching local error lower bounds, for both the first-order and second-order Trotter formulas on the normalized ground state of the attractive one-body Coulomb Hamiltonian. These asymptotics show that the worst-case local exponent is sharp and rule out a global convergence exponent larger than $1/4$ uniformly over all numbers of Trotter steps $L\geq1$. They do not, however, establish a fixed-time global lower bound in the many-step regime, as the number of steps tends to infinity. Our result shows that the rate degradation is not a phenomenon specific to the first-order Trotter formula, but persists for the second-order Trotter formula as well. This indicates that passing from the first-order to the second-order formula alone cannot resolve the fundamental loss of convergence rate induced by the Coulomb singularity.

Our results lead us to conjecture that further increasing the Trotter order does not overcome the quarter-order bottleneck for general initial states, since the cancellations underlying higher-order Suzuki formulas require additional regularity that is not available in general for Coulomb Hamiltonians; the increasingly involved ordering of the singular factors may also make the analysis more delicate.

Our second main result shows that this worst-case limitation is not universal. Under the angular-momentum and Sobolev regularity conditions in \cref{asp:angular_sobolev_condition}, we obtain a hierarchy of improved convergence rates. In the one-body case, the full first-order rate is recovered when $\ell\geq2$, while the full second-order rate is recovered when $\ell\geq4$; the intermediate rates are given in \cref{eq:def_gamma_1,eq:def_gamma_2}. These conditions are related to the behavior of the wave function near the Coulomb singularity and, for hydrogen eigenstates, to the orbital angular momentum. Importantly, our analysis is not restricted to eigenstates and applies to general initial states satisfying \cref{asp:angular_sobolev_condition}.

Our third main result establishes improved convergence rates for permutation-invariant many-body Coulomb Hamiltonians with fermionic initial states. The main observation is that fermionic antisymmetry removes the zeroth angular momentum component in every relative coordinate. For initial data in $H^{2+\delta}(\mathbb{R}^{3N})$, where $0\leq\delta<\frac12$, both the first-order and second-order Trotter formulas converge with rate $(1+\delta)/2$, with explicit polynomial dependence on the particle number. In particular, without imposing any additional regularity beyond $H^2$, both formulas achieve a convergence rate of $1/2$, rather than the general $1/4$ rate.

Taken together, our results reveal a rather complete picture for many-body Coulomb interactions: while Coulomb singularities impose a fundamental bottleneck in the worst case, there still exist physically relevant states that can significantly outperform this limit. This underscores the importance of incorporating structural information about the quantum state into complexity analysis, rather than relying solely on worst-case general bounds.

From a mathematical perspective, we also identify a \REV{set of} Sobolev regularity feature\REV{s} of Coulomb systems (see~\cref{prop:Pgel-propagation}, \cref{prop:Pgel-weighted} and \cref{prop:Nbody_Pge1}), which may be of independent interest beyond quantum simulation.

Our strong $L^2$ error estimates also imply expectation-value error bounds with the same convergence rates for bounded observables. It remains an interesting question whether smooth or localized observables admit higher observable-dependent convergence rates through additional cancellations that are not visible in the strong error. For unbounded observables, such as position or momentum, additional domain and propagation assumptions are needed. Investigating such cases is an interesting direction for future work.

A natural question is how these continuum-limit results relate to the finite spatial discretizations used in practice. First, as the discretization size increases, the discrete system must recover the continuum behavior; otherwise, it would indicate an inconsistency in the discretization scheme. Second, even at finite discretization, numerical results~\cite[Figures 1 and 6]{BurgarthFacchiHahnJohnssonYuasa2024} observe the $1/4$ convergence rate. More specifically, the observed convergence behavior exhibits an effective slope that decreases as the number of spatial basis functions increases, approaching the $1/4$ rate. This can be interpreted as a crossover phenomenon: while higher-order convergence may be visible with few spatial modes, the regime in which such behavior appears shrinks as the basis size increases. Moreover, this crossover is expected to occur more rapidly as the particle number $N$ grows. Notably, \cite{BurgarthGalkeHahnvanLuijk2023} gives explicit sufficient conditions under which discretization schemes are compatible with Trotterization for unbounded Hamiltonians.

A closely related open problem is to rigorously quantify spatial discretization error, including the number of basis functions required as a function of system size and target accuracy. This direction is promising in light of our technical results, which provide control of time evolution under unbounded operators together with system-size-dependent Sobolev norm estimates. We are actively investigating this problem.

Several directions remain for future investigation.
First, our previous work shows that for sufficiently smooth potentials (e.g., $V \in C^4$~\cite{JahnkeLubich2000,AnFangLin2021}), Trotter formulas recover their nominal convergence rates (first-order remains first-order, second-order remains second-order) for initial conditions with good regularity, as in the bounded-operator setting. In contrast, for Coulomb interactions, the general-data error bounds for both the first-order and second-order Trotter formulas have convergence exponent $1/4$. This raises a natural question: does such a rate degradation occur for all singular potentials?

Our ongoing work indicates that the long-range nature of the Coulomb interaction is not necessary for such a rate degradation. For example, we find that Coulomb-Yukawa-type potentials, which retain Coulomb singularities at short distances but exhibit decay at long range, display a similar general-case convergence rate. It remains an open question whether the improved convergence rates established here for many-body fermionic Coulomb systems extend to Coulomb-Yukawa interactions or to more general singular pair potentials. Although the same antisymmetry mechanism remains present, a complete analysis of the corresponding many-body weighted estimates and Sobolev propagation is still needed. This highlights an important conceptual message: while bounded operators exhibit broadly uniform behavior in such analyses, unbounded operators must be treated on a case-by-case basis, with their specific structural properties playing a decisive role.

Our central message is that, unlike bounded operators, unbounded operators do not admit a uniform theory (even at the level of convergence rates) and must be analyzed in a problem-specific manner. Nevertheless, our work provides a framework for rigorously formulating and analyzing quantum simulation in the presence of unbounded operators, and lays the ground for systematically studying a wider class of problems. More broadly, our results show that unboundedness does not preclude rigorous convergence, but can fundamentally alter both the rate and structure of approximation. This highlights the essential role of mathematical tools from PDEs and functional analysis in understanding the capabilities and limitations of quantum simulation algorithms.

\section*{Acknowledgements}
The authors thank Garnet Chan, Lin Lin, John Preskill, Avy Soffer for their valuable comments during the preparation stage of the manuscript. They also thank Alexander Hahn and Daniel Burgarth for the helpful comments on the writing at a later stage. D.F. acknowledges the support from the U.S. Department of Energy, Office of Science, Accelerated Research in Quantum Computing Centers, Quantum Utility through Advanced Computational Quantum Algorithms, grant no. DE-SC0025572, and National Science Foundation via the NSF CAREER award DMS-2438074. X.W. acknowledges the support from Australian Laureate Fellowships, grant FL220100072. 

\noindent
\textbf{Data Availability}. Data sharing is not applicable to this article, as no data sets were generated or analyzed
during the current study.

\noindent \textbf{Conflict of interest}. There is no conflict of interest.

\appendix

\section{Proofs of Auxiliary Estimates}\label{sec: tool}
In this section, we present proofs of auxiliary estimates used to prove the main results. Each of them may be of independent interest.

\subsection{Proof of~\cref{prop: S2H2}}\label{sec:spherical_hardy_pf}

\begin{proof}[Proof of~\cref{prop: S2H2}]
We first argue for $f \in C_c^\infty(\mathbb{R}^3)$ and then extend to $H^2(\mathbb{R}^3)$ by density. Recall the standard representation
\begin{equation}\label{eq:polar-lap}
    \Delta = \partial_r^2 + \frac{2}{r}\partial_r + \frac{1}{r^2}\Delta_{S^2}.
\end{equation}
By the chain rule and the spherical representation
\begin{equation}
    \begin{cases}
        x_1 = r\cos\theta,\\
        x_2 = r\sin\theta\cos\varphi,\\
        x_3 = r\sin\theta\sin\varphi,
    \end{cases}
    \qquad
    \theta \in [-\tfrac{\pi}{2}, \tfrac{\pi}{2}], \;
    \varphi \in [0, 2\pi),
\end{equation}
we obtain
\begin{equation}
    \partial_r
    = \cos\theta\,\partial_{x_1}
      + \sin\theta\cos\varphi\,\partial_{x_2}
      + \sin\theta\sin\varphi\,\partial_{x_3}.
\end{equation}
This yields
\begin{equation}
    \|\partial_r^2\|_{H^2 \to L^2} \leq 9,
\end{equation}
and, together with the Hardy-type inequality 
\begin{equation}\label{eq:HLS_one-body}
     \Bigl\|\tfrac{1}{r}\tfrac{1}{|p|}\Bigr\|_{L^2 \to L^2} = 2.
\end{equation}
See, e.g., \cite[Theorem~2.5]{Herbst1977} (see \cite[Equation (43) and (44)]{FangWuSoffer2025}).
\begin{equation}
    \Bigl\|\tfrac{2}{r}\partial_r\Bigr\|_{H^2 \to L^2}
    \leq 2 \sum_{j=1}^3
        \Bigl\|\tfrac{1}{r}\tfrac{1}{|p|}\Bigr\|
        \,\|\,|p|\partial_{x_j}\|_{H^2\to L^2}
    \leq 12.
\end{equation}
Combining these estimates with~\cref{eq:polar-lap} gives~\cref{eq: Sr2}.
\end{proof}

\subsection{Proof of~\cref{lem:lower-order-propagation}}\label{sec:pf_lemma_coulomb_regularity_H_sigma}

\begin{proof}[Proof of~\cref{lem:lower-order-propagation}]
We first prove the endpoint estimate on $H^1(\mathbb{R}^3)$. Hardy's inequality and Cauchy-Schwarz give, for every $u\in H^1(\mathbb R^3)$,
\[
\begin{aligned}
    \int_{\mathbb R^3}\frac{|u(x)|^2}{|x|}\,dx
    &\leq \norm{|x|^{-1}u}\norm{u} \\
    &\leq 2\norm{\nabla u}\norm{u}.
\end{aligned}
\]
Consequently, for every $\varepsilon>0$, Young's inequality yields
\begin{equation}\label{eq:coulomb-form-bound}
    \left|c\int_{\mathbb R^3}\frac{|u(x)|^2}{|x|}\,dx\right|
    \leq \varepsilon\norm{\nabla u}^2
       +\frac{c^2}{\varepsilon}\norm{u}^2.
\end{equation}
Thus the Coulomb potential is infinitesimally form bounded with respect to
$-\Delta$. Choose $0<\varepsilon<1$ and then choose $A>0$ sufficiently large.
For the closed quadratic form of $H+A$,
\[
    q_A[u]
    :=\norm{\nabla u}^2
      +c\int_{\mathbb R^3}\frac{|u(x)|^2}{|x|}\,dx
      +A\norm{u}^2,
\]
\eqref{eq:coulomb-form-bound} gives constants $0<a_c\leq b_c<\infty$ such
that
\begin{equation}\label{eq:form-equivalence}
    a_c\norm{u}_{H^1}^2
    \leq q_A[u]
    \leq b_c\norm{u}_{H^1}^2,
    \qquad u\in H^1(\mathbb R^3).
\end{equation}
In particular, $H+A$ is positive, its form domain is $H^1(\mathbb R^3)$,
and
\[
    q_A[u]=\norm{(H+A)^{1/2}u}^2.
\]

Set $U(t)=e^{- i tH}$. Since $U(t)$ and $(H+A)^{1/2}$ are both Borel
functions of the self-adjoint operator $H$, they commute. The spectral
theorem and the unitarity of $U(t)$ therefore imply
\[
\begin{aligned}
    q_A[U(t)u]
    &=\norm{(H+A)^{1/2}U(t)u}^2 \\
    &=\norm{U(t)(H+A)^{1/2}u}^2
      =q_A[u].
\end{aligned}
\]
Using \eqref{eq:form-equivalence} on both sides gives
\begin{equation}\label{eq:H1-propagation}
    \norm{U(t)u}_{H^1}
    \leq \left(\frac{b_c}{a_c}\right)^{1/2}\norm{u}_{H^1},
    \qquad t\in\mathbb R.
\end{equation}
At the other endpoint, unitarity gives
\begin{equation}\label{eq:L2-propagation}
    \norm{U(t)u}=\norm{u}.
\end{equation}
Finally, complex interpolation between \eqref{eq:L2-propagation} and
\eqref{eq:H1-propagation} gives the intermediate Sobolev estimates. Here
$[X_0,X_1]_\sigma$ denotes the complex interpolation space between $X_0$ and
$X_1$ with interpolation parameter $\sigma$. We use the standard identity
\[
    [L^2(\mathbb R^3),H^1(\mathbb R^3)]_\sigma
    =H^\sigma(\mathbb R^3),
\]
where equality means equivalence of norms. It follows that
\[
    \norm{U(t)u}_{H^\sigma}
    \leq C_{c,\sigma}\norm{u}_{H^\sigma}
\]
uniformly for $t\in\mathbb R$ and $0\leq\sigma\leq1$.
\end{proof}

\section{The Eigenstate Setting}
\label{app:eigenstate_setting}

We record the specialization of \cref{thm: 2-rd Trotter} to eigenstates. In this setting, the weighted estimate needed in the proof is only the single-sector counterpart of \cref{prop:Pgel-weighted}, which follows directly from the partial-wave fractional Hardy inequality. Moreover, the exact solution remains an eigenstate up to a phase factor, while the free evolution is an isometry on the relevant Sobolev space and preserves the angular momentum sector. Consequently, the proof does not require either the angular induction or the nontrivial Sobolev propagation estimates established in the main text.

\begin{prop}[Second-order Trotter Rate for Eigenstates]
\label{cor:second_order_eigenstate_rate}
Let $H=A+B$ be the one-body Coulomb Hamiltonian defined in \cref{1-SE,def:A_and_B}. Suppose that
\begin{equation}
H\psi_0
=
E\psi_0,
\qquad
P_{  \ell }\psi_0
=
\psi_0,
\qquad
\psi_0
\in
H^{1+  \ell +\delta}(\mathbb{R}^3)
\end{equation}
for some $  \ell \in\mathbb{N}^+$ and $0\leq\delta<\frac12$. Let $T>0$, $L\in\mathbb{N}^+$, and $t=T/L\leq1$. Then
\begin{equation}
\label{eq:second_order_eigenstate_rate}
\norm{
\left(
e^{-iHT}
-
\left(
e^{-iAt/2}e^{-iBt}e^{-iAt/2}
\right)^L
\right)
\psi_0
} 
\leq
C_{  \ell ,\delta}\,
Tt^{\gamma_2(  \ell ,\delta)}
\norm{\psi_0}_{H^{1+  \ell +\delta}},
\end{equation}
where $\gamma_2$ is defined in \cref{eq:def_gamma_2}.
\end{prop}

\begin{proof}
Since $P_{  \ell }\psi_0=\psi_0$, we also have $P_{\geq  \ell }\psi_0=\psi_0$. Moreover, the eigenvalue equation gives
\begin{equation}
e^{-isH}\psi_0
=
e^{-isE}\psi_0.
\end{equation}
The free evolution is an isometry on $H^{1+  \ell +\delta}(\mathbb{R}^3)$ and commutes with $P_{  \ell }$. Therefore,
\begin{equation}
\sup_{\tau,s\in\mathbb{R}}
\norm{
e^{-i\tau(-\Delta)}
e^{-isH}\psi_0
}_{H^{1+  \ell +\delta}}
=
\norm{\psi_0}_{H^{1+  \ell +\delta}}.
\end{equation}
Since the free evolution commutes with $P_{  \ell }$, the state $e^{-i\tau(-\Delta)}\psi_0$ remains in the same angular momentum
sector. Applying the single-sector case of \cref{prop:Pgel-weighted} therefore gives
\begin{equation}
\sup_{\tau,s\in\mathbb{R}}
\norm{
\frac{1}{|x|^{1+  \ell +\delta}}
e^{-i\tau(-\Delta)}
e^{-isH}\psi_0
}
\leq
C_{  \ell ,\delta}
\norm{\psi_0}_{H^{1+  \ell +\delta}}.
\end{equation}
Thus, the Sobolev and weighted estimates needed in the proof of \cref{thm: 2-rd Trotter} hold directly in the eigenstate setting.
Applying the remaining local error estimates and summing over the $L$ time steps proves the result.
\end{proof}

 \bibliographystyle{unsrt}
 \bibliography{bib}

@article{Herbst1977,
  title={Spectral theory of the operator (p 2+ m 2) 1/2- Ze 2/r},
  author={Herbst, Ira W},
  journal={Communications in Mathematical Physics},
  volume={53},
  number={3},
  pages={285--294},
  year={1977},
  publisher={Springer}
}

@misc{StroeksLentermanTerhalHerasymenko2024,
    title={Solving Free Fermion Problems on a Quantum Computer},
    author={Maarten Stroeks and Daan Lenterman and Barbara Terhal and Yaroslav Herasymenko},
    year={2024},
    eprint={2409.04550},
    archivePrefix={arXiv},
    primaryClass={quant-ph},
    note={arXiv preprint arXiv:2409.04550}
}

@article{BabbushBerrySandersKivlichanSchererWeiLoveAspuru-Guzik2017,
   title={Exponentially more precise quantum simulation of fermions in the configuration interaction representation},
   volume={3},
   ISSN={2058-9565},
   url={http://dx.doi.org/10.1088/2058-9565/aa9463},
   DOI={10.1088/2058-9565/aa9463},
   number={1},
   journal={Quantum Science and Technology},
   publisher={IOP Publishing},
   author={Babbush, Ryan and Berry, Dominic W and Sanders, Yuval R and Kivlichan, Ian D and Scherer, Artur and Wei, Annie Y and Love, Peter J and Aspuru-Guzik, Alan},
   year={2017},
   month=dec, pages={015006} }

@article{BabbushBerryDominicKivlichanWeiLoveAspuru-Guzik2016,
doi = {10.1088/1367-2630/18/3/033032},
url = {https://dx.doi.org/10.1088/1367-2630/18/3/033032},
year = {2016},
month = {mar},
publisher = {IOP Publishing},
volume = {18},
number = {3},
pages = {033032},
author = {Babbush, Ryan and Berry, Dominic W and Kivlichan, Ian D and Wei, Annie Y and Love, Peter J and Aspuru-Guzik, Alan},
title = {Exponentially more precise quantum simulation of fermions in second quantization},
journal = {New Journal of Physics}
}

@article{AspuruGuzikDutoiLoveHead-Gordon2005,
   title={Simulated Quantum Computation of Molecular Energies},
   volume={309},
   ISSN={1095-9203},
   url={http://dx.doi.org/10.1126/science.1113479},
   DOI={10.1126/science.1113479},
   number={5741},
   journal={Science},
   publisher={American Association for the Advancement of Science (AAAS)},
   author={Aspuru-Guzik, Alan and Dutoi, Anthony D. and Love, Peter J. and Head-Gordon, Martin},
   year={2005},
   month=sep, pages={1704–1707} }

@article{
RubinBerryKononovEtAlLeeNevenBabbushBaczewski2024,
author = {Nicholas C. Rubin  and Dominic W. Berry  and Alina Kononov  and Fionn D. Malone  and Tanuj Khattar  and Alec White  and Joonho Lee  and Hartmut Neven  and Ryan Babbush  and Andrew D. Baczewski },
title = {Quantum computation of stopping power for inertial fusion target design},
journal = {Proceedings of the National Academy of Sciences},
volume = {121},
number = {23},
pages = {e2317772121},
year = {2024},
doi = {10.1073/pnas.2317772121},
URL = {https://www.pnas.org/doi/abs/10.1073/pnas.2317772121},
eprint = {https://www.pnas.org/doi/pdf/10.1073/pnas.2317772121}}

@article{BabbushWiebeMcCleanMcClainNevenChan2018,
  title = {Low-Depth Quantum Simulation of Materials},
  author = {Babbush, Ryan and Wiebe, Nathan and McClean, Jarrod and McClain, James and Neven, Hartmut and Chan, Garnet Kin-Lic},
  journal = {Phys. Rev. X},
  volume = {8},
  issue = {1},
  pages = {011044},
  numpages = {40},
  year = {2018},
  month = {Mar},
  publisher = {American Physical Society},
  doi = {10.1103/PhysRevX.8.011044},
  url = {https://link.aps.org/doi/10.1103/PhysRevX.8.011044}
}

@article{BabbushHugginsBerryUngZhaoEtAlBaczewskiLee2023,
   title={Quantum simulation of exact electron dynamics can be more efficient than classical mean-field methods},
   volume={14},
   ISSN={2041-1723},
   url={http://dx.doi.org/10.1038/s41467-023-39024-0},
   DOI={10.1038/s41467-023-39024-0},
   number={1},
   journal={Nature Communications},
   publisher={Springer Science and Business Media LLC},
   author={Babbush, Ryan and Huggins, William J. and Berry, Dominic W. and Ung, Shu Fay and Zhao, Andrew and Reichman, David R. and Neven, Hartmut and Baczewski, Andrew D. and Lee, Joonho},
   year={2023},
   month=jul }

@article{KassalJordanLoveMohseniAspuru-Guzik2008,
   title={Polynomial-time quantum algorithm for the simulation of chemical dynamics},
   volume={105},
   ISSN={1091-6490},
   url={http://dx.doi.org/10.1073/pnas.0808245105},
   DOI={10.1073/pnas.0808245105},
   number={48},
   journal={Proceedings of the National Academy of Sciences},
   publisher={Proceedings of the National Academy of Sciences},
   author={Kassal, Ivan and Jordan, Stephen P. and Love, Peter J. and Mohseni, Masoud and Aspuru-Guzik, Alán},
   year={2008},
   month=dec, pages={18681–18686} }

@article{KivlichanMcCleanWiebeGidneyAspuru-GuzikChanBabbush2018,
   title={Quantum Simulation of Electronic Structure with Linear Depth and Connectivity},
   volume={120},
   ISSN={1079-7114},
   url={http://dx.doi.org/10.1103/PhysRevLett.120.110501},
   DOI={10.1103/physrevlett.120.110501},
   number={11},
   journal={Physical Review Letters},
   publisher={American Physical Society (APS)},
   author={Kivlichan, Ian D. and McClean, Jarrod and Wiebe, Nathan and Gidney, Craig and Aspuru-Guzik, Alan and Chan, Garnet Kin-Lic and Babbush, Ryan},
   year={2018},
   month=mar }

@misc{KuChenHuHsieh2025,
      title={Optimizing Quantum Chemistry Simulations with a Hybrid Quantization Scheme}, 
      author={Calvin Ku and Yu-Cheng Chen and Alice Hu and Min-Hsiu Hsieh},
      year={2025},
      eprint={2507.04253},
      archivePrefix={arXiv},
      primaryClass={quant-ph},
      url={https://arxiv.org/abs/2507.04253}, 
      note={arXiv preprint arXiv:2507.04253}
}

@misc{ChenXuZhaoYuan2024,
    title={Error Interference in Quantum Simulation},
    author={Boyang Chen and Jue Xu and Qi Zhao and Xiao Yuan},
    year={2024},
    eprint={2411.03255},
    archivePrefix={arXiv},
    primaryClass={quant-ph},
    note={arXiv preprint arXiv:2411.03255}
}

@article{BeckerGalkeSalzmannLuijk2024,
   title={Convergence Rates for the Trotter Splitting for Unbounded Operators},
   ISSN={1615-3383},
   url={http://dx.doi.org/10.1007/s10208-025-09730-w},
   DOI={10.1007/s10208-025-09730-w},
   journal={Foundations of Computational Mathematics},
   publisher={Springer Science and Business Media LLC},
   author={Becker, Simon and Galke, Niklas and van Luijk, Lauritz and Salzmann, Robert},
   year={2025},
   month=sep }

@article{BurgarthFacchiHahnJohnssonYuasa2024,
   title={Strong error bounds for Trotter and strang-splittings and their implications for quantum chemistry},
   volume={6},
   ISSN={2643-1564},
   url={http://dx.doi.org/10.1103/PhysRevResearch.6.043155},
   DOI={10.1103/physrevresearch.6.043155},
   number={4},
   journal={Physical Review Research},
   publisher={American Physical Society (APS)},
   author={Burgarth, Daniel and Facchi, Paolo and Hahn, Alexander and Johnsson, Mattias and Yuasa, Kazuya},
   year={2024},
   month=nov }

@article{FangLiuSarkar2025,
   title={Time-Dependent Hamiltonian Simulation via Magnus Expansion: Algorithm and Superconvergence},
   volume={406},
   ISSN={1432-0916},
   url={http://dx.doi.org/10.1007/s00220-025-05314-5},
   DOI={10.1007/s00220-025-05314-5},
   number={6},
   journal={Communications in Mathematical Physics},
   publisher={Springer Science and Business Media LLC},
   author={Fang, Di and Liu, Diyi and Sarkar, Rahul},
   year={2025},
   month=may }

@article{BornsweilFangZhang2025,
   title={Discrete Superconvergence Analysis for Quantum Magnus Algorithms of Unbounded Hamiltonian Simulation},
   volume={407},
   ISSN={1432-0916},
   url={http://dx.doi.org/10.1007/s00220-025-05531-y},
   DOI={10.1007/s00220-025-05531-y},
   number={2},
   journal={Communications in Mathematical Physics},
   publisher={Springer Science and Business Media LLC},
   author={Borns-Weil, Yonah and Fang, Di and Zhang, Jiaqi},
   year={2026},
   month=jan }

@misc{FangQu2025,
      title={Uniform semiclassical observable error bound of Trotter-Suzuki splitting: a simple algebraic proof}, 
      author={Di Fang and Conrad Qu},
      year={2025},
      eprint={2507.02783},
      archivePrefix={arXiv},
      primaryClass={math.NA},
      url={https://arxiv.org/abs/2507.02783}, 
      note = {arXiv:2507.02783}
}

@article{ZhaoZhouChilds2024,
   title={Entanglement accelerates quantum simulation},
   volume={21},
   ISSN={1745-2481},
   url={http://dx.doi.org/10.1038/s41567-025-02945-2},
   DOI={10.1038/s41567-025-02945-2},
   number={8},
   journal={Nature Physics},
   publisher={Springer Science and Business Media LLC},
   author={Zhao, Qi and Zhou, You and Childs, Andrew M.},
   year={2025},
   month=jul, pages={1338–1345} }

@article{HuangTongFangSu2023,
  title = {Learning Many-Body Hamiltonians with Heisenberg-Limited Scaling},
  author = {Huang, Hsin-Yuan and Tong, Yu and Fang, Di and Su, Yuan},
  journal = {Phys. Rev. Lett.},
  volume = {130},
  issue = {20},
  pages = {200403},
  numpages = {7},
  year = {2023},
  month = {May},
  publisher = {American Physical Society},
  doi = {10.1103/PhysRevLett.130.200403},
  url = {https://link.aps.org/doi/10.1103/PhysRevLett.130.200403}
}

@misc{YuXuZhao2024,
    title={Observable-Driven Speed-ups in Quantum Simulations},
    author={Wenjun Yu and Jue Xu and Qi Zhao},
    year={2024},
    eprint={2407.14497},
    archivePrefix={arXiv},
    primaryClass={quant-ph},
    note={arXiv preprint arXiv:2407.14497}
}

@misc{Somma2015,
      title={Quantum simulations of one dimensional quantum systems}, 
      author={Rolando D. Somma},
      year={2016},
      eprint={1503.06319},
      archivePrefix={arXiv},
      primaryClass={quant-ph},
      url={https://arxiv.org/abs/1503.06319}, 
      note={arXiv preprint arXiv:1503.06319}
}

@article{LowSuTongTran2023,
  title = {Complexity of Implementing Trotter Steps},
  author = {Low, Guang Hao and Su, Yuan and Tong, Yu and Tran, Minh C.},
  journal = {PRX Quantum},
  volume = {4},
  issue = {2},
  pages = {020323},
  numpages = {46},
  year = {2023},
  month = {May},
  publisher = {American Physical Society},
  doi = {10.1103/PRXQuantum.4.020323},
  url = {https://link.aps.org/doi/10.1103/PRXQuantum.4.020323}
}

@article{GongZhouLi2023,
  title={A Theory of Digital Quantum Simulations in the Low-Energy Subspace},
  author={Gong, Weiyuan and Zhou, Shuo and Li, Tongyang},
  journal={arXiv preprint arXiv:2312.08867},
  year={2023}
}

@article{SuHuangCampbell2021,
  doi = {10.22331/q-2021-07-05-495},
  url = {https://doi.org/10.22331/q-2021-07-05-495},
  title = {Nearly tight {T}rotterization of interacting electrons},
  author = {Su, Yuan and Huang, Hsin-Yuan and Campbell, Earl T.},
  journal = {{Quantum}},
  issn = {2521-327X},
  publisher = {{Verein zur F{\"{o}}rderung des Open Access Publizierens in den Quantenwissenschaften}},
  volume = {5},
  pages = {495},
  month = jul,
  year = {2021}
}

@misc{ZengSunJiangZhao2022,
      title={Simple and high-precision Hamiltonian simulation by compensating Trotter error with linear combination of unitary operations}, 
      author={Pei Zeng and Jinzhao Sun and Liang Jiang and Qi Zhao},
      year={2022},
      eprint={2212.04566},
      archivePrefix={arXiv},
      primaryClass={quant-ph},
      note={arXiv preprint arXiv:2212.04566}
}

@article{AnFangLin2021,
  doi = {10.22331/q-2021-05-26-459},
  url = {https://doi.org/10.22331/q-2021-05-26-459},
  title = {Time-dependent unbounded {H}amiltonian simulation with vector norm scaling},
  author = {An, Dong and Fang, Di and Lin, Lin},
  journal = {{Quantum}},
  issn = {2521-327X},
  publisher = {{Verein zur F{\"{o}}rderung des Open Access Publizierens in den Quantenwissenschaften}},
  volume = {5},
  pages = {459},
  month = {may},
  year = {2021}
}

@article{AnFangLin2022,
  title={Time-dependent Hamiltonian simulation of highly oscillatory dynamics and superconvergence for Schr{\"o}dinger equation},
  author={An, Dong and Fang, Di and Lin, Lin},
  journal={Quantum},
  volume={6},
  pages={690},
  year={2022},
  publisher={Verein zur F{\"o}rderung des Open Access Publizierens in den Quantenwissenschaften}
}

@article{BornsWeilFang2022,
   title={Uniform Observable Error Bounds of Trotter Formulae for the Semiclassical Schrödinger Equation},
   volume={23},
   ISSN={1540-3467},
   url={http://dx.doi.org/10.1137/23M1583922},
   DOI={10.1137/23m1583922},
   number={1},
   journal={ {Multiscale Model. Simul.}},
   publisher={Society for Industrial & Applied Mathematics (SIAM)},
   author={Borns-Weil, Yonah and Fang, Di},
   year={2025},
   month=jan, pages={255–277} }

@article {FangTres2023,
    AUTHOR = {Fang, Di and Tres Vilanova, Albert},
     TITLE = {Observable error bounds of the time-splitting scheme for
              quantum-classical molecular dynamics},
   JOURNAL = {SIAM J. Numer. Anal.},
  FJOURNAL = {SIAM Journal on Numerical Analysis},
    VOLUME = {61},
      YEAR = {2023},
    NUMBER = {1},
     PAGES = {26--44},
      ISSN = {0036-1429},
   MRCLASS = {65C05 (35Q41 81Q20 82M37)},
  MRNUMBER = {4537563},
       DOI = {10.1137/21M1462349},
       URL = {https://doi.org/10.1137/21M1462349},
}

@Book{Lubich2008book,
  Title                    = {From quantum to classical molecular dynamics: reduced models and numerical analysis},
  Author                   = {Lubich, C.},
  Publisher                = {European Mathematical Society},
  Year                     = {2008}
}

@book{ReedSimonII1975,
  author    = {Reed, Michael and Simon, Barry},
  title     = {Methods of Modern Mathematical Physics. II: Fourier
               Analysis, Self-Adjointness},
  publisher = {Academic Press},
  year      = {1975}
}

@Article{Trotter1959,
  Title                    = {On the product of semi-groups of operators},
  Author                   = {Trotter, H.F.},
  Journal                  = {Proc. Amer. Math. Soc.},
  Year                     = {1959},
  Pages                    = {545},
  Volume                   = {10}
}

@InProceedings{GilyenSuLowEtAl2019,
  author    = {Gily{\'e}n, Andr{\'a}s and Su, Yuan and Low, Guang Hao and Wiebe, Nathan},
  title     = {Quantum singular value transformation and beyond: exponential improvements for quantum matrix arithmetics},
  booktitle = {Proceedings of the 51st Annual ACM SIGACT Symposium on Theory of Computing},
  year      = {2019},
  pages     = {193--204}, 
  doi = {10.1145/3313276.3316366}, 
}

@Article{LowChuang2017,
  author  = {Guang Hao Low and Isaac L. Chuang},
  title   = {Optimal {H}amiltonian Simulation by Quantum Signal Processing},
  journal = {Phys. Rev. Lett.},
  year    = {2017},
  volume  = {118},
  pages   = {010501}, 
  doi = {10.1103/PhysRevLett.118.010501},
}

@Article{BerryChildsCleveEtAl2015,
  Title                    = {Simulating {Hamiltonian} dynamics with a truncated {Taylor} series},
  Author                   = {Berry, D. W. and Childs, A. M. and Cleve, R. and Kothari, R. and Somma, R. D.},
  Journal                  = {Phys. Rev. Lett.},
  Year                     = {2015},
  Pages                    = {090502},
  Volume                   = {114}, 
  doi = {10.1103/PhysRevLett.114.090502}
}

@Article{LowWiebe2019,
  Title                    = {Hamiltonian Simulation in the Interaction Picture},
  Author                   = {Low, G. H. and Wiebe, N.},
  archiveprefix = {arXiv},
  eprint        = {1805.00675},
  primaryclass  = {quant-ph},
  Year                     = {2019},
  note={arXiv preprint arXiv:1805.00675}
}

@Article{JahnkeLubich2000,
  author     = {Jahnke, Tobias and Lubich, Christian},
  title      = {Error bounds for exponential operator splittings},
  journal    = {BIT},
  year       = {2000},
  volume     = {40},
  number     = {4},
  pages      = {735--744},
  issn       = {0006-3835},
  doi        = {10.1023/A:1022396519656},
  fjournal   = {BIT. Numerical Mathematics},
  mrclass    = {65M15 (47D08)},
  mrnumber   = {1799313},
  mrreviewer = {Hasan Ta\c{s}eli},
}

@Article{ChildsSuTranEtAl2020,
  author        = {Andrew M. Childs and Yuan Su and Minh C. Tran and Nathan Wiebe and Shuchen Zhu},
  title         = {Theory of Trotter Error with Commutator Scaling},
  journal = {Phys. Rev. X},
  volume = {11},
  issue = {1},
  pages = {011020},
  numpages = {49},
  year = {2021},
  publisher = {American Physical Society},
  doi = {10.1103/PhysRevX.11.011020},
}

@Article{BerryChildsSuEtAl2020,
  author  = {Berry, D. W. and Childs, A. M. and Su, Y. and Wang, X. and Wiebe, N.},
  title   = {Time-dependent {H}amiltonian simulation with $L^{1}$-norm scaling},
  journal = {Quantum},
  year    = {2020},
  volume  = {4},
  pages   = {254}, 
  doi = {10.22331/q-2020-04-20-254}, 
}

@Article{DescombesThalhammer2010,
  author  = {Descombes, St{\'e}phane and Thalhammer, Mechthild},
  title   = {{An exact local error representation of exponential operator splitting methods for evolutionary problems and applications to linear Schr{\"o}dinger equations in the semi-classical regime}},
  journal = {BIT Numer. Math.},
  year    = {2010},
  volume  = {50},
  number  = {4},
  pages   = {729--749}, 
  doi = {10.1007/s10543-010-0282-4},
}

@Article{SahinogluSomma2020,
   title={Hamiltonian simulation in the low-energy subspace},
   volume={7},
   ISSN={2056-6387},
   url={http://dx.doi.org/10.1038/s41534-021-00451-w},
   DOI={10.1038/s41534-021-00451-w},
   number={1},
   journal={npj Quantum Information},
   publisher={Springer Science and Business Media LLC},
   author={Şahinoğlu, Burak and Somma, Rolando D.},
   year={2021},
   month=jul }

@article{SuBerryWiebeEtAl2021,
  title={Fault-tolerant quantum simulations of chemistry in first quantization},
  author={Su, Yuan and Berry, Dominic W and Wiebe, Nathan and Rubin, Nicholas and Babbush, Ryan},
  journal={PRX Quantum},
  volume={2},
  number={4},
  pages={040332},
  year={2021},
  publisher={APS}
}

@article{KieferovaSchererBerry2019,
   title={Simulating the dynamics of time-dependent Hamiltonians with a truncated Dyson series},
   volume={99},
   ISSN={2469-9934},
   url={http://dx.doi.org/10.1103/PhysRevA.99.042314},
   DOI={10.1103/physreva.99.042314},
   number={4},
   journal={Physical Review A},
   publisher={American Physical Society (APS)},
   author={Kieferová, Mária and Scherer, Artur and Berry, Dominic W.},
   year={2019},
   month={Apr}
}

@article {LasserLubich2020,
    AUTHOR = {Lasser, Caroline and Lubich, Christian},
     TITLE = {Computing quantum dynamics in the semiclassical regime},
   JOURNAL = {Acta Numer.},
  FJOURNAL = {Acta Numerica},
    VOLUME = {29},
      YEAR = {2020},
     PAGES = {229--401},
      ISSN = {0962-4929},
   MRCLASS = {81-08 (65Mxx 81Q20 81V35)},
  MRNUMBER = {4189293},
       DOI = {10.1017/s0962492920000033},
       URL = {https://doiorg.libproxy.berkeley.edu/10.1017/s0962492920000033},
}

@article{ZhaoZhouShawEtAk2021,
   title={Hamiltonian Simulation with Random Inputs},
   volume={129},
   ISSN={1079-7114},
   url={http://dx.doi.org/10.1103/PhysRevLett.129.270502},
   DOI={10.1103/physrevlett.129.270502},
   number={27},
   journal={Physical Review Letters},
   publisher={American Physical Society (APS)},
   author={Zhao, Qi and Zhou, You and Shaw, Alexander F. and Li, Tongyang and Childs, Andrew M.},
   year={2022},
   month=dec }

@article{ChildsLengEtAl2022,
   title={Quantum simulation of real-space dynamics},
   volume={6},
   ISSN={2521-327X},
   url={http://dx.doi.org/10.22331/q-2022-11-17-860},
   DOI={10.22331/q-2022-11-17-860},
   journal={Quantum},
   publisher={Verein zur Forderung des Open Access Publizierens in den Quantenwissenschaften},
   author={Childs, Andrew M. and Leng, Jiaqi and Li, Tongyang and Liu, Jin-Peng and Zhang, Chenyi},
   year={2022},
   month=nov, pages={860} }

@article{FangWuSoffer2025,
   title={On the Trotter Error in Many-body Quantum Dynamics with Coulomb Potentials},
   volume={407},
   ISSN={1432-0916},
   url={http://dx.doi.org/10.1007/s00220-026-05593-6},
   DOI={10.1007/s00220-026-05593-6},
   number={4},
   journal={Communications in Mathematical Physics},
   publisher={Springer Science and Business Media LLC},
   author={Fang, Di and Wu, Xiaoxu and Soffer, Avy},
   year={2026},
   month=mar }

@misc{FangZhang2025,
      title={Superconvergence of High-order Magnus Quantum Algorithms}, 
      author={Di Fang and Jiaqi Zhang},
      year={2025},
      eprint={2509.22897},
      archivePrefix={arXiv},
      primaryClass={math.NA},
      url={https://arxiv.org/abs/2509.22897}, 
      note={arXiv preprint arXiv:2509.22897}
}

@misc{FangLiuZhu2025,
      title={High-order Magnus Expansion for Hamiltonian Simulation}, 
      author={Di Fang and Diyi Liu and Shuchen Zhu},
      year={2025},
      eprint={2509.06054},
      archivePrefix={arXiv},
      primaryClass={quant-ph},
      url={https://arxiv.org/abs/2509.06054}, 
      note={arXiv preprint arXiv:2509.06054}
}

@article{Strang1968,
author = {Strang, Gilbert},
title = {On the Construction and Comparison of Difference Schemes},
journal = {SIAM Journal on Numerical Analysis},
volume = {5},
number = {3},
pages = {506-517},
year = {1968},
doi = {10.1137/0705041},

URL = { 
    
        https://doi.org/10.1137/0705041
    
    

},
eprint = { 
    
        https://doi.org/10.1137/0705041
    
    

}

}

@book{BlanesCasas2017book,
  title={A concise introduction to geometric numerical integration},
  author={Blanes, Sergio and Casas, Fernando},
  year={2017},
  publisher={CRC press}
}

@article{BurgarthGalkeHahnvanLuijk2023,
   title={State-dependent Trotter limits and their approximations},
   volume={107},
   ISSN={2469-9934},
   url={http://dx.doi.org/10.1103/PhysRevA.107.L040201},
   DOI={10.1103/physreva.107.l040201},
   number={4},
   journal={Physical Review A},
   publisher={American Physical Society (APS)},
   author={Burgarth, Daniel and Galke, Niklas and Hahn, Alexander and van Luijk, Lauritz},
   year={2023},
   month=apr }

@misc{LiuZhuLowLinYang2025,
      title={Block encoding with low gate count for second-quantized Hamiltonians}, 
      author={Diyi Liu and Shuchen Zhu and Guang Hao Low and Lin Lin and Chao Yang},
      year={2025},
      eprint={2510.08644},
      archivePrefix={arXiv},
      primaryClass={quant-ph},
      url={https://arxiv.org/abs/2510.08644}, 
}

@misc{MizutaIkeda2025Fujii2025,
      title={Explicit error bounds with commutator scaling for time-dependent product and multi-product formulas}, 
      author={Kaoru Mizuta and Tatsuhiko N. Ikeda and Keisuke Fujii},
      year={2025},
      eprint={2410.14243},
      archivePrefix={arXiv},
      primaryClass={quant-ph},
      url={https://arxiv.org/abs/2410.14243}, 
}

@article{MizutaKuwahara2025,
  title = {Trotterization is Substantially Efficient for Low-Energy States},
  author = {Mizuta, Kaoru and Kuwahara, Tomotaka},
  journal = {Phys. Rev. Lett.},
  volume = {135},
  issue = {13},
  pages = {130602},
  numpages = {6},
  year = {2025},
  month = {Sep},
  publisher = {American Physical Society},
  doi = {10.1103/q87n-5xhz},
  url = {https://link.aps.org/doi/10.1103/q87n-5xhz}
}
\end{document}